\documentclass[12pt]{article}
\usepackage{graphicx}
\usepackage{pgfplots}
\usepackage{lineno}
\usepackage[colorlinks=true,citecolor=blue]{hyperref}
\usepackage{natbib}
\usepackage{graphicx}
\usepackage{amsfonts}
\usepackage{amsmath}
\usepackage{amssymb}
\usepackage{url}
\usepackage{fancyhdr}
\usepackage{indentfirst}
\usepackage{enumerate}
\usepackage{titlesec}
\usepackage{amsthm}
\usepackage{dsfont}
\usepackage[misc]{ifsym}

\usepackage{subcaption}

\theoremstyle{definition}

\newtheorem{assump}{Assumption}  %%% added by Liyuan
\newenvironment{myassump}[2][]
  {\renewcommand\theassump{#2}\begin{assump}[#1]} {\end{assump}}

\newtheorem{example}{Example}

\newtheorem{definition}{Definition}

\theoremstyle{plain}
\newtheorem{theorem}{Theorem}
\newtheorem{lemma}{Lemma}
\newtheorem{proposition}{Proposition}

\theoremstyle{remark}

\usepackage{cases}

\def\laweq{\buildrel \mathrm{d} \over =}

\theoremstyle{definition}

\def\P{\mathbb{P}}
\def\p{\mathbb{P}}
\def\E{\mathbb{E}}

\def\R{\mathbb{R}}

\def\X{\mathcal{X}}

\def\X{\mathcal{X}}

\def\d{\,\mathrm{d}}
\DeclareMathOperator*{\esssup}{ess\text{-}sup}
\DeclareMathOperator*{\essinf}{ess\text{-}inf}

\newcommand{\blambda}{{\boldsymbol{\lambda}}}

\renewcommand{\epsilon}{\varepsilon}

 \newcommand{\mo}{\mapsto}

\usepackage[onehalfspacing]{setspace}

\usepackage{bm}
\usepackage{tikz-qtree}
\usepackage{tikz}

\def\id{\mathds{1}}

\title{Optimal risk sharing, equilibria, and welfare with empirically realistic risk attitudes
}
\author{Jean-Gabriel Lauzier\thanks{Dept.\ of Economics, Memorial University of Newfoundland,
Canada. \Letter~{\url{jlauzier@mun.ca}}}
\and Liyuan Lin\thanks{Dept.\ of Econometrics \& Business Statistics, Monash University,
Australia. \Letter~{\url{liyuan.lin@monash.edu}}}
\and Peter Wakker\thanks{Dept.\ of Economics, Erasmus University Rotterdam,
the Netherlands. \Letter~{\url{wakker@ese.eur.nl}}}
 \and Ruodu Wang\thanks{Dept.\ of Statistics and Actuarial Science, University of Waterloo,
Canada. \Letter~{\url{wang@uwaterloo.ca}}}}
\date{July 2026}

\begin{document}
\maketitle
\begin{abstract}
This paper examines optimal risk sharing. It brings in empirical realism, reckoning with the risk seeking found empirically. We provide results on Pareto optimality, competitive equilibria, utility frontiers, and the first and second theorems of welfare. Empirical studies have found prevailing risk seeking in several subdomains. Thus, as a first step to increase empirical realism, we allow for some
 risk-seeking agents, still assuming expected utility. Yet more empirical realism is obtained by generalizing expected utility and allowing agents' attitudes to combine risk aversion in some domains with risk seeking in others. We provide results and show directions for future research. \vspace{1ex}
 % Previous abstract:
% This paper examines optimal risk sharing.
 % Its novelty is to bring in empirical realism, in particular reckoning with the
 %  risk seeking that is found empirically. We provide results on Pareto optimality,
% competitive equilibria, utility frontiers, and the first and second theorems of welfare.
% Contrary to common theoretical assumptions, empirical studies find prevailing risk
 % seeking in several subdomains. As a first step to increase empirical realism, we allow
 %  for some
 % risk-seeking agents, still assuming expected utility. Yet more empirical realism is 
 % obtained by generalizing expected utility and allowing agents' attitudes to combine
 %  risk aversion in some domains with risk seeking in others. We then provide results
 %  under these empirically realistic assumptions, showing directions for future research.
 %  Our main new tool is a
 %  counter-monotonic improvement theorem.\vspace{1ex}

\noindent {\it Keywords}: Risk sharing,
 Counter-monotonicity, Competitive equilibrium, Risk seeking,
 Rank-dependent utility\vspace{1ex}

\noindent {\it JEL codes}: D91, D81, D60, C60
\end{abstract}

 %\noindent\rule{\textwidth}{0.5pt}
\newpage
\section{Introduction}

\noindent The current literature on optimal risk sharing almost invariably assumes the normative expected utility, and the common universal risk aversion (\citealp{MWG95}, Chapter 10). However, many behavioral studies have argued for relaxing expected utility for empirical and often even for normative reasons. There have now been a few studies on risk sharing that assumed nonexpected utility, however, as yet from a normative perspective with still universal risk aversion assumed. Modern empirical studies have found much risk seeking though, so much that it deserves systematic study. For instance, in the loss domain, relevant in
 cost-sharing problems and in times of economic crises, risk seeking is even prevailing rather than risk aversion.\footnote{See
 %%%%%%%%%%%%%%%%
\citeauthor{{K98}}'s (\citeyear{{K98}})
 meta-analysis of 136 studies,
\citeauthor{{E96}}'s (\citeyear{{E96}}) review in finance,
\citeauthor{{LV19}}'s (\citeyear{{LV19}}) analysis of a
 world-wide representative student sample,
\cite{{MP97}} for competitive equilibria,
\cite{{ABK13}} for financial professional traders,
\cite{{LPC80}} for managers,
\cite{{O97}} for professional investors, and
\cite{{SZ25}} for  109,658 Chinese subjects.
Risk seeking for losses underlies the disposition effect in finance (\citealp{BH08};
\citealp{HIIW25}). \cite{N09} gave evolutionary arguments for risk seeking for losses.  \cite{CET13} wrote, on classical studies: ``Risk seekers seem to be forgotten.''}

It is obviously important to study optimal risk sharing for empirically realistic risk attitudes, and we introduce this approach. We extend classical results on Pareto optimality, competitive equilibria, and utility frontiers, and the first and second  welfare theorems. Our main tool is a new
counter-monotonic improvement theorem (Theorem \ref{thm:CT_improvement}).

Our first,
 %# rest of para made firmer, and removed end
intermediate step is to relax universal risk aversion, still within the classical expected utility framework. Even if the majority of agents are risk averse, as for gains, then it still is not realistic to assume so for all agents. We thus allow for some agents to be risk seeking, so that they have convex utility functions. Besides increasing empirical realism, these results are also of theoretical interest. To achieve full empirical realism, we allow individual agents to be neither universally risk averse nor universally risk seeking, but have a combination of both, in agreement with the prevailing empirical findings. The latter requires abandoning expected utility, so that we extend our analysis to the currently most popular nonexpected utility model for risk, \citeauthor{{Q82}}'s (\citeyear{{Q82}})
 rank-dependent utility, agreeing with \citeauthor{{TK92}}'s (\citeyear{{TK92}}) prospect theory. Generalizing to nonexpected utility is, of course, important in its own right.
 %[##
 % We provide results for a number of situations but leave other questions for future research.
 % The main restriction in our analysis of nonexpected utility is that we there assume
 % homogenous agents, leaving extensions to hetereogenous agents to future works. The main
 % purpose of this paper is to open this new direction of research.
%##}

\cite{{M90}}, the first to show that risk aversion is equivalent to concave utility under expected utility, pleaded for universal risk aversion, dismissing risk seekers as follows:
   %%%%%%%%%%%%%%%%%%%%%
\begin{center}
\smallskip
\it since experience shows that they are likely to engender a restless, feverish character, unsuited for steady work as well as for the higher and more solid pleasures of life.
\end{center}\smallskip
   %%%%%%%%%%%%%%%%%%%%%
 The common classical belief holds that risk seeking is yet more implausible in stable markets. After all, as soon as there are two or more risk seekers, they can engage in mutual risky
 zero-sum games, using any randomization device. It will end only if some boundary restriction is reached, such as bankruptcy.
  %# rest para rewritten
Although
 % JPE
 \cite{F53} predicted such behavior and Section \ref{sec:EU_gen} gives
 real-world examples where this actually happened, fundamental laws of economics have been established under the assumption of universal risk aversion (\citealp{AD54}), avoiding the phenomenon. This common classical assumption has guided economics for centuries and the supposed fate of risk seekers was rarely considered.

Since the 1980s, economics has become empirically oriented, and it then became widely understood that there is much risk seeking. \cite{B12} and \cite{ES15} provided formal results on extreme fates of risk seekers, where repeated individual decisions can lead to unstoppable gambling unless bankruptcy, assuming risk seeking as found in prospect theory. However, implementations of nonexpected utility theories in such dynamic settings are controversial (\citealp{M89}). We therefore focus on static decisions. Further, we will study risk sharing rather than individual decisions.

\cite{ABCN17} provided the first formal statement about risk seekers in optimal risk sharing. They considered the special case of infinite sequences with an extra assumption of ``reasonable strict optimism''. Building on that, \cite{ACGN18} established a competitive equilibrium for markets that have risk averters with total prior endowment large enough to clear the market. \cite{HZ22} provided generalizations to incomplete markets. We provide an exact statement of the classically believed fate of risk seekers in full generality, without any assumption about the remaining market (Theorems \ref{thm:EU} and
 \ref{thm:CE_only}). This and other results in our paper essentially invoke external randomizations, providing a new rationalization for using \citeauthor{{CS83}}'s (\citeyear{{CS83}}) sunspots. \cite{DS25}
   %DS next abbreviated
 allowed general utility functions that may comprise risk seeking but, regarding deviations from expected utility, focused on risk averse deviations through quasiconvexity with respect to
 probability mixing.
 %\footnote{Their
 %%%%%%%%%%%%
 %general model does not restrict other components, and those can be partly
% risk seeking.}
 %%%%%%%%%%%%%%%%%%%%
 Then, in many situations,
 Pareto-efficient allocations only involve lotteries over
 two-dimensional outcomes.

\citeauthor{LM94}'s (\citeyear{LM94}) comonotonic improvement theorem is an important tool for classical risk sharing: optimality is only possible for
 risk-averse agents if their allocations are mutually comonotonic.
Comonotonicity means that the individual risks of the agents are maximally aligned, so that all mutual hedging possibilities have been used up. Our new
 counter-monotonic improvement theorem provides the analogous tool for risk seeking: optimality and stability are only possible for
 risk-seeking agents if their allocations are mutually
 counter-monotonic. Counter-monotonicity means that the individual risks of the agents are minimally aligned, so that all mutual leveraging possibilities have been used up. We use this theorem to derive our formal results.

Counter-comonotonicity, also called
anti-comonotonicity and first studied by \cite{D72},
  appears as an optimal structure in risk sharing under quantile models
(\citealp{ELW18,ELMW20}).   %\citeauthor{{ELW18}}\ \citeyear{{ELW18}};
 %\citeauthor{{ELMW20}}\ \citeyear{{ELMW20}};
 %\citeauthor{{LLW23a}}\ \citeyear{{LLW23a}}).
Quantile models are important in finance
 (``Value-at-Risk''), but do not provide empirically realistic decision models. This paper focuses on the latter.
\cite{LLW23a} obtained a stochastic representation of
 counter-monotonicity, used in
 Section \ref{sec:countermon} to introduce jackpot allocations
(``winner-take-all'') and their duals, scapegoat allocations. These allocations formalize the above boundary restrictions for
risk-seeking agents, to some extent confirming Marshall's pessimistic view and the classical economic views on risk seeking, but showing what remains possible.
For expected utility agents, some of whom are risk seeking,
Theorems \ref{thm:Pareto}--\ref{thm:EU} in Section \ref{sec:EUPO} analyze Pareto optimality
and Theorems \ref{thm:1WT}--\ref{thm:CE} in Section \ref{sec:CEWT}
establish welfare and competitive equilibria under different conditions.

To achieve full empirical realism, another refinement is
  %# next firmer word.
 needed. All decision models discussed so far were normatively oriented and assumed for each agent either entire (for all lotteries) risk aversion or entire risk seeking. However, the prevailing empirical finding is that agents do not have such ``global''
domain-independent risk attitudes. Risk aversion is prevailing for gains of moderate and high probability, but risk seeking is prevailing for
  %# refs
 small-probability gains (\citealp{FE12}; \citealp{TK92}; \citealp{LV19}; \citealp{W10}), with these phenomena reflected for losses. Such probability dependence cannot be accommodated by expected utility, and generalized models are called for.\footnote{\citeauthor{{FS48}}'s (\citeyear{{FS48}})
 %%%%%%%%%%%%
 famous attempt to incorporate partial risk seeking into EU did not work empirically; see \citet[p.\ 227]{M18}.}
 %%%%%%%
 This
probability dependence explains, for instance, the coexistence of gambling and insurance, a paradox for classical EU.
 %{#
 There are many other empirical reasons why expected utility has to be abandoned when seeking empirical realism (\citealp{A53}; \citealp{S00}) %(\citeauthor{{A53}}  \citeyear{{A53}};\citeauthor{{S00}}  \citeyear{{S00}})
.
 %#}
 We will consider the most popular generalization of EU, \citeauthor{{Q82}}'s (\citeyear{{Q82}})
 rank-dependent utility (RDU), which for gains agrees with \citeauthor{{TK92}}'s (\citeyear{{TK92}}) prospect theory.\footnote{The
 %%%%%%%%%%%%
 quantile models mentioned above are special cases of RDU, and they are neither universally risk averse nor universally risk seeking. But they are not empirically realistic for decision making. \cite{BW23} also did not need universal risk aversion or seeking for their
first-order optimality conditions for risk sharing. However, their results give conditions in terms of preference functionals, rather than preferences, and only apply to interior solutions. The new jackpot and scapegoat allocations in this paper are not interior.}
 Unlike with the classical models considered as yet, for behavioral models there is a special role for one outcome, formalized as the reference outcome and scaled as outcome 0. Better, positive outcomes are gains and worse, negative outcomes are losses.
  %#The rest of the para is rewritten.
  
We achieve full empirical realism in the sense that all our assumptions about the
 risk-preference functional are based on majority findings in the empirical literature. We then show that there cannot be very general simple results, but the results depend on subtle interplays of various factors. Section \ref{sec:RDU} provides first results for several decision situations, but the topic is too complex to be entirely resolved in one paper. Thus we will, for simplicity, assume homogeneous preferences and gain-outcomes. There have been thousands of papers on risk sharing,
%\footnote{A 
   %
%Scopus search gave over 1400 published papers with ``risk sharing" in their title.}   %%%%%%%%%%%%%%%%%%%%%%
  but to our best knowledge we are the first to consider risk preferences where all properties assumed are the prevailing ones found in empirical studies.
 Given the importance of risk sharing and empirical realism, we encourage future research in this direction. The main purpose of this paper is to initiate this new direction of research.

Proofs are in Appendices \ref{app:lemmas} and \ref{app:proof}.
Online Appendices \ref{app:CT}--\ref{onl.app.:scapegoats} provide additional results. Numbers such as O.1 refer to the online appendices.

\section{Model setting} \label{sec:preliminaries}

\noindent We consider a
 one-period economy, with uncertainty realized at the end of the period. By $(\Omega,\mathcal F,\p)$ we denote a probability space, with $\mathcal F$ the
$\sigma$-algebra of events, and by $\E$ we denote expectation under $\p$. Let $\X$ be a set of random variables % (\emph{rvs})
on $\Omega$, referred to as
 \emph{payoffs}, which represent random monetary payoffs at the end of the period. We assume that $\X$ is a convex cone, i.e., it is closed under addition and positive scalar multiplication.
For instance, $\X$ may be the space $L^1$ of integrable random variables. Two random variables $X$ and $Y$ are almost surely equal if $\p(X=Y)=1$, and we identify them, omitting ``almost surely'' in equalities unless we want it emphasized.
 Let $\R_+=[0,\infty)$. We assume $n$ agents for some $n>0$ and write $[n]=\{1,\dots,n\}$.\footnote{\cite{A66}
 %#  Footnote was rewritten.
  %%%%%%%%%%%%%%%%%
 derived a competitive equilibrium in a riskless market for possibly nonconvex preferences. Crucial in his analysis is the idealized assumption that there is a continuum of agents, so that the nonconvexity of preference of each single agent does not impact the market. We prefer to avoid such an idealized effect and to investigate real effects of risk seeking.}
 %%%%%%%%%%%%%%%%%%%%%
 Let
\begin{align*}
 \Delta_n=\left\{(\theta_1,\dots,\theta_n)\in  \mathbb R_+^n: \sum_{i=1}^n \theta_i=1  \right\}
 \end{align*}
be the standard simplex in $\R^n$. We write $\Delta_n(v)=v \Delta_n=\{v\boldsymbol \theta:\boldsymbol\theta \in\Delta_n\}$ for $v>0$.
Throughout, for a scalar $z$ and a vector $\mathbf y=(y_1,\dots,y_n)$, we write
$z\mathbf y$ for the vector $(zy_1,\dots,zy_n)$.
Denote by $\mathbf 0$ and $\mathbf 1$ the vectors $(0,\dots,0)$ and $(1,\dots,1)$ in $\R^n$. Thus, $y\mathbf 1=(y,\dots,y)$ and $\mathbf 1/y=(1/y,\dots,1/y)$ for $y>0$.
We use boldface capital letters for (possibly random)
 $n$-dimensional vectors. Throughout, $0/0 = 0$.

Our setting of risk sharing concerns $n$ agents who share a random variable
 $X\in \X$, the \emph{total payoff}, interpreted as the total random wealth to be allocated among the agents.
The set of all \emph{allocations} of $X\in \X$ is
 $$
\mathbb A_n(X)= \left\{(X_1,\dots,X_n)\in \X^n : \sum_{i=1}^n X_i =X\right\}.
 $$
 That is, an allocation of $X$ is a random vector whose components $X_i$ sum to $X$. This means the wealth is completely redistributed among the agents without any transfers outside the group.
Note that the choice of $\X$ is important in the definition of $\mathbb A_n$ as it restricts the possible allocations. An allocation is \emph{nontrivial} if it has at least two
non-zero components.

 For each agent $i$, her preference relation $\succsim_i$ is represented by a preference functional $\mathcal U_i$, that is, $$
X\succsim_i Y ~\iff ~ \mathcal U_i(X)\ge \mathcal U_i(Y).
$$
The value $\mathcal U_i(X)$ is the \emph{utility} of $X$ for agent $i$.
We assume that if $X$ and $Y$ are equally distributed, denoted by $X\laweq Y$, then $\mathcal U_i(X)=\mathcal U_i(Y)$. This means that $\mathcal U_i$ represents a decision model under risk, that all agents agree on the probability measure $\p$.
 % We assume that each $\mathcal U_i$ represents a decision model
 % under risk, and all agents agree on the probability measure
 % $\p$. This means that if  $X$ and $Y$ are equally distributed,
 % denoted by $
 %X\laweq Y$, then  $\mathcal U_i(X)=\mathcal U_i(Y)$.
For instance, $\mathcal U_i$ may be an EU preference functional $
\mathcal U_i:
X\mapsto \E[u_i(X)] $  for some increasing function $u_i:\R\to \R$ (called a \emph{utility function});
such agents are called \emph{EU agents} or \emph{EU maximizers}.
Throughout, the terms ``increasing" and ``decreasing" are in the
 non-strict sense.
We will study Pareto optimality and
 Arrow--Debreu competitive equilibria in risk sharing, explained next.

\textbf{Pareto optimality.}
For two allocations $\mathbf X=(X_1,\dots,X_n) $ and $\mathbf Y=(Y_1,\dots,Y_n) $ in $ \mathbb A_n(X)$,
we say that $\mathbf X$ \emph{dominates}  $\mathbf Y$
 if $\mathcal U_i(X_i)\geq\mathcal U_i(Y_i)$ for all $i$, and the domination is \emph{strict} if at least one of the inequalities is strict.
The allocation $\mathbf X$ is \emph{Pareto optimal}
if it is not strictly dominated by any allocation in $\mathbb{A}_n(X)$.
Pareto optimality is closely connected to the optimization of a linear combination of the utilities.
For $\boldsymbol{\lambda}=(\lambda_1,\dots, \lambda_n)\in \R_+^n\setminus \{\mathbf 0\}$, an allocation $\mathbf X$ is
 \emph{$\blambda$-optimal}  in $\mathbb{A}_n(X)$  if $\sum_{i=1}^n\lambda_i \mathcal U_i(X_i)$ is maximized over $\mathbb A_n(X)$ at $\mathbf X$. Here, the vector $\boldsymbol \lambda$ is called a \emph{Negishi weight vector} (\citealp{N60}).
We use the term sum optimality for the case $\boldsymbol \lambda=\mathbf 1$. It is
 well known and straightforward to check that
 $\blambda$-optimality for $\blambda$ with positive components implies Pareto optimality. The converse holds under some additional conditions (\citealp[Chapter 16]{{MWG95}}). This paper will provide new results of this kind.

\textbf{Competitive equilibria.}  Suppose that each agent has an initial endowment, summarized by the vector $\boldsymbol \xi =(\xi_1,\dots,\xi_n)\in \mathbb A_n(X)$.
Consider the individual optimization problem for agent $i$:
  \begin{equation}
  \label{eq:indi}
  \mbox{maximize}~  \mathcal U_i(X_i) ~~~~ \mbox{over $X_i\in \mathcal X_i$} ~~~~\mbox{subject to} ~\E^Q [ X_i] \le \E^Q[ \xi _i], \end{equation}
where  $Q$ is a probability measure  representing a linear pricing mechanism,
\footnote{It is without loss of generality to assume that  $Q$ is absolutely continuous with respect to $\p$, because $Q$ can be arbitrary on events with $0$ probability under $\p$, which does not affect equilibria.} $\E^Q$ is the expected value under $Q$, and
$\E^Q [ X_i] \le \E^Q[ \xi _i]$ is the budget condition. For each $i$, $\X_i$ is the set of possible choices $X_i$ for agent $i$. In Section \ref{sec:CEWT}, we will choose $\X_i$ to be the set of all random variables $Y$ satisfying $0\le Y\le X$, where $X$ is assumed to be nonnegative.

The tuple $(X_1,\dots,X_n,Q) $
 is a \emph{competitive equilibrium} if  (a) individual optimality holds: $X_i$ solves \eqref{eq:indi} for each $i$;  and  (b) market clearance holds:
  $\sum_{i=1}^n X_i=X$.
 Then $(X_1,\dots,X_n)$ is an \emph{equilibrium allocation},
 and $Q$ is an \emph{equilibrium price}.
We then often do not mention the initial endowments, meaning that \eqref{eq:indi} is solved for some $\boldsymbol \xi$.

\textbf{Individual rationality.} For an initial endowment vector $\boldsymbol\xi \in \mathbb A_n(X)$, an allocation $\mathbf X\in \mathbb A_n(X)$ is \emph{individually rational} if it dominates $\boldsymbol\xi$. Then risk sharing is beneficial for each agent.

An equilibrium allocation is always individually rational because of individual optimality.
Pareto-optimal allocations and equilibrium allocations are intimately connected through the two fundamental theorems of welfare economics. Under certain conditions, the first welfare theorem states that every equilibrium allocation is Pareto optimal, and the second welfare theorem states that every
 Pareto-optimal allocation is an equilibrium allocation for some initial endowments and equilibrium price. Again, this paper will provide new results of this kind.

 \section{Counter-monotonic improvement}
\label{sec:countermon}

\noindent This section introduces our new tools for analyzing risk sharing. They concern mathematical relations between lotteries and make no assumptions about risk attitudes yet. Later sections will then use them to analyze risk attitudes and risk sharing. We assume $\X=L^1$.

\subsection{Convex order, risk aversion, and comonotonicity}

\noindent A random variable $X$ is \emph{smaller than} a random variable $Y$ in \textit{convex order}, denoted $X\le_{\rm{cx}} Y$, if $\E[\phi(X)] \le \E[\phi(Y )]$ for every convex function $\phi: \R \to \R$ provided that both expectations exist (\citealp{{R13}}; \citealp{SS07}).
That is, $X$ is less risky than $Y$ in the sense of \cite{RS70}.
If $X\le_{\rm{cx}} Y$, then $\E[X]=\E[Y]$, meaning that the relation of convex order compares random variables with the same mean; hence, $Y$ is also called a
 mean-preserving spread of $X$.
Preference functional $\mathcal U$ is \emph{(strongly) risk averse} (\citealp{RS70}) if
 $$X\le_{\rm cx} Y ~\Longrightarrow~ \mathcal  U(X)\ge \mathcal
 U(Y).$$
We usually omit ``strongly'' because we do not consider weak versions.
\emph{Strict risk aversion} holds if  $X<_{\rm {cx}} Y$ (meaning $X\le_{\rm cx} Y$ and $Y \not \le_{\rm cx}X$) implies $\mathcal U(X)>\mathcal U(Y)$.
Similarly, $$X\le_{\rm cx} Y ~\Longrightarrow~ \mathcal  U(X)\le \mathcal
 U(Y)$$
defines \emph{risk seeking}, and the strict version is analogous. Under EU, (strict) risk aversion is equivalent to (strictly) concave utility, and (strict) risk seeking is equivalent to (strictly) convex utility.

Two random variables $X,Y$ are \emph{comonotonic} if
  %%%%%%%%%%%%%%%%%%%
 \begin{align*}
    (X(\omega)-X(\omega'))(Y(\omega)-Y(\omega')) \ge 0~~\mbox{for $(\p\times
 \p)$-almost every~}(\omega,\omega')\in \Omega^2.
 \end{align*}
 %%%%%%%%%%%%
An allocation $(X_1,\dots, X_n)\in \mathbb A_n(X)$ is \emph{comonotonic} if every pair of its components is. The random variables $X_1,\dots,X_n$ are comonotonic if and only if there exists a random variable $Z$ such that each $X_i$ is an increasing transformation of $Z$. We can take $Z=\sum_{i=1}^n X_i$
(\citealp{D94}, Proposition 4.5).

The comonotonic improvement theorem (\citealp[Theorem 10.50]{R13})
states that, for every $X\in \X=L^1$ and every $(X_1,\dots, X_n)\in \mathbb A_n(X)$, there exists a comonotonic allocation $(Y_1,\dots, Y_n)\in \mathbb A_n(X)$ such that
  %##
$Y_i \le_{\rm{cx}} X_i$ for every $i$.
 Consequently, when all agents are strictly risk averse,
Pareto-optimal allocations must be comonotonic (\citealp{CDG12}). A similar result holds under a different label: in an exchange economy with aggregate risk, the individual consumption of strictly
 risk-averse EU maximizers is increasing in aggregate wealth.

\subsection{Counter-monotonicity and jackpot allocations}
\label{sec:ct.1}

\noindent When agents are risk averse, comonotonicity, an extreme type of positive dependence, will appear. When, to the contrary, agents are risk seeking, a form of negative dependence will appear, defined next. Two random variables $X,Y$ are
\emph{counter-monotonic} if   $X,-Y$ are comonotonic.
 An allocation $(X_1,\dots, X_n)\in \mathbb A_n(X)$ is
 \emph{counter-monotonic} if every pair of its components is. Unlike comonotonicity, which allows for arbitrary marginal distributions,
 counter-monotonicity in dimension $n\ge3 $ puts strong restrictions on the marginal distributions (\citealp{{D72}}).  Appendix \ref{app:CT} collects technical background on
 counter-monotonicity.

Let $\Pi_n$ be the set of all
 $n$-compositions (ordered partitions) of $\Omega$, that is, $$\Pi_n=\left\{(A_1,\dots,A_n)\in \mathcal F^n: \bigcup_{i\in [n]} A_i=\Omega \mbox{~and~$A_1,\dots,A_n$ are disjoint}\right\}.$$
The indicator function $\id_A$  for an event $A$ is defined by $\id_A(\omega)=1$ if $\omega\in A$ and $\id_{A}(\omega)=0$ otherwise. \citet[Theorem 1]{LLW23a} obtained
 a stochastic representation of
 counter-monotonic random vectors $(X_1,\dots,X_n)$ with at least three
 non-constant components,
 and they have the form $X_i = Y \id_{A_i} +m_i$ for some $m_1,\dots,m_n\in \R$, $(A_1,\dots,A_n)\in \Pi_n$, and $Y\ge 0$ or $Y\le 0$;
 a precise statement is in Proposition \ref{prop:PCT} in Appendix \ref{app:CT}. It formalizes the
 ``winner-take-all" or ``loser-lose-all" structure of
counter-monotonic allocations.
The introduction discussed the special case where in every state
 all-but-one get ruined. It can also happen that
 all-but-one achieve a best outcome.
 We are particularly interested in the special case
\begin{align}\label{eq:jackpot}
X_i=X\id_{A_i}~\mbox{for all $i\in [n]$, with $(A_1,\dots,A_n)\in \Pi_n$},
\end{align}
where either $X\ge 0$ or $X\le0$.
\begin{definition}
An allocation $(X_1,\dots,X_n)$ is a \emph{jackpot allocation}
if \eqref{eq:jackpot} holds
for some $X\ge 0$, and it is
 a \emph{scapegoat allocation} if \eqref{eq:jackpot} holds
for some  $X\le 0$.
\end{definition}
\noindent Thus, the 0 outcome serves as a maximal or minimal outcome, resulting for all but one agent. It does not have to be a status quo or reference point, and, for instance, outcomes denoted as positive numbers may still refer to losses, as we will discuss later. We formulate the main results of this paper for the case of a minimal outcome, leading to jackpot allocations. For brevity, similar results for maximal outcomes with scapegoat allocations, readily following as dual reformulations, are given in Appendix \ref{onl.app.:scapegoats}. As explained above, all
 counter-monotonic random vectors with at least three
non-constant components can be obtained by adding a
deterministic vector to a jackpot or scapegoat allocation.

Figure \ref{fig:1} displays a jackpot allocation and a comonotonic allocation.
In a jackpot allocation, the random vector $(\id_{A_1},\dots,\id_{A_n})$ can be arbitrarily correlated with $X$. For instance, it may be independent of $X$ or may be fully determined by $X$ (as in Figure \ref{fig:1}).

\begin{figure}[t]
\caption{An illustration of a comonotonic allocation $(X/2,X/2)$  of $X$ and a jackpot allocation $(X\id_{A} ,X\id_{A^c})$ of $X$. In this example, $A$ coincides with the event $\{X\ge 1\}$.}
\label{fig:1}
\begin{center}
\begin{subfigure}[b]{0.23\textwidth}
\centering
\begin{tikzpicture}[scale=2.4]
\draw[->] (0,0) --(1.05,0) node[below] {$\omega$};
 \draw[->] (0,-0.2) --(0,0.8) node[above] {\footnotesize payoff};
\draw[very thick, blue] (0,0.5) sin (0.25,0.75) cos (0.5,0.5)  sin (0.75,0.25) cos (1,0.5) ;
 %\draw[dashed] (1,0) -- (1,0.5);
\draw[dotted] (0,0.5) -- (1,0.5);
\node[below] at (-0.06,0.09) {\scriptsize $0$};
\node[below] at (-0.06,0.59) {\scriptsize $1$};
\end{tikzpicture}
\caption{$X(\omega)$}
\end{subfigure}
\begin{subfigure}[b]{0.23\textwidth}
\centering
\begin{tikzpicture}[scale=2.4]
\draw[->] (0,0) --(1.05,0) node[below] {$\omega$};
\draw[->] (0,-0.2) --(0,0.8) node[above] {\footnotesize payoff};
\draw[very thick, blue] (0,0.5/2) sin (0.25,0.75/2) cos (0.5,0.5/2)  sin (0.75,0.25/2) cos (1,0.5/2) ;
\draw[dotted] (0,0.5) -- (1,0.5);
\node[below] at (-0.06,0.09) {\scriptsize $0$};
\node[below] at (-0.06,0.59) {\scriptsize $1$};
\end{tikzpicture}
\caption{$X(\omega)/2$}
\end{subfigure}
\begin{subfigure}[b]{0.23\textwidth}
\centering
\begin{tikzpicture}[scale=2.4]
\draw[->] (0,0) --(1.05,0) node[below] {$\omega$};
\draw[->] (0,-0.2) --(0,0.8) node[above] {\footnotesize payoff};
\draw[very thick, blue] (0,0.5) sin (0.25,0.75) cos (0.5,0.5) ;
\draw[very thick, blue] (0.5,0)--(1,0);
\draw[dashed] (0.5,0.5) -- (0.5,0);
\draw[dotted] (0,0.5) -- (1,0.5);
\node[below] at (-0.06,0.09) {\scriptsize $0$};
\node[below] at (-0.06,0.59) {\scriptsize $1$};
\node[below] at (0.25,0) {$A$};
\node[below] at (0.75,0) {$A^c$};
\end{tikzpicture}
\caption{$(X\id_A)(\omega)$}
\end{subfigure}
\begin{subfigure}[b]{0.23\textwidth}
\centering
\begin{tikzpicture}[scale=2.4]
\draw[->] (0,0) --(1.05,0) node[below] {$\omega$};
\draw[->] (0,-0.2) --(0,0.8) node[above] {\footnotesize payoff};
\draw[very thick, blue] (0,0)--(0.5,0);
\draw[very thick,  blue] (0.5,0.5)  sin (0.75,0.25) cos (1,0.5) ;
\draw[dashed] (0.5,0.5) -- (0.5,0);
\draw[dotted] (0,0.5) -- (1,0.5);
\node[below] at (-0.06,0.09) {\scriptsize $0$};
\node[below] at (-0.06,0.59) {\scriptsize $1$};
\node[below] at (0.25,0) {$A$};
\node[below] at (0.75,0) {$A^c$};
\end{tikzpicture}
\caption{$(X \id_{A^c})(\omega)$}
\end{subfigure}
\end{center}
\end{figure}
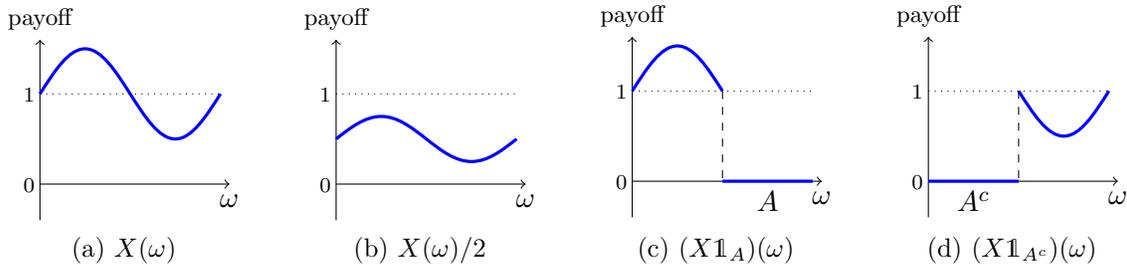

We call $(\id_{A_1},\dots,\id_{A_n})$ in  \eqref{eq:jackpot} a \emph{jackpot vector}, and denote by $\mathbb J_n$ the \emph{set of all jackpot vectors} in $\R^n$, that is,
$$
\mathbb J_n=\{(\id_{A_1},\dots,\id_{A_n}) :(A_1,\dots,A_n)\in \Pi_n\}.
$$
The set $\mathbb J_n$ is precisely the set of all random vectors with a generalized Bernoulli distribution (also known as a multinomial distribution with $1$ trial).
With this,
any jackpot allocation or scapegoat allocation has the form $X\mathbf J$ for some $\mathbf J \in \mathbb J_n$. Both types of allocations are often observed in daily life. For instance, the simple lottery ticket (only one winner) is
a jackpot allocation, and
the ``designated driver" of a party is a scapegoat allocation.\footnote{The
   %%%%%%%%%%%%%%%%%
 ``designated driver" is the randomly selected driver in a party who cannot drink.}
   %%%%%%%%%%%%%%%%%

An equivalent condition for a random vector $(X_1,\dots, X_n)$  to be a jackpot allocation is
\begin{align}\label{eq:jackpot2}
X_i\ge 0~\mbox{and}~ X_iX_j=0 \mbox{ for all $i\ne j$.}
\end{align}
A \emph{probabilistic mixture} of two random vectors with   distributions $F$ and $G$
is a random vector  distributed as
 $\lambda F + (1-\lambda)G $
for some $\lambda \in [0,1]$. Using  \eqref{eq:jackpot2}, we arrive at the following result.
\begin{proposition} \label{pr:3} A probabilistic mixture of two jackpot allocations is a jackpot allocation.\end{proposition}
\noindent %In the context of risk sharing with enough richness,
Proposition \ref{pr:3}  will be used to justify that  the utility possibility set of all jackpot allocations is a convex set for EU agents.
 For two general
 counter-monotonic allocations other than jackpot and scapegoat allocations,
their mixture is not necessarily
 counter-monotonic.

% The next result yields that jackpot allocations are closed under probabilistic mixtures, which directly follows from \eqref{eq:jackpot2}. The same holds for scapegoat allocations by symmetry.

  %%%%%%%%%%%%%%%%%%

\subsection{Counter-monotonic improvement theorem}
\label{sec:ct.improve}

\noindent A simple assumption of external randomization is important in the subsequent analysis.
 %%%%%%%%%%%%%%%
 \begin{myassump}{ER}\label{ass:ER}
There exists a standard uniform (i.e., uniformly distributed on $[0,1]$)  random variable $U$ on $(\Omega,\mathcal F,\p)$  independent of $X$.
\end{myassump}
 %%%%%%%%%%%%%%%%%
\noindent Assumption \ref{ass:ER} implies that the probability space $(\Omega,\mathcal F,\p)$ is \emph{atomless}---that is, there exists a standard uniform random variable defined on this space. Assumption \ref{ass:ER} is easy to satisfy in any practical situation, as an independent uniform random variable can be achieved by  a sequence of dice rolls,   spins of roulette wheels, or sunspots. The assumption
 %#
has been used before to provide divisibility if goods themselves are indivisible
(\citeauthor{{KLP02}},  \citeyear{{KLP02}}). It did not yet receive attention in the risk sharing literature with divisible goods because universal risk aversion was commonly assumed there and then no agent is interested in adding further randomness. As soon as there are two or more
 risk-seeking agents, though, they will want to involve randomization devices and, thus, we provide a new rational for using sunspots.
 % Although the assumption may seem to be restrictive from
 % a mathematical perspective, it is hardly so for
 % applications.
%People can always adopt a randomization device if they want and
% risk-seeking agents will want it.
Randomization has indeed been widely employed when useful, for instance to achieve efficiency in allocation problems, similar to the jackpot allocations that we derive (\citealp{BCKM13})
 %{#
and for behavioral stopping rules (\citeauthor{{HHOZ17}},  \citeyear{{HHOZ17}}). Even in individual decision making, people often prefer to invoke randomizations
 % JPE
(\citeauthor{{AO17}},  \citeyear{{AO17}}).
% v.k. (reviewed by \citeauthor{{AO22}},  \citeyear{{AO22}}).
 \cite{GKKL21} found that
 risk-seeking subjects in an experiment indeed preferred to resort to gambling rather than sharing risks,
  %{#
 confirming the empirical realism of randomization.
 %#}
The following simple example illustrates the use of randomization.

\begin{example}\label{ex:simple}
 Suppose that the total wealth is a constant $X=1$, initially equally endowed among $n$ strictly
 risk-seeking EU agents. Without external randomization, under individual rationality, no redistribution is possible. The allocation $\mathbf 1/n$ cannot be
 Pareto-improved. But with external randomization, the jackpot allocation that assigns all $X=1$ to agent $i$ if, say, side $i$ comes up of a fair $n$-sided die, with 0 left for the others, is a strict Pareto improvement. All agents are better off because they are strictly risk seeking. In many arguments in this paper, the random variable $U$ of Assumption \ref{ass:ER} is similarly used to generate a jackpot vector.
\end{example}

We are now ready to present the main technical result in this section.

\begin{theorem}[Counter-monotonic improvement]\label{thm:CT_improvement}
Assume that $X_1,\dots,X_n\in L^1$ are nonnegative, $X=\sum_{i=1}^n X_i$, and Assumption \ref{ass:ER} holds.
Then there exists $(Y_1,\dots,Y_n)\in \mathbb A_n(X)$ such that
\begin{enumerate}[(i)]
\item $(Y_1,\dots,Y_n)$ is counter-monotonic;
\item $Y_i\ge_{\rm cx} X_i$ for all $ i \in [n]$;
\item $ Y_1,\dots,Y_n $ are nonnegative.
\end{enumerate}
Moreover, $(Y_1,\dots,Y_n)$ can be chosen as a jackpot allocation of $X$.
\end{theorem}

\noindent Theorem \ref{thm:CT_improvement} is the
counter-monotonic analog of \citeauthor{LM94}'s (\citeyear{LM94}) comonotonicity improvement theorem.
%Note that it applies to general risk seeking, and is not restricted to particular models such as EU.
 When  applied to
 risk-seeking agents (not necessarily EU), it has a clear implication: $Y_1,\dots,Y_n$ will be preferred over the original payoffs $X_1,\dots,X_n$.
 Therefore, one may anticipate for strictly
 risk-seeking agents, constrained to the set of nonnegative random variables, that any
Pareto-optimal allocation or equilibrium allocation, if it exists, must be a jackpot allocation; this is formalized for EU agents in Theorem \ref{thm:Pareto}.
 As another immediate consequence, for any vector of initial endowments, a jackpot allocation obtained in Theorem \ref{thm:CT_improvement} is individually rational for any
 risk-seeking agent.

 \noindent \textbf{Proof sketch of Theorem \ref{thm:CT_improvement}.}
Here we provide a proof under a condition stronger than Assumption~\ref{ass:ER}, which allows for an explicit construction of $(Y_1,\dots,Y_n)$.
 \begin{myassump}{ER*}\label{ass:ERS}
There exists a standard uniform random variable $U$ on $(\Omega,\mathcal F,\p)$  independent of $(X_1,\dots,X_n)$.
\end{myassump}
 % \begin{align}
 % \mbox{There exists a standard uniform $U$ independent
 % of $(X_1,\dots,X_n)$.}    \tag{ER*}\label{eq:ERp}
 % \end{align}
\noindent Write $Z_i= ({\sum_{j=1}^i X_j}/{X} ) \id_{\{X>0\}}$ for $i\in [n]$ and $Z_0=0$.
Let $A_i =\{ Z_{i-1}\le U  < Z_i\}$ for $i \in [n]$, which are disjoint events and satisfy $\p(\bigcup_{i=1}^n A_i)=\p(X>0)$.
Let
$
Y_i =X\id_{A_i}  ~\mbox{for $i\in [n]$}.
$
Note that $\sum_{i=1}^n Y_i=X\id_{\{X>0\}} =X.$
Clearly, $(Y_1,\dots,Y_n)$ is
 counter-monotonic and it is a jackpot allocation of $X$.
For $i \in [n]$, we have
\begin{align*}
 \E\left[  Y_i \mid X_1,\dots,X_n\right]
& = \E\left[X \id_{\{Z_{i-1}\le U  < Z_i\}}\mid X_1,\dots,X_n \right] \\
&  = \E\left[X  (Z_i-Z_{i-1})\mid X_1,\dots,X_n \right]  =   X \frac{X_i}{X}\id_{\{X>0\}} = X_i.
\end{align*}
Hence, Jensen's inequality yields  $X_i\le_{\rm cx} Y_i$.
 %%%%%%%%%%%%%%%%%%%
 %\footnote{The conditional expectation relation
 % $X\laweq \E[Y|Z]$ for some $Z$ is one of the
 % equivalent conditions for $X\le_{\rm cx}Y$ used
 % to characterize ``increase in risk'' by \cite{RS70}.}
 %%%%%%%%%%%%%%%%%%%%
This proves the statement of Theorem \ref{thm:CT_improvement} under Assumption \ref{ass:ERS}.
To show the result under Assumption \ref{ass:ER}, we will use a few technical lemmas in Appendix \ref{app:lemmas}. The weaker Assumption \ref{ass:ER} is more desirable in the study of risk sharing, as it requires randomization only for the aggregate payoff $X$, rather than for each allocation.  The difference between Assumptions \ref{ass:ER} and  \ref{ass:ERS} is subtle, which affects the applicability of some results. This issue is discussed in detail in Appendix \ref{app:ER}.

\begin{example}
\label{ex:th1fails}
We give an example illustrating how the conclusions in Theorem \ref{thm:CT_improvement} fail without the external randomization in Assumption \ref{ass:ER}.
Consider $n=3$ and a finite space $\Omega=\{\omega_1,\dots,\omega_4\}$ of 4 states with equal probability.  Let $X_i=3 \id_{\{\omega_i\}}+ \id_{\{\omega_4\}}$ for $i\in [3]$ and $X=X_1+X_2+X_3$ (thus $X=3$).  Suppose that  there exists $(Y_1,Y_2,Y_3)\in \mathbb A_3(X)$ satisfying (i)--(iii) in Theorem \ref{thm:CT_improvement}. Then it has the form $Y_i = a \id_{A_i} +m_i$ for some $a,m_1,m_2,m_3\in \R$, $(A_1,A_2,A_3)\in \Pi_3$   (see Proposition \ref{prop:PCT}).
To fulfill the conditions  $Y_i\ge_{\rm cx} X_i$ and $0\le Y_i\le 3$, each $Y_i$ must only take the values $0$ and $3$. Hence, the mean of $Y_i$ is  a multiple of $3/4$, violating $\E[Y_i]=\E[X_i]=1$.
Therefore, such $(Y_1,Y_2,Y_3)$ does not exist.
Under Assumption \ref{ass:ER} (here $\Omega$ cannot be finite),
one can take $\p(A_i)=1/3$ to satisfy all desired conditions.
\end{example}

% , but it is not only a technical concern. To apply Theorem \ref{thm:CT_improvement} in risk sharing, we need the result under Assumption \ref{ass:ER}, rather than \ref{ass:ERS}. This issue is discussed in detail in Appendix \ref{app:ER}.

The assumption that $X_1, \dots, X_n$ are nonnegative implies that there is a minimum outcome $0$ at which
 risk seekers are  prevented from further
 zero-sum gambling, and this is necessary to obtain the jackpot allocations of Theorem \ref{thm:CT_improvement}. A similar statement can be made for scapegoat allocations. Then an assumption
 $X_i\leq 0$ for all $i$ implies that there is a maximum outcome $0$ at which
risk seekers are prevented from further
 zero-sum gambling, and this is necessary to obtain scapegoat allocations (Theorem \ref{thm:dualCT_improvement}). Because the convex order is invariant under constant shifts, Theorem \ref{thm:CT_improvement} and the above statement on negative payoffs immediately imply the following result, which is presented in the same form as its comonotonic counterpart.
\begin{proposition}
\label{coro:ct_improve}
 Suppose that $X_1,\dots,X_n\in L^1$ are all bounded from above or all bounded from below, $X=\sum_{i=1}^n X_i$, and that Assumption \ref{ass:ER} holds. Then there exists a
 counter-monotonic allocation $(Y_1,\dots,Y_n)\in \mathbb A_n(X)$ such that  $Y_i\ge_{\rm cx} X_i$ for all $i$.
\end{proposition}

Whereas results in this section are formulated on $L^1$ for generality,
the next few sections will focus on bounded random variables in the analysis of risk sharing problems.

\section{Pareto-optimal allocations for EU agents} \label{sec:EUPO}

\noindent This section analyzes Pareto optimality under EU without the restriction of universal risk aversion.

\subsection{Setting}

\noindent To include
 risk-seeking agents, it is warranted to specify bounds for the set of feasible allocations. Otherwise, these agents can keep on improving through continued mutual
 zero-sum gambles.
In the main text we focus on the case of lower bounds, leading to jackpot allocations. Results for upper bounds, leading to scapegoat allocations, readily follow as dual reformulations. For brevity, they are given in Appendix \ref{onl.app.:scapegoats}.

For convenience, we will denote the lower bound by 0. It can be interpreted as ruin, or as a
 no-short selling/borrowing constraint. We emphasize that, until Section
\ref{sec:RDU}, the 0 outcome does not refer to the status quo or to a reference point. It can designate a serious loss, such as ruin, and outcomes denoted as positive real numbers can then still designate losses.

Because risk seeking, the topic of this section, is the prevailing empirical finding for losses (as in times of economic decay; see the footnote in the introduction), the results in this section and Appendix \ref{onl.app.:scapegoats} are more relevant for losses than classical results in the literature. The latter invariably assumed universal risk aversion. Cases of merely loss outcomes occur when there is a loss upper bound for the outcomes, as in
 cost-sharing problems. For convenience, the main text considers the traditional outcome domain $\R_+$, where there is a lower bound. The dual results for upper rather than lower bounds are in Appendix \ref{onl.app.:scapegoats}. The flexible interpretation of our domain $\R_+$ (until Section \ref{sec:RDU}), with 0 a minimal outcome that may be a gain but also a loss, and dual results in Appendix \ref{onl.app.:scapegoats}, should be understood.

 Our main setting is described by the following assumption, followed by special cases.
 % If 0 is interpreted as ruin, then u will be very
 % concave there. Oh well. So be it.

\begin{myassump}{EU}\label{ass:EU}
 Each agent maximizes  EU with a strictly increasing continuous utility function $u_i:\R_+ \rightarrow \R_+$ with $u_i(0)=0$.
The domain of allocations is $\X=L^\infty_+$, where $L^\infty_+$ is the set of all nonnegative bounded random variables. The total payoff $X\in \X$ satisfies $\p(X>0)>0$.
\end{myassump}

 % \begin{myassump}{EURS}\label{ass:EURS}
 %  All agents are EU maximizers with utility function
 % $u_i:\R_+ \rightarrow \R_+$ increasing and strictly convex
 % with $u_i(0)=0$ for all $i\in[n]$.
 % The domain of allocations is $\X=L^\infty_+$, where
 % $L^\infty_+$ is the set of all nonnegative bounded
 % random variables. The total payoff $X\in \X$
 % satisfies $\p(X>0)>0$.
 % \end{myassump}
 %%%%%%%%
\noindent Here, $\p(X>0)>0$ avoids triviality and $u_i(0)=0$ is a normalization.

 % We will further sometimes consider the special simple
 % case of homogeneous
 % risk-seeking EU agents, formally stated below.\

 For EU agents,
the case of universal risk aversion is well understood in the classic risk sharing literature (e.g., \citealp{MWG95}).  We will focus on the following few settings, besides the classic case of universal risk aversion.
\begin{myassump}{EURS}\label{ass:EURS}
On top of Assumption \ref{ass:EU}, all agents are strictly risk seeking.
\end{myassump}

\begin{myassump}{EUM}\label{ass:EUM}
On top of Assumption \ref{ass:EU},  agents in a subgroup
$S\subseteq[n]$ are strictly risk seeking and the others in  $T=[n]\setminus S$ are strictly risk averse, with $S,T$ nonempty.
\end{myassump}

\noindent We also consider the following subcase of Assumption \ref{ass:EURS}.
\begin{myassump}{H-EURS}\label{ass:HO}
On top of Assumption \ref{ass:EURS}, agents are {homogeneous}; that is, $u_1 = \dots = u_n = u$ holds.
\end{myassump}

In informal discussions we often take
 risk seeking and risk aversion strict, as in the above assumptions.
 %To study Pareto-optimal allocations, we define the utility
 % possibility set and the utility possibility frontier of
 % the risk exchange economy.
 For the total wealth $X$, the \emph{utility possibility set} (\emph{UPS}) is the set of utility vectors $ (\E[u_1(X_1)],\dots, \E[u_n(X_n)])$
for $(X_1,\dots, X_n)\in \mathbb A_n(X)$, denoted by $\mathrm{UPS}(X)$.
The \emph{utility possibility frontier} (\emph{UPF}) is the subset of utility vectors achieved by
 Pareto-optimal allocations, formally denoted by $\mathrm{UPF}(X) $. We denote by $\mathrm{UPJ}(X)$ the subset of jackpot utility vectors.

We do not involve the initial endowments in this section because they are irrelevant for Pareto optimality. They are needed though for individual rationality, used in the next section. In our setting, individually rational
 Pareto-optimal allocations always exist (Lemma \ref{lem:IR} in Appendix \ref{app:lemmas}).

\subsection{Risk-seeking EU agents}

\noindent This subsection provides an explicit characterization of Pareto optimality under universal risk seeking.
%\subsection{Risk-seeking EU agents: Full characterization}
 % The next result formalizes the observation from the
 %  counter-monotonic improvement theorem that
 %  Pareto-optimal allocations for strictly
 %  risk-seeking agents are jackpot allocations. Moreover,
 %  Pareto-optimal allocations are then precisely
 %  $\blambda$-optimal allocations.
The following function will serve as a useful tool. For any Negishi weight vector $\blambda=(\lambda_1, \dots, \lambda_n) \in \R_+^n\setminus \{\mathbf 0\}$ we define
   %%%%%%%%%%%%
\begin{equation} \label{eq:V} V_{\blambda}: x \mo  \max_{i\in [n]} \lambda_i u_i(x). \end{equation}
   %%%%%%%%%%%%
 Under Assumption \ref{ass:EURS}, for each Pareto-optimal allocation $\mathbf X$ there exists a Negishi weight vector $\blambda$ such that $V_{\blambda}$ is an affine combination of the individual utility functions (Statement (iii) below) and is maximized at $\mathbf X$ (Statement (ii)). That is, $V_{\blambda}$  can be interpreted as the utility function of a representative agent, maximized at that particular
 Pareto-optimal allocation. Like the individual agents, this representative agent is risk seeking with $V_{\blambda}$ convex.

\begin{theorem}[Pareto optimality for risk seekers]\label{thm:Pareto}
Suppose that Assumptions \ref{ass:ER} and \ref{ass:EURS} hold. For an allocation $\mathbf X=(X_1,\dots, X_n)\in \mathbb A_n(X)$,
the following statements are equivalent:
\begin{enumerate}[(i)]
 \item $\mathbf X$ is Pareto optimal;
 \item $\mathbf X$ is
 $\blambda$-optimal for some $\boldsymbol{\lambda} \in \Delta_{n}$;
\item  $\mathbf X$ satisfies $\sum_{i=1}^n \lambda_iu_i(X_i) = V_{\boldsymbol{\lambda}}(X)$ for some $\boldsymbol{\lambda}=(\lambda_1,\dots,\lambda_n)\in \Delta_{n}$;
 \item  $\mathbf X$ satisfies $\sum_{i=1}^n \lambda_i\E[u_i(X_i)] = \E\left[V_{\boldsymbol{\lambda}}(X)\right]$ for some $\boldsymbol{\lambda}=(\lambda_1,\dots,\lambda_n)\in \Delta_{n}$;
 \item $\mathbf X $
  %#
 is a jackpot allocation, written as $\mathbf X=X(\id_{A_1},\dots,\id_{A_n})$, satisfying, for some $\boldsymbol{\lambda}=(\lambda_1,\dots,\lambda_n)\in \Delta_{n}$,
 $ \lambda_i u_i(X)\id_{A_i}  =V_{\blambda}(X) \id_{A_i}$ for each $i\in [n]$.

 \end{enumerate}
 \end{theorem}

\noindent Statement (iv) is a remarkable weakening of Statement (iii). The proof that $\mathbf X$ is a jackpot allocation in (v) is elementary, formalizing Marshall's intuition from the Introduction. The derivation of the explicit expression given after is not elementary, involving the
 Hahn--Banach Theorem. With this explicit expression, Theorem  \ref{thm:Pareto} has completely solved the case of
 risk-seeking agents under EU.

\noindent\textbf{Proof sketch of Theorem \ref{thm:Pareto}.}
The equivalence (i)$\Leftrightarrow$(ii), a special case of Proposition \ref{prop:convex} below, is based on the
Hahn--Banach Theorem
and the convexity of UPS, which relies on Assumption \ref{ass:ER}.
 % If $\blambda$ in (ii) has only positive components,
 % then (ii)$\Rightarrow$(i) in Theorem \ref{thm:Pareto}
 % is straightforward; the case with some zero
 % components requires additional arguments.
 % The direction (i)$\Rightarrow$(ii)
 % follows from the Hahn--Banach Theorem
 % and the convexity of UPS
 % in Proposition \ref{prop:convex}, which relies on
 % Assumption \ref{ass:ER}.
 % This fact is similar to the convexity of the
 % utility possibility set for concave EU agents
 % (e.g., \citealp{MWG95} for finite states),
 % but the latter result does not need Assumption
 % \ref{ass:ER}. Technically, this is because concave
 % utilities get improved by combinations of random
 % variables, but convex utilities do not have such a
 % property and the proof requires a technique of
 % probabilistic mixture via randomization.
The equivalence (ii)$\Leftrightarrow$(iv)
can be proved by verifying that $\E[V_\blambda(X)]$ is the maximum of $\sum_{i=1}^n \lambda_i\E[u_i(X_i)]$ over $(X_1,\dots,X_n)\in \mathbb A_n(X)$. %and(v)$\Rightarrow$(iv) is direct.
The implications
(v)$\Rightarrow$(iii)$\Rightarrow$(iv)  are straightforward.
The final (i)$\Rightarrow$(v) is proved by arguing that the above maximum can only be attained by the jackpot allocations in (iv) with Theorem \ref{thm:CT_improvement} and some techniques from probability theory.

Theorem \ref{thm:Pareto} immediately yields the relation
$\mathrm{UPF}(X) \subseteq \mathrm{UPJ}(X) \subseteq \mathrm{UPS}(X).$ In particular,
 Pareto-optimal allocations are jackpot allocations.
 % also mentioned after Theorem \ref{thm:CT_improvement}.
The converse does not hold in general if agents are not homogeneous. For example, if there are two agents, one with a more convex utility but receiving a smaller payoff on a larger event, as in the following example.
There, $\mathrm{UPF}(X)$ is a curve, different from $\mathrm{UPJ}(X)$, which is a convex set.

\begin{example}\label{ex:UPF_not_simplex}
 Set $u_1(x)= 3x^2$ and $u_2(x)=4x^3$ for $x\geq 0$ and let $X$ be uniformly distributed over $[0,1]$. We have $\E[u_1(X)]=\E[u_2(X)]=1$, and for any composition $(A_1,A_2)$ independent of $X$, we have
    $$\E[u_1(X\id_{A_1})] +\E[u_2(X\id_{A_2})]= \P(A_1)\E[u_1(X)] + \P(A_2)\E[u_2(X)]=1.$$
Now consider $A_1=\{X\in [0,3/4]\}$ and $A_2=\{X\in [3/4,1]\}$ so that $(A_1, A_2)$ is a composition that is not independent of $X$. We can compute
$
\E[u_1(X\id_{A_1})]
= \int_0^{3/4} 3x^2\d x \approx 0.422
$
and
$
\E[u_2(X\id_{A_2})]
= \int_{3/4}^1 4x^3\d x  \approx 0.684,
$
and hence this allocation is better than some jackpot allocations built using events independent of $X$.
Intuitively, the allocation that maximizes the equally weighted sum of the welfare gives everything to the agent who has the highest utility pointwise; see the left panel of Figure \ref{fig:example_UPF_not_simplex}.
Maximizing differently weighted sums of the welfare gives the UPF by Theorem \ref{thm:Pareto}, plotted in the right panel of Figure \ref{fig:example_UPF_not_simplex}.
\end{example}
 %%%%%%%%%%%%%%%%%%%%%%	

\begin{figure}[t]
    \centering
    \caption{An illustration of Example \ref{ex:UPF_not_simplex}. Left panel: the utility functions. Right panel: the utility possibility frontier.}
    \label{fig:example_UPF_not_simplex}
    \vspace{0.5cm}
\resizebox{0.45\textwidth}{0.4\textwidth}
{
\begin{tikzpicture}
    \begin{axis}[
        domain=0:1,
        samples=100,
        axis lines=middle,
        xlabel={$x$},
        ylabel={$u(x)$},
        every axis x label/.style={at={(ticklabel* cs:1)}, anchor=west},
        legend pos=north west,
        y label style={anchor=west},
        width=10cm,
        height=10cm,
    ]
    \addplot[blue, thick] {3*x^2};
    \addlegendentry{$u_1(x) = 3x^2$}
    \addplot[red, thick] {4*x^3};
    \addlegendentry{$u_2(x) = 4x^3$}
    \end{axis}
\end{tikzpicture}
}~~~
\resizebox{0.45\textwidth}{0.4\textwidth}{
\begin{tikzpicture}
    \begin{axis}
    [
        axis x line=middle,
        axis y line=middle,
        width=10cm,
        height=10cm,
        xlabel={$\E[u_1(X_1)]$},
        ylabel={$\E[u_2(X_2)]$},
        every axis x label/.style={at={(ticklabel* cs:1)}, anchor=south},
    xmin=0, xmax=1,
        ymin=0, ymax=1,
       enlargelimits=upper,
                 legend style={at={(0.5,-0.1)},anchor=north}
    ]
        % Plot E1 vs E2
       \addplot[color=black, thick] coordinates {
        (0,1)
        (1.0494569670357E-06,	0.999999999760532)
        (1.35696733779666E-05,	0.999999855151481)
        (3.59157928601271E-05,	0.999995191453758)
        (0.000578531108409359,	0.99995062699511)
        (0.00234841048119861,	0.999683499990672)
        (0.006591796875,	0.998764038085938)
        (0.015625,	0.99609375)
        (0.0332060131684642,	0.989325623249334)
        (0.0658561561967106,	0.973405295974267)
        (0.125,	0.9375)
        (0.231064701737553,	0.858210112746497)
        (0.421875,	0.68359375)
        (0.770254630117638,	0.293933254376659)
        (1,0)

    };
        % Diagonal line
        \addplot[
            color=lightgray,
            dashed,
            line width=1pt
        ] coordinates {(0,1) (1,0)};

        % Red point for lambda = 0.5
        \addplot[
            only marks,
            mark=*,
            mark size=2pt,
            color=red
        ] coordinates {(0.421875, 0.68359375)};
 % Coordinates for the red point
        % Label for the red point
       \node at (axis cs:0.421875, 0.68359375) [anchor=west, color=red] {$(0.42, 0.68)$};
      %  \legend{$\E[u_1(X_1)]$ vs $\E[u_2(X_2)]$}
    \end{axis}
\end{tikzpicture}
}
\end{figure}
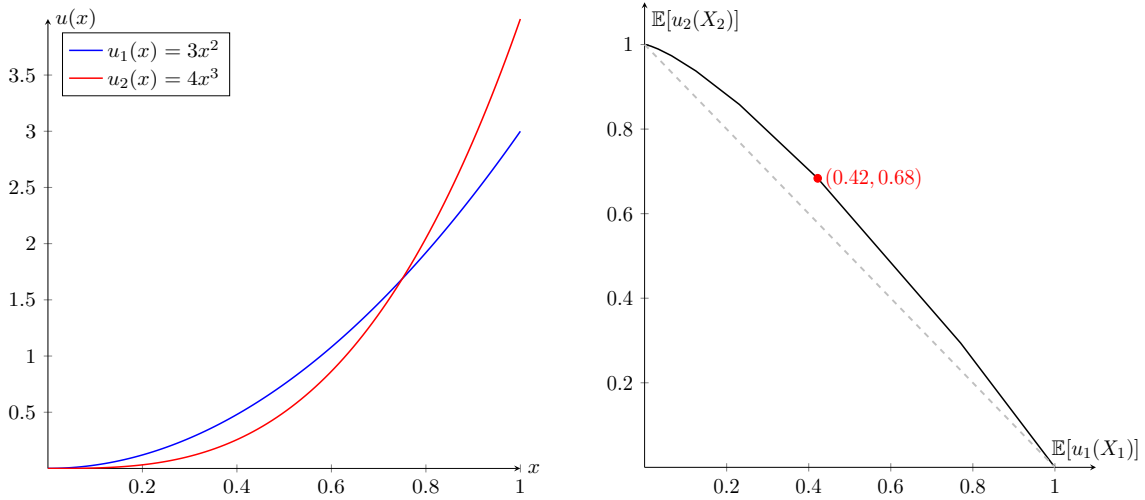

\cite{GKKL21} provided empirical support for our analysis. They found, in the gain domain with the classical risk aversion prevailing, that subjects preferred to share risks, as in classical theorems on risk sharing. However, in the loss domain subjects preferred to resort to gambling, in agreement with our claims that risk seeking is prevailing there, leading to phenomena as in Theorem \ref{thm:Pareto}.

Theorem \ref{thm:Pareto} considers the special case of a minimal outcome, denoted 0. Similar results hold if there is a maximal outcome but no minimal outcome (Theorem \ref{thm:dualPareto}). We then obtain scapegoat allocations in Statement (v). If there are both a minimal and a maximal outcome, then risk seekers take up more risk until they hit the first boundary, either the maximal or minimal outcome. Further examination of the latter case is left to future studies.

Although generally
$\mathrm{UPF}(X) \ne \mathrm{UPJ}(X)$,
there are two special cases in which equality holds and thus all jackpot allocations are Pareto optimal.
The first special case of $\mathrm{UPF}(X)=\mathrm{UPJ}(X)$ occurs when the total payoff is a constant.
\begin{proposition} \label{prop:PO_jackpot}
If $X=x>0$ is a constant and Assumptions \ref{ass:ER} and \ref{ass:EURS} hold, then all jackpot allocations of $X$ are Pareto optimal.
\end{proposition}
We next turn to Assumption \ref{ass:HO}, where all utility functions are the same, implying that the condition on $(A_1,\dots,A_n)$ in Statement (v) holds true for $\blambda =\mathbf 1/n$ because $V_{\blambda}(x)=u(x)/n$ for $x\in \R_+$.
This is
formally stated in the result below, which further characterizes the UPF as a simplex.

\begin{proposition}\label{prop:UPF}
 Under Assumptions \ref{ass:ER} and  \ref{ass:HO},
 $\mathrm{UPF}(X)=\mathrm{UPJ}(X)=\Delta_n(\E[u(X)])$.
\end{proposition}

\noindent Proposition \ref{prop:UPF}
 involves a structure that is uncommon for EU agents:
 Pareto optimality is not only implied by maximizing the social planner's problem where all the Negishi weights are normalized to 1, but also implies it, so that it is equivalent. This equivalence also holds for agents using quantiles or monetary risk measures as their preference functionals (\citealp[Proposition~1]{ELW18}).
Proposition \ref{prop:UPF} will be useful in Section \ref{sec:CEWT} to unify all four types of allocations: Pareto optimal, equilibrium, sum-optimal, and jackpot.

\subsection{General EU agents}
\label{sec:EU_gen}

\noindent Even in the gain domain, where risk aversion is prevailing, the classical assumption that all agents are risk averse is empirically unrealistic. It is desirable to still allow for at least some risk seeking agents. This section, thus, considers mixes of risk averse and risk seeking agents, which of course is very desirable for the loss domain.

   \begin{proposition}
    \label{prop:convex}
    Under Assumptions \ref{ass:ER} and \ref{ass:EU},
    \begin{enumerate}[(i)]
    \item
    both $\mathrm{UPJ}(X)$
    and $\mathrm{UPS}(X)$  are convex;
    \item  an allocation of $X$ is Pareto optimal if and only if
it is $\blambda$-optimal for some $\boldsymbol{\lambda} \in \Delta_{n}$.
    \end{enumerate}
\end{proposition}

 % \noindent On the other hand, $ \mathrm{UPF}(X)$,
 % on the b%et, is not necessarily
 % convex. below, where the UPS and UPF are depicted
 % in the case of two agents.
 %
 %
 % \begin{proposition}
 % \label{prop:EU_2}
 % Under Assumptions \ref{ass:ER} and \ref{ass:EU},
 % an allocation $\mathbf X \in \mathbb A_n(X)$  is
 % Pareto optimal if and only if
 % it is $\blambda$-optimal for some $\boldsymbol{\lambda}
 % \in \Delta_{n}$.
 % \end{proposition}
 %
\noindent For a given  $\boldsymbol{\lambda} \in \Delta_{n}$ or $\boldsymbol{\lambda} \in \R_+^n\setminus\{\mathbf 0\}$, computing the corresponding
 $\blambda$-optimal allocations is a standard numerical task, explained in
 Section \ref{sec:details_ex2}.

We next consider the remaining case under EU, where there are both
 risk-seeking and risk-averse agents.
For $\mathbf X=(X_1,\dots,X_n)$ and $S\subseteq[n]$, write
$\mathbf X_S=(X_i)_{i\in S} $.
\begin{theorem}[Subgroups]
\label{thm:EU}
Suppose that Assumptions \ref{ass:ER} and \ref{ass:EU} hold, and that $\mathbf X \in \mathbb A_n(X)$ is Pareto optimal.
 For any set $S\subseteq[n]$ of strictly
 risk-seeking agents, $\mathbf X_S$ is a jackpot allocation.
 For any set $T\subseteq[n]$ of strictly
 risk-averse agents,
 $\mathbf X_T$ is a comonotonic allocation.
\end{theorem}

\noindent Theorem \ref{thm:EU} directly follows from the two improvement theorems and some standard arguments. It is of course of theoretical interest to settle the case of risk-seeking agents, and to formalize the classical intuition of Marshall described in the introduction, as Theorem \ref{thm:EU} does. The main motivation for the theorem, though, is to increase empirical relevance. Thus, \cite{JS00} suggested that British racetrack bettors are part of the
 risk-seeking subgroup in society.
  %{#Remainder of para was rewritten.
So are clients in the gambling industry. The theorem documents incessant risk taking until some outside boundary is hit. In big markets it is unlikely to be the idealized case in the theorem of one agent taking the whole market endowment, and there will be other
 context-dependent boundaries where the whole market is allocated to some agents.
  %#}
 An impactful
 real-world example comes from Nick Leeson in 1995, a derivatives dealer who, in the loss domain below his usual status quo, exhibited risk seeking. He continued gambling until he hit a boundary, which in his case was the lower bound of ruin (outcome 0 in our notation) of the
 Barings Bank. Similar cases will exist but did not become publicly known. Even if we agree with Marshall's criticisms, we should not ignore the empirical existence of the phenomena described in Theorem \ref{thm:EU} and their dangers. Gambling addiction
(\citeauthor{{DHMRH25}},  \citeyear{{DHMRH25}};
\citeauthor{{TWC24}},  \citeyear{{TWC24}}) provides another impactful
 real-world example.
  %{#
 These are all
 real-world examples of the extreme behavior displayed in the above theorems, more precisely, their dual versions in Appendix \ref{onl.app.:scapegoats} with an extreme scapegoat allocation, hitting a ruinous lower bound for an agent.
  %#}

 \cite{LP26} derived equilibria when there are two commodities and two agents who maximize expected utility, where one agent is fully risk averse and the other is risk averse for one commodity but risk seeking for the other. Unlike Theorem \ref{thm:EU}, the solution may then be entirely interior.

 % \begin{corollary}[Pareto optimality for a
 % risk-seeking subgroup]\label{sbgrprsk}
 % Suppose that Assumption \ref{ass:ER} holds, the domain
 % of allocations is $\X=L^\infty_+$, the total payoff
 % $X\in \X$ satisfies $\p(X>0)>0$,
 % $(X_1,\dots,X_n)$ is Pareto optimal, and there is a
 % subgroup $A\subseteq[n]$ of agents who maximize EU
 % with strictly convex increasing utility functions.
 % Then  $(X_i)_{i\in A}$ is a jackpot allocation. That is,
 % for every state (with probability one) all but one of
 % them receive $0$.
 % \end{corollary}
   %%%

Theorem \ref{thm:EU} considered
  %#
 within-group exchanges, within the
 risk-averse group and then within the
 risk-seeking group. We next consider
  %#
 between-group  exchanges for these two groups.
For a given $\blambda\in \R_+^n\setminus\{\mathbf 0\}$, a
$\blambda$-optimal allocation $(X_1,\dots,X_n)$ can be obtained in two steps. Let $\blambda_S=(\lambda_i)_{i \in S} $, $\blambda_T=(\lambda_i)_{i \in T} $, and $\blambda^{ST}=(\sum_{i\in S}\lambda_i, \sum_{i\in T}\lambda_i)$. First, for each pair $(X_S,X_T)$ that sums to $X$, we find a $\blambda_S$-optimal allocation of $X_S$ to agents in $S$, which is a jackpot allocation and a $\blambda_T$-optimal allocation of $X_T$ to agents in $T$, which is a comonotonic allocation (Theorem~\ref{thm:EU}). By doing this, we obtain the total weighted utility of each group applied to $X_S$ and $X_T$. Second, we determine a $\blambda^{ST}$-optimal $(X_S,X_T)$  via a
 one-dimensional optimization  by maximizing the total weighted utility
 point-wise; see \eqref{eq:2groups}. The resulting optimal allocation $(X_S,X_T)$  may be a jackpot allocation itself if agents in $S$ have larger Negishi weights (``risk seeking prevails", see Example \ref{ex:RARS}.a), but it may also be a proportional allocation if agents in $T$ have larger Negishi weights (``risk aversion prevails", see Example \ref{ex:RARS}.b), or it may be somewhere in between. The $\blambda^{ST}$-optimal $(X_S,X_T)$ in the second step  then yields an $\blambda$-optimal allocation $(X_1,\dots,X_n)$ in the first step. The above steps are numerically efficient and sometimes they admit explicit formulas. Details of these steps, as well as computations for the next example, are in Section \ref{sec:details_ex2}.

\begin{example}
\label{ex:RARS}

Let $t$ be the cardinality of $T$.
For $i\in S$, let  $u_{i}$
be the convex function $u_i(x)= 3x+x^2 $
  %Our outcome space throughout the paper is Re. Then not strictly incr.
and for $i\in T$, let $u_i$ be
a common strictly increasing and strictly concave function satisfying
 $u_i(x)= 5 x - tx^2  $ on $[0, 2/t]$.
Let the aggregate payoff $X$ be distributed on $[0,2]$.
We consider two cases of $\blambda =(\lambda_1,\dots,\lambda_n)\in \R_+^n$ with $\lambda_i=\lambda_S>0$ for $i\in S$ and $\lambda_i=\lambda_T>0$ for $i\in T$; here we do not normalize $\blambda $.

\begin{enumerate}
    \item[(a)] Let $\lambda_S=5/4$ and $\lambda_T=1$.
A $\blambda$-optimal allocation  $(X_1,\dots,X_n)$ is
given by
\begin{align}
\label{eq:RARS_a}X_i  &= X J_i \id_{\{X> c\}} , ~i\in S
\mbox{~~~and~~~}
X_i = \frac{X}{t}\id_{\{X\le c\}}, ~i\in T,
\end{align}
where $c=5/9$ and $(J_i)_{i\in S}$ is any jackpot vector.
    \item[(b)] Let  $\lambda_S=1$ and $\lambda_T=2$. A
$\blambda$-optimal allocation may not be a jackpot allocation. For instance, if we take $X=2$, then   $X_S=1/2$ and $X_T =3/2$  necessarily hold, and hence any
 $\blambda$-optimal allocation cannot be a jackpot.
 % Agents in $A$ do not gamble with agents in $S$ collectively.
\end{enumerate}
  %%How can readers know where the example ends and the regular text resumes?
 In case (a), agents in $T$ collectively gamble with agents in $S$, whereas they do not in case (b).
\end{example}

 % \begin{example}
 % \label{ex:RARS}
 % Take $\blambda=\mathbf 1 /n$,  $S\subseteq [n]$
 % and $A=[n]\setminus S$.  Let $t$ be the cardinality of
 % $A$, and assume $S,A$ nonempty.
 % Let
 %   $u_{i}$ for $i\in S$
 % be given by the convex function $x\mapsto x^2$ and
 % $u_i$ for $i\in A$
 % be given by the concave function
 % $x\mapsto (( t x+1)^{1/2}-1)/t$.
 % With computation explained in Appendix \ref{app:ex}, we
 % obtain that there exists a constant $c\approx 0.453398$
 % such that for any aggregate payoff $X$,
 % a $\blambda$-optimal and
 % Pareto-optimal allocation $(X_1,\dots,X_n)$ is given by
 % \begin{align*}
 % X_i &= \frac{X}{t}\id_{\{X\le c\}}, ~i\in A
 % \mbox{~~~and~~~}X_i  = X J_i \id_{\{X> c\}} , ~i\in S,
 % \end{align*}
 % where $(J_i)_{i\in S}$ is any jackpot vector
 % independent of $X$.
 % The risk-averse agents collectively gamble with the
 % risk-seeking agents, but among the risk-averse agents,
 % they share equally; among the risk-seeking agents, they
 % further gamble among themselves (Theorem \ref{thm:EU}).
 % Moreover, if the total payoff is small so that $\p(X\le c)=1$,
 % then only the risk-averse agents participate;
 % if the total payoff is large so that $\p(X>c)=1$, then
 % only the risk-seeking agents participate.
 % Use a different $\blambda$ can change the value of $c$, but
 % it can also change the general structure of the
 % Pareto-optimal allocation.
 % [This example will also contain a
 % non-gambling Pareto-optimal allocation for later reference.]
 % \end{example}

\section{Competitive equilibria and welfare theorems}
\label{sec:CEWT}

  %# The intro of this section has been rewritten

\noindent We now analyze competitive equilibria of the risk sharing economy allowing for
 risk-seeking agents.  On traditional equilibria we can be brief: they will not exist with risk seeking agents because, without a bound imposed, those agents will always continue to increase stakes. However, this result is not empirically possible. In practice, there will always be boundaries. We therefore introduce bounded competitive equilibria, with an empirically realistic upper bound: the total supply of the market. In classical analyses with only risk averse agents, this condition is automatically implied by the other conditions and in this sense our condition is not an extra restriction. However, \cite{D59} already showed that this condition, when imposed as a prior restriction, essentially changes the strategic nature of the economy, e.g., by involving the total endowment of all agents, and it can lead to different efficient allocations and equilibria. We think that our new bounded equilibria are more useful for studying risk seeking than the traditional unbounded equilibria because they are empirically realistic whereas classical equilibria do not even exist. For
 risk-seeking EU agents, we fully characterize our competitive equilibria and obtain welfare theorems.

\subsection{General results}

 %Suppose that our agents satisfy Assumption \ref{ass:EURS}.
 %Then
\noindent For EU agents,
the individual optimization problem for agent $i$ that we consider is
  %%%%%%%%%%%%%%%%%%%%%%%
  \begin{align}
  \label{eq:indi_2}
&  \mbox{maximize}~ \E[u_i(X_i)] ~~~~ \mbox{over $0\le X_i\le X$} ~~~~\mbox{subject to} ~\E^Q [ X_i] \le \E^Q[ \xi _i], \end{align}
where $(\xi_1,\dots,\xi_n)\in \mathbb A_n(X)$ is the vector of initial endowments and
$Q$ is a probability measure.
 A {\it bounded equilibrium allocation\/} %(or
% {\it b-equilibrium allocation\/} for short)
$(X_1,\dots,X_n)\in \mathbb A_n(X)$  solves
 \eqref{eq:indi_2} for some  $(\xi_1,\dots,\xi_n) $ and  $Q$, which can be left unspecified, and market clearance holds. The constraint $X_i\ge 0$  in \eqref{eq:indi_2} reflects our setting of nonnegative random variables in $L^\infty_+$.
Our new constraint, $X_i\le X$, means that the agent's payoff $X_i(\omega)$ at each state $\omega$ cannot exceed  the total supply $X(\omega)$ in the economy. In what follows, we will omit ``bounded'' for brevity, where it is understood that all our equilibria henceforth are bounded.

With the above formulation of individual optimization,
 %#
Statement (ii) below provides a version of the second welfare theorem with the equilibrium price $Q$ explicitly given.

\begin{theorem}[Welfare]\label{thm:1WT}
 %Suppose that Assumptions \ref{ass:ER} and \ref{ass:EU}  hold.
The following statements hold.
\begin{enumerate}[(i)]
\item   Under Assumption \ref{ass:EU}, every equilibrium allocation of $X$ is Pareto optimal.
\item Under Assumptions \ref{ass:ER} and  \ref{ass:EURS}, every
Pareto-optimal allocation of $X$ is an equilibrium allocation,
with an equilibrium price $Q$ given by
\begin{align}\label{eq:general_price}
      \frac{\d Q}{\d \p} = \frac{V_{\boldsymbol{\lambda}}(X)}{X}\frac{1}{\E[V_{\boldsymbol{\lambda}}(X)/X]} ,
\end{align}
 for some  $ \boldsymbol{\lambda}\in \Delta_n$, where  $V_{\blambda}$ is defined in \eqref{eq:V}.
\end{enumerate}
\end{theorem}
  %%%%%%%%%%%%%%%%%%%%%%%%%%
 % The logical implications of Theorems \ref{thm:Pareto}
 % and \ref{thm:1WT} are summarized below:
 % \begin{center}
 % \begin{tabular}{ccc}
 %  Jackpot allocation  & $\Longleftarrow$ &
 % Equilibrium allocation\\
 %  $\Uparrow$ &  & $\Updownarrow$ \\
 % $\boldsymbol{\lambda}$-optimal allocation
 % (for some $\blambda$) & $\Longleftrightarrow$ &
 %  Pareto-optimal allocation
 % \end{tabular}
 % \end{center}

\noindent Statement (i) of Theorem \ref{thm:1WT} follows from standard techniques (see \citealp[Proposition 16.C.1]{MWG95} for a finite space), and our main novelty lies in Statement (ii), on which we make a few remarks.
In Theorem \ref{thm:Pareto}, a
 Pareto-optimal allocation (which is necessarily a jackpot allocation) is
 $\blambda$-optimal for some $\blambda \in \Delta_n$;
this vector $\blambda$ is the same one in \eqref{eq:general_price}, shown in the proof of Theorem \ref{thm:1WT}.
The vector of initial endowments associated with the equilibrium allocation $(X_1,\dots,X_n)$  is not unique, and it can be $(X_1,\dots,X_n)$ itself.
Another possibility concerns the proportional endowments
$$
\xi_i = \frac{\E^Q[X_i]}{\E^Q[X]}X, ~~~  i\in [n] .
$$

The equilibrium price $Q$ is generally not unique, even for a given vector of initial endowments. Uniqueness will be studied in Section \ref{sec:HEU} below.
The component
 $\E[V_{\boldsymbol{\lambda}}(X)/X]^{-1}$ in \eqref{eq:general_price} is a normalization to guarantee that $Q$ is a probability measure. The key property of $Q$ is that, since $u_1,\dots,u_n$ are convex, so is $V_{\boldsymbol{\lambda}}$. Hence, $V_{\boldsymbol{\lambda}}(x)/x$ is increasing in $x$. This implies that $V_{\boldsymbol{\lambda}}(X)/ X $ and $X$ are comonotonic. That is, the equilibrium price density $\d Q/\d \p$ increases as the aggregate endowment becomes more abundant.
This comonotonicity property is in stark contrast to the classical setting with strictly
risk-averse EU agents, where the equilibrium price is decreasing in the aggregate endowment $X$. Yet, this is not surprising, as the price structure reflects the marginal utility of agents. For strictly
 risk-averse EU agents, consumption is cheaper in
 high-endowment states, whereas for strictly
 risk-seeking EU agents it is reversed.

Theorem \ref{thm:1WT} implies that for
 risk-seeking agents,
 Pareto-optimal allocations and equilibrium allocations are equivalent.
This is not the case for general EU agents. The following result gives a necessary condition for equilibrium allocations.

\begin{theorem}
    \label{thm:CE_only}
   Suppose that Assumptions \ref{ass:ER} and \ref{ass:EU} hold, and  $\mathbf X \in \mathbb A_n(X)$ is an equilibrium allocation.
 For any set $S\subseteq[n]$ of strictly
 risk-seeking agents, $(\mathbf X_S,Y)$ is a jackpot allocation, where
 $Y=X-\sum_{i\in S} X_i.$
\end{theorem}

\noindent The equilibrium allocations in
 Theorem \ref{thm:CE_only} may exhibit, besides gambling within
 risk-seeking agents as in Theorem \ref{thm:EU}, also gambling between
 risk-seeking and other agents.
By Theorem \ref{thm:CE_only},
the Pareto-optimal allocation in
 Example \ref{ex:RARS}.b (``risk aversion prevails") cannot be an equilibrium allocation. The
 Pareto-optimal allocations in
Example \ref{ex:RARS}.a (``risk seeking prevails")
can be equilibrium allocations though; see Appendix \ref{app:mixed}.

A full understanding of competitive equilibria for the general mixed case is not yet available.
The next subsection focuses on homogeneous
 risk-seeking EU agents, where we can provide full understanding. We can then explicitly construct the equilibrium from any initial endowment vector.

\subsection{Homogeneous risk-seeking EU agents}
\label{sec:HEU}

\noindent For risk-seeking EU agents,
 Pareto-optimal allocation, $\blambda$-optimal allocations, and equilibrium allocations are equivalent, and all of those allocations are jackpot allocations. The reverse implication holds in the homogeneous case, as summarized next.

\begin{proposition} \label{prop:equiv}
Suppose that Assumptions \ref{ass:ER} and \ref{ass:HO} hold. For an allocation $\mathbf X=(X_1,\dots, X_n)\in \mathbb A_n(X)$, the following statements are equivalent:
(i) $\mathbf X$ is a jackpot allocation;
(ii) $\mathbf X$ is Pareto optimal;
(iii) $\mathbf X$ is sum-optimal;
(iv) $\mathbf X$ is an equilibrium allocation;
 (v)
      $\sum_{i=1}^n \E[u(X_i)]  = \E\left[u(X)\right]$.

 % \begin{enumerate}[(i)]
 %    \item $\mathbf X$ is a jackpot allocation;
 %     \item $\mathbf X$ is Pareto optimal;
 %     \item $\mathbf X$ is sum-optimal;
 %     \item $\mathbf X$ is an equilibrium allocation;
 %     \item
 %       $\sum_{i=1}^n \E[u(X_i)]  = \E\left[u(X)\right]$.
 % \end{enumerate}
\end{proposition}
 %%%%%%%%%%%%%%%
\noindent Proposition \ref{prop:equiv}
 follows by combining Proposition \ref{prop:UPF} and Theorems \ref{thm:Pareto} and   \ref{thm:1WT}.
Statement (v) of Proposition \ref{prop:equiv} implies that there is a vector $(\theta_1,\dots,\theta_n)\in \Delta_n$ such that for every $i$ we have  $\E[u(X_i)]=\theta_i\E[u(X)]$. The value $\theta_i$ can represent both the  probability of winning the lottery for agent $i$ and the relative purchase power of agent $i$ at equilibrium, formalized in Theorem \ref{thm:CE} below.
% The vector $(\theta_1,\dots,\theta_n)$ has several interpretations.
% In the context of statement (i), consider a jackpot allocation
%  $X\mathbf J$ where $\mathbf J\in \mathbb J_n$ is independent of $X $.
% With the relation  $(\theta_1,\dots,\theta_n)=\E[\mathbf J]$, $\theta_i$ can be interpreted as agent $i$'s probability of winning the prize $X$ in a lottery from an independent draw. The special case $ \E[\mathbf J]=\mathbf 1/n$ corresponds to a lottery in which every agent has a $1/n$ chance to win.
% In the context of statement (iv), for an equilibrium price $Q$, we can let  $\theta_i = \E^Q[\xi_i]/\E^Q[X]$  and interpret the vector $(\theta_1,\dots,\theta_n)$ as a vector of relative purchasing power.
% This point is further clarified in the next result.
 %Theorem \ref{thm:CE} below will clarify this point.

% As another feature of homogeneous
%  risk-seeking EU agents,
%  for any initial endowment vector $(\xi_i, \dots, \xi_n)\in \mathbb{A}_n(X)$, we can

 Next, we explicitly solve the competitive equilibrium  for any initial endowment vector $(\xi_1, \dots, \xi_n)\in \mathbb{A}_n(X)$, and show that when $(\xi_1, \dots, \xi_n)$ is nontrivial, the equilibrium price $Q$ is uniquely given by
 \begin{equation}\label{eq:price}
 \frac{\d Q}{\d \p} = \frac{u(X)}{X}\frac{1}{\E[u(X)/X]} ~ ,
 \end{equation}
 which is precisely
\eqref{eq:general_price} in Theorem \ref{thm:1WT} under Assumption \ref{ass:HO}.\footnote{Uniqueness should be interpreted in the  $\p$-almost sure sense, and  without loss of generality  we here assume that $Q$ is absolutely continuous with respect to $\p$.} In case  $(\xi_1,\dots,\xi_n)$ is trivial (e.g., $\xi_1=X$, $\xi_2 =\cdots = \xi_n=0$), the equilibrium allocation is just $(\xi_1,\dots,\xi_n)$ and the equilibrium price can be any probability measure  $Q$ with $\d Q/\d \p>0$.

  \begin{theorem}[Uniqueness]\label{thm:CE}
Suppose that Assumptions \ref{ass:ER} and  \ref{ass:HO} hold, and fix an initial endowment vector $\boldsymbol\xi \in \mathbb{A}_n(X)$. Let $Q$ be given by \eqref{eq:price}
 % $\theta_i= \E^Q[\xi_i]/\E^Q[X]$ for all $i$,
and $\boldsymbol \theta=\E^Q[\boldsymbol\xi]/\E^Q[X]$. %(\theta_1,\dots,\theta_n).$
\begin{enumerate}[(i)]
\item  The tuple $(X\mathbf J,Q)$ is a competitive equilibrium for any $\mathbf J\in \mathbb J_n$  independent of $X$
with $\E[\mathbf J]= \boldsymbol \theta$.
  \item The equilibrium price $Q$  is uniquely given by \eqref{eq:price} if $\boldsymbol\xi$ is nontrivial.
  \item  The utility vector of any equilibrium allocation is uniquely given by $ \E[u(X)] \boldsymbol \theta $.
 % where $v=\E[u(X)] $.
  \end{enumerate}
   \end{theorem}

\noindent Theorem \ref{thm:CE} implies, in particular, that competitive equilibria always exist for any initial endowment vector.
 Most remarkably, although the equilibrium allocation and its utility vector depend on the initial endowments, the equilibrium price does not. This is different from the case of heterogeneous agents, as \eqref{eq:general_price} depends on the initial endowments implicitly through the parameter $\blambda\in \Delta_n$.
For a given $(\xi_1,\dots,\xi_n)$ and
a competitive equilibrium $(X\mathbf J,Q)$, the jackpot vector $\mathbf J$ does not have to be independent of $X$ as in Statement (i) of Theorem \ref{thm:CE}. Any choice of $\mathbf J=(J_1,\dots,J_n)$
satisfying the budget constraint $\E^Q[XJ_i]=\E^Q [\xi_i]$ for all $i$
is indeed suitable for $(X\mathbf J,Q)$ to be a competitive equilibrium. Thus, for a given $\boldsymbol \xi$, the equilibrium allocation  is not unique, although $Q$ is unique.

Appendix \ref{app:CE} provides two further results on the existence of a competitive equilibrium under Assumption \ref{ass:EURS}  for any initial endowment.
In particular, a competitive equilibrium exists for (a) any initial endowment and  $n=2$, and (b) for initial endowments that are proportional to $X$.
These results also illustrate that the equilibrium price is not necessarily unique.

 %Our next result shows that, for any two
 %risk-seeking EU agents, a competitive equilibrium exists
 % for any initial endowment. The result also illustrates that
 % the equilibrium price is not unique.

\section{Rank-dependent utility agents} \label{sec:RDU}

\noindent The previous two sections assumed EU, and special attention was paid to agents who are entirely (on the whole domain) risk seeking or entirely risk averse. Whereas the assumption of risk seeking is a conceptually desirable addition to classical analyses, yet further empirical improvements are desirable. Empirical studies have shown that most people are neither entirely risk averse nor entirely risk seeking, but exhibit both attitudes in different subdomains (\citealp{BH08}; \citealp{ BMOT13}; \citealp{HIIW25}; \citealp{B25}). Thus, although risk seeking prevails for losses in general, there still is prevailing risk aversion, rather than seeking, for
 small-probability losses, e.g., as common in insurance. Similarly, even in the gain domain there is some risk seeking, which is even prevailing for small probabilities as in lotteries. These phenomena require deviating from classical EU.

This section provides results reckoning with the above 
%empirical phenomena and, thus, provides
 %# rest changed a bit.
% results for risk sharing that fully reckon with the 
empirically prevailing patterns of risk attitudes. Even though deviations from EU can be expected to be smaller among financial traders than among average human beings, it has now been accepted that those deviations are also prominent in finance, and many studies reckon with them (\citealp{BJW21}; \citealp{G23}). The study of cost sharing under
 non-EU is desirable anyhow given the widespread evidence against EU.

We assume \citeauthor{{Q82}}'s (\citeyear{{Q82}}) RDU where,
 %#
 for simplicity, we focus on gains, which has in fact been implicitly done for empirical purposes in all classical studies that assumed universal risk aversion and no reference point. For gains, RDU agrees with \citeauthor{{TK92}}'s  (\citeyear{{TK92}})  prospect theory, and these are the most popular behavioral models for risk (\citealp{BH08}; \citealp{BMOT18}; \citeauthor{{FE12}}, \citeyear{{FE12}}, Section 6; \citealp{HIIW25};
  % JPE 
 \citealp{SW10}). Prospect theory does require a specification of a reference outcome, which we denote by 0. Thus, we here give up the flexibility of outcome domains of the preceding sections. Contrary to the preceding sections, our restriction to nonnegative outcomes now implies that we only consider gains.
   %#
    
We assume the empirically prevailing shape of utility for gains: it is concave throughout
(\citealp[p.\ 264]{W10}), but (approximately) linear for moderate gains, say on an interval $[0,x_0]$.\footnote{This
  %%%%%%%%%%%%%%%%%%%%
 linearity assumption for small stakes has been affirmed by
 \citet[p.~189]{LV19}, \citet[Book V]{M90}, and
\citet[p.\ 264]{W10}
%\citealp{LV19}, p.~189;  \citealp{M90}, Book V, and  \citealp{W10}, p.\ 264. \citet[p.\ 24]{W10} 
 suggested linear utility below
 two months' salary.}
 %%%%%%%%%%%%%%%%%
\citet[p.~2512]{BMOT13} and \citet[p.\ 244]{W10}
explained that utility is more linear than classically thought because deviations from risk neutrality are in part due to factors other than utility curvature. With these assumptions, our analysis will illustrate effects of risk seeking beyond convex utility. For simplicity our main results will, further, assume homogeneous agents.
   % JPE
 \cite{E20} and  \cite{F10}
 provided methods for replacing heterogeneous agents by a homogeneous one.

We now define the RDU model.
As in Sections \ref{sec:EUPO} and \ref{sec:CEWT}, we consider nonnegative bounded random variables: $\X=L^\infty_+$,
 although now outcomes are interpreted as gains.
A function $w:[0,1] \rightarrow [0,1]$ is called a \emph{(probability) weighting function} if it is increasing and satisfies $w(0)=0$ and $w(1)=1$.
For a weighting function $w$ and an increasing utility function $u:\R_+ \to\R_+$, the RDU preference functional is given by
 %%%%%%%%%%%%%%
\begin{align}\label{eq:RDU_def}
\mathcal U(Y) = \int_\Omega u(Y)\d(w\circ\p) =
\int_0^\infty w(\p(u(Y)>x)) \d x,~~~Y \in \X.
\end{align}
 %%%%%%%%%%%%%%
Here, the first integral is a Choquet integral.
An \emph{RDU agent} has a preference functional given by
 \eqref{eq:RDU_def} for some weighting function $w$ and utility function $u$.
 When $u$ is linear, the RDU agent's preferences are represented by the dual utility functional of \cite{Y87}, denoted by
 %%%%%%%%%%%%%%%%
 $$\rho_w( Y)= \int Y \d (w \circ \p),~~~Y \in \X. $$
 %%%%%%%%%%%%%%%%
The RDU functional $\mathcal U(Y)$ in
 \eqref{eq:RDU_def} thus is $\rho_w(u(Y)).$
Under continuity, an RDU agent is risk seeking if and only if $w$ is concave and $u$ is convex, and is risk averse if and only if $w$ is convex and $u$ is concave
  %#
  (\citealp{CKS87}, assuming differentiability; \citealp{SZ08}, in general)
 %(\citealp{CKS87} assuming differentiability; \citealp{SZ08} in general).

The empirically prevailing weighting functions are neither convex nor concave, but a mix of those: they are inverse
 S-shaped (\citealp{BMOT13}, Section II.A; \citealp{FE12}; \citealp{LV19}). Although different formal definitions have been given, we will use the prevailing one (cavexity) because of its analytical convenience for our purposes. We call $w$ \emph{cavex} if it is continuous and there exists $p\in [0,1]$  such that $w$ is concave on $[0,p]$ and convex on $[p,1]$.\footnote{If
   %%%%%%%%%%%%%%%%%
 $p=0$ then $w$ is convex, and if $p=1$ then $w$ is concave. Empirically, $w(p)$ is not far from $p$, but we will not need this assumption.}
   %%%%%%%%%%%%%%%%%
Now there is neither entire risk aversion nor entire risk seeking, but there are factors going either way. Risk attitudes result from interactions between these factors, with different implications in different subdomains. Hence no simple general results can be expected to be found, and more complex and refined analyses are needed. The main purpose of this paper is to initiate the study of such results, which is warranted given the importance of achieving empirical realism. We will specify some contexts where we can already draw conclusions about which factors prevail.

 \cite{TK92} introduced the following family of cavex weighting functions
\begin{equation}
    \label{eq:TK}
 w_{\rm{TK}}(t)=\frac{t^\gamma}{(t^\gamma
 +(1-t)^\gamma)^{1/\gamma}}, ~~~t\in [0,1],
\end{equation}
 %%%%
for some parameter $\gamma$ estimated to be $0.71$ on average (e.g., \citealp{WZG04}).
 %#
We will use them below.

For a weighting function $w$, the \emph{concave envelope} of $w$, denoted $\overline{w}:[0,1]\rightarrow [0,1]$, will be a useful tool in our analysis. It is defined by
$$
\overline{w}(t) = \inf\{g(t): g\ge w \mbox{~on [0,1] and $g$ is concave}\}.
$$
\noindent It is concave. Because $\overline{w}\ge w$, we have $\rho_{\overline w}(Y)\ge \rho_w(Y)$ for all $Y\in \X$.
For a cavex weighting function $w$, its concave envelope $\overline w$ coincides with $w$ on an interval  $[0,\beta_w]\subseteq[0,1]$ (possibly empty) and
it is linear on $[\beta_w,1]$; see panel (b) of Figure \ref{fig:RDU}. The linear part captures the lowest line emanating from $(1,1)$ that dominates $w$, touching but not crossing $w$'s graph.
 % Useful implications are that
 % It is immediate that $w(\beta_w) \ge \beta_w$.
 % and that the line through $(\beta_w,w(\beta_w))$ and
 % $(1,1)$ dominates $w$.
As we will see in Theorem  \ref{thm:RDU}, enough risk seeking is implied this way to generate global
 risk-seeking-type results. We next state our main assumptions on the homogeneous RDU agents.

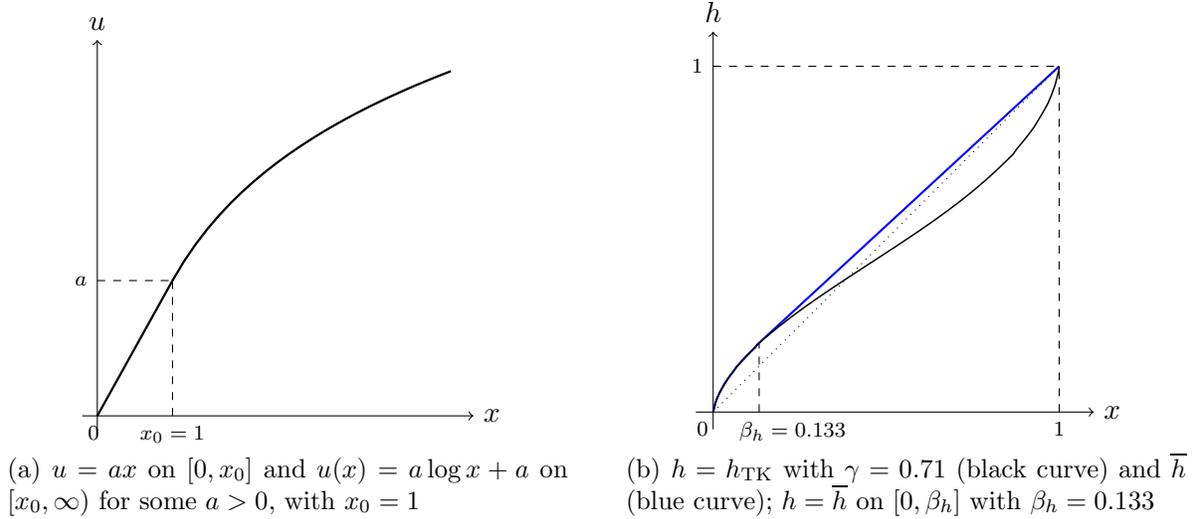
\begin{figure}
   \caption{An example of a utility function (left panel)
    and a weighting function (right panel) satisfying Assumption \ref{ass:RDU} for $n\ge 8$ and condition
 \eqref{eq:tech_con}.
    }
    \label{fig:RDU}
\begin{center}
\begin{subfigure}[b]{0.45\textwidth}
\centering
\begin{tikzpicture}[scale=1]
\draw[->] (-0.2,0) --(5,0) node[right] {$x$};
\draw[->] (0,-0.2) --(0,5) node[above] {$u$};
\draw[ domain=0:1,line width=0.3mm, black] plot ({\x}, {1.8*\x});
\draw[ domain=1:4.7,line width=0.3mm, black] plot ({\x}, {1.8*ln(\x)+1.8});
\draw[dashed] (1,0) -- (1,1.8);
 \draw[dashed] (0,1.8) -- (1,1.8);
\node[left] at (0,1.8) {\scriptsize $a$};
\node[below] at (-0.05,0) {\scriptsize $0$};
\node[below] at (1,0) {\scriptsize $x_0=1$};
\end{tikzpicture}
\caption{$u =a x$  on $[0,x_0]$ and $u(x)= a\log x+a$ on $[x_0, \infty)$ for some $a>0$, with $x_0=1$}
\end{subfigure}
~~~~
\begin{subfigure}[b]{0.45\textwidth}
\centering
\begin{tikzpicture}[scale=4.6]
\draw[->] (-0.05,0) --(1.1,0) node[right] {$x$};
\draw[->] (0,-0.05) --(0,1.1) node[above] {$w$};
\draw[ domain=0:0.133, blue, thick] plot ({\x},
{(\x^0.71)/((\x^0.71 +(1-\x)^0.71)^(1/0.71))});
\draw[domain=0.133:1, blue, thick] plot ({\x},
{(1-(0.133^0.71)/(0.133^0.71
+(1-0.133)^0.71)^(1/0.71))/(1-0.133)*(\x-0.133)
 +(0.133^0.71)/(0.133^0.71 +(1-0.133)^0.71)^(1/0.71))});
\draw[ domain=0:1.00, samples=100, line width=0.2mm, black] plot
({\x}, {(\x^0.71)/((\x^0.71 +(1-\x)^0.71)^(1/0.71))});
\draw[line width=0.2mm, black] (0.998,0.98) --(1,1) ;
\draw[dashed] (1,0) -- (1,1);
\draw[dashed] (0,1) -- (1,1);
\draw[dotted] (0,0) -- (1,1);
\draw[dashed] (0.133,0) -- (0.133,0.19793);
\draw[dashed]   (0.452,0) -- (0.452,0.429);
\node[below] at (-0.03,0) {\scriptsize $0$};

\node[left] at (0,1) {\scriptsize $1$};
\node[below] at (0.133,-0) {\scriptsize $\beta_w$};
\node[below] at (0.452,-0) {\scriptsize $p$};
\node[below] at (1,0) {\scriptsize $1$};
\end{tikzpicture}
\caption{$w=w_{\rm{TK}}$  (black) with $\gamma=0.71$  and $\overline{w}$ (blue); $w=\overline{w}$ on $[0,\beta_w]$ with $\beta_w=0.133$}
\end{subfigure}
\end{center}
\end{figure}

\begin{myassump}{H-RDU}\label{ass:RDU}
Each agent $i \in [n]$ is an RDU maximizer with a concave and increasing utility function $u$ on $\R_+$ that is linear on $[0,x_0]$ with $u(0)=0$ and $u(x_0)>0$, and a cavex weighting function $w$ with $w=\overline{w}$ for $t\in [0,1/n]$.
The domain of allocations is $\X=L^\infty_+$. The total payoff $X\in \X$ satisfies $\p(X>0)>0$.
\end{myassump}

\noindent Figure \ref{fig:RDU} depicts one pair $(u,w)$ that satisfies Assumption \ref{ass:RDU}. The condition that $w= \overline{w}$  on $[0,1/n]$ is equivalent to $n\ge 1/\beta_w$, which is easily satisfied for large $n$. Using $\gamma=0.71$ in \eqref{eq:TK},  we have  $w= \overline{w}$ for $t \in [0, 0.133]$ as depicted in panel (b) of Figure \ref{fig:RDU}.
Here, $n \ge 8$ is sufficient for this condition,
 %{#
 so that we maintain empirical realism.
%#}
It also implies that $w$ is concave on $(0,1/n)$, and we will later use strict concavity of $w$ on $(0,1/n)$ for some statements on strict domination.
The conditions on the set $\X$ and the total wealth $X$ are the same as in Assumption \ref{ass:EURS}.

The next result shows that previous results about Pareto optimality for (entirely)
 risk-seeking EU agents can be extended to RDU agents when they are sufficiently risk seeking for small probabilities. We will compare a jackpot allocation $X\mathbf J$
with the proportional allocation $\overline{\mathbf X}:=(X/n,\dots, X/n)$. The following condition, explained after the next theorem, is used in Statement (iii) there:
\begin{align}
    \label{eq:tech_con}
\limsup_{t\downarrow 0} \frac{w(t/n)}{w(t)}<1,~~w(1/n)<1,\mbox{~~and~~}  \lim_{x\to \infty} \frac{u(x/n)}{u(x)} =1.
\end{align}
In \eqref{eq:tech_con}, the conditions on $w$ hold for  $w_{\rm TK}$ in \eqref{eq:TK}, and the condition on $u$ holds for concave utility functions $u$ that are exponential, or logarithmic on $[z_0,\infty)$ for some $z_0>0$ (see panel (a) of Figure \ref{fig:RDU}), as well as for any bounded utility functions.\footnote{Some
   %%%%%%%%%%%%%%%%%%%%%%%%
extensions of Theorem \ref{thm:RDU} are presented in Appendix \ref{app:extensions}. In particular, the conclusions hold for utility functions $u$ such that $x\mapsto u(x)/x$ increases on $[0,x_0]$, which is weaker than assuming either linearity or convexity on $[0,x_0]$.}
 %%%%%%%%%%%%%%%%
In all results considered below, probability weighting induces risk seeking but utility induces risk aversion. If the effect of probability weighting is stronger then we get jackpot allocations similarly as in preceding results, but if the effect of utility is stronger then we do not.
 %# Move back?
 The following theorem presents two positive and then two negative results.

 %%%%%%%%%%%%%%%%%%%%%
\begin{theorem}[RDU]\label{thm:RDU}
    Suppose that Assumptions \ref{ass:ER} and \ref{ass:RDU} hold.
 Let the jackpot vector
  $\mathbf J\in \mathbb J_n$ be independent of $X$
 % and
 satisfying $\E[\mathbf J]=\mathbf 1/n$.
 % Write $\overline{\mathbf X}=(X/n,\dots, X/n)$.
  \begin{enumerate}[(i)]
      \item If  $X\le x_0$, then the jackpot allocation  $X\mathbf J$ is Pareto optimal and
 sum-optimal.
      \item
  If $X\le x_0$, $n\ge 2$, and $w$ is strictly concave on $(0,1/n)$, then $X\mathbf J$ strictly dominates $\overline{\mathbf X}$.
\item  If \eqref{eq:tech_con} holds, then there exists $y_0>0$ such that for $X\ge y_0$, $\overline{\mathbf X}$ strictly dominates $X\mathbf J$.

      \item  If $X$ is a positive constant, $u$ is
strictly increasing and differentiable on $\R_+$, and $w$ is strictly concave on $(0,1/n)$, then
 $\overline{\mathbf X}-\epsilon\mathbf 1 + n\epsilon \mathbf J$
strictly dominates $\overline{\mathbf X}$  for $\epsilon>0$ small enough.

  \end{enumerate}
\end{theorem}
 %%%%%%%%%%%%%%%%%%%%%

\noindent Statement (i) specifies conditions under which the
 risk-seeking components of RDU prevail for optimal risk sharing, so that essentially the same conclusions hold as in Theorems \ref{thm:CT_improvement} and \ref{thm:Pareto}, with Pareto optimality of some jackpot allocations, maximizing risks. Three points lead to this implication. First, whereas concave utility enhances risk aversion, we are now in a domain where utility is linear, not contributing any risk aversion. Second, there are enough agents to divide the risk over to ensure that the probability for each agent can be pushed below $\beta_w$, where $w$ generates risk seeking. Third, the conditions on $\overline{w}$ then ensure that risk seeking also prevails globally. Statement (ii) reinforces Statement (i).

  %  Statement s (iii) and (iv) of Theorem \ref{thm:RDU} are essentially negative results.
 Statement (iii)  specifies conditions under which, for large enough outcomes, the
 risk-averse implications of concave utility prevail over the
 risk-seeking implications of $w$ near $p=0$ so much that the jackpot allocation of Statement (i) can no more be Pareto optimal, being dominated by the safer proportional allocation. The latter allocation can even be optimal if the utility function has a satiation point (Example \ref{ex:RDU} below).

 \cite{CCZ20} found risk aversion for very large stakes, supporting the conditions of Statement (iii).
Using an arrangement with state lotteries in China, they could implement real incentives with extremely small probabilities,
 $10^{-5}$ and extremely large winning amounts, \${$10^6$}. They found that risk aversion then prevails, i.e., concavity of utility then dominates concavity of probability weighting, as in Statement (iii).
Statement (iv) specifies opposite conditions. For small enough outcome variations, the
risk-seeking implications of $w$ near $p=0$ (at $p = 1/n$) prevail over the
 risk-averse implications of concave utility so much that the safe proportional allocation can no more be Pareto optimal. Engaging in mutual
 zero-sum games of paying a sure $\epsilon$ in return for a risky $n \epsilon$ with probability $1/n$
(counter-monotonic, but not a jackpot) is surely appreciated for small $\epsilon$.

 \citet[p.~153--154]{M52} speculated that people choose risks for small stakes but safety for large stakes, but sought to accommodate it using EU. Theorem \ref{thm:RDU} has given an empirically realistic basis to his speculations using nonEU
  %{#
(and adding a market context),
  %#}
 confirmed by the wide existence of lotteries, sports betting, and casinos.

Unlike in Theorem \ref{thm:Pareto}, we are not able to show that all
 Pareto-optimal allocations are jackpot allocations or that they are
 $\blambda$-optimal. This is due to the
 non-linearity and nonconvexity in probabilistic mixtures of the RDU functionals that we consider.
 % \footnote{The
 %%%%%%%%%%%%%%%%%
 % functional $\rho_w$ is convex in probabilistic mixtures if
 % and only if $w$ is concave; see e.g.,
 % \citet[Theorem 3]{WWW20b}.}
   %%%%%%%%%%%%%%%%%
 We do not know whether the UPS is convex.

Finally, we present two elementary cases where competitive equilibria exist.
 %# nv.k. albeit elementarily so. We leave further study of equilibria as a topic for future research.

\begin{example}
  \label{ex:RDU}
 Assume that $ X\ge n y_0$ and that $u$ is constant on $[y_0,\infty)$ for some $y_0>x_0$; that is, $u$ has a satiation point. Then $\overline{\mathbf X}$ is Pareto optimal because it yields the maximum utility for every agent. It readily gives competitive equilibria, for instance for initial endowment $\overline{\mathbf X}$ with any price vector. The allocation $\overline{\mathbf X}$ strictly dominates any jackpot allocation when $w(1/n)<1$ (details are given in Section \ref{sec:ex-sati}).
\end{example}

 \begin{proposition} \label{thm:CE.RDU}
Suppose that Assumptions \ref{ass:ER} and \ref{ass:RDU} hold,
 $X=x$ is a constant in $(0,x_0]$, and the vector of initial endowments $\boldsymbol\xi=(\xi_1,\dots,\xi_n)\in \mathbb A_n(x)$ satisfies $\E[\xi_i]\le x\beta_w $ for all $i$.
Then $(x\mathbf J,\p)$ is a competitive equilibrium, for any $\mathbf J\in \mathbb J_n$ satisfying  $\E[x\mathbf J]=\E[\boldsymbol\xi]$.
 \end{proposition}

\noindent It is not surprising that $\p$ is the equilibrium price in Proposition \ref{thm:CE.RDU}, as the total endowment is constant across all states.
Although we assume that there is no aggregate uncertainty, the initial endowments in Proposition \ref{thm:CE.RDU} can be random.

Recall that  $\beta_w\ge 1/n$ under Assumption \ref{ass:RDU}. Hence, $(\xi_1,\dots,\xi_n)=(X/n,\dots,X/n)$ satisfies the condition in Proposition \ref{thm:CE.RDU}, and the corresponding equilibrium allocation is $X\mathbf J$ with $\E[\mathbf J]=\mathbf 1/n$ as in Theorem \ref{thm:RDU}.
The condition $\E[\xi_i]\le x\beta_w $ means that
each agent's initial endowment is not too large compared to the total endowment $x$.
Intuitively, agents tend to gamble for
 small-probability gains, which is the case when they have similar initial endowments. On the other hand, if one agent has a relatively large initial endowment, say $0.9x$, then it is no longer optimal for this agent to gamble, because the utility of $0.9x$ is $0.9  u(x)$, and
the utility of $x\id_{A_1}$ with $\p(A_1)=0.9$ is $w(0.9)u(x)$. Typically $w(0.9)<0.9$ (see Figure \ref{fig:RDU}), generating risk aversion for
 large-probability gains.
In this case, $(x\mathbf J,\p)$ is not a competitive equilibrium, and we do not know whether equilibria exist.

% # The next two paras were interchanged, and the last was rewritten.
  Extending our results to ambiguity (unknown probabilities) is a topic for future research. Some papers did consider the%, according to some, normative 
  case of ambiguity aversion
%(\citealp{HMRT26} and its references)
(\citealp{HMRT26}, and references therein). Here, as for risk, to achieve empirical realism, one has to allow for partial ambiguity seeking, even within individual agents
(\citealp{TV15}).

The results of this section show that no simple general result holds, with different factors going in opposite directions and prevailing factors depending on contexts. We did specify several contexts where conclusions could be drawn, each based on empirically documented assumptions. Those conclusions went in opposite directions. Thus, empirically realistic risk attitudes generate a more refined pattern than reported before in the literature, and more complex analyses are needed. Given the importance of risk sharing and empirical realism, such analyses are an important topic for future research.

\section{Conclusion}

\noindent The results in this paper lay a foundation for studying risk sharing with empirically realistic risk attitudes.
 %##
 %The
 %  counter-monotonic improvement theorem (Theorem \ref{thm:CT_improvement}) will be a useful tool for future studies.
We summarize what we know and what we do not. First, the setting of EU agents reveals a mix of insights and challenges.
 Pareto-optimal allocations and equilibrium allocations for
 risk-seeking agents (which is more realistic for losses than the commonly assumed risk aversion) are fully characterized, and the corresponding welfare theorems are established (Theorems \ref{thm:Pareto} and \ref{thm:1WT}).
 When there are subgroups of agents with different risk attitudes, conditions on
 Pareto-optimal allocations and equilibrium allocations within and across subgroups are obtained (Theorems \ref{thm:EU} and
 \ref{thm:CE_only}).
The case of homogeneous
 risk-seeking agents is better understood, as we can fully describe the competitive equilibria (with a bound needed for risk seeking added) with a unique equilibrium price (Theorem \ref{thm:CE}). Competitive equilibria may or may not exist when some agents are risk seeking and some are risk averse. A characterization is obtained for
 two-point aggregate payoffs (Theorem \ref{thm:CEEUM} in Appendix \ref{app:mixed}), and the more general case is not clear. For
 risk-seeking agents and a given initial endowment, the existence of competitive equilibria is not proved in general (with some results discussed in Appendix \ref{app:CE}) although we suspect that they always exist.

To achieve full empirical realism, we consider RDU agents, where patterns are more refined than reported before in theoretical studies. For the case of homogeneous agents, we obtain
 Pareto-optimal allocations for
 %#
 moderate
 payoffs (Theorem \ref{thm:RDU}), but we are not able to offer a complete characterization of
 Pareto-optimal allocations.
We provide two
  %# v.k. elementary
 cases of competitive equilibria for RDU agents. General results on competitive equilibria for heterogeneous RDU agents are topics for future research, as are extensions to ambiguity, but they will involve nontrivial new analyses. We hope that our paper will inspire future studies on optimal risk sharing under behaviorally realistic models, with further insights into the refined patterns that empirical reality presents.

\newpage
\begin{appendix}

\bigskip
\begin{center}
  \Large {Appendices}
\end{center}

We throughout follow the convention, also followed in the main text, that many claims are implicitly assumed to hold almost surely, that is, with probability one.

\section{Four lemmas} \label{app:lemmas}

\noindent We first provide four technical lemmas that are useful in the proofs of our main results. Denote by $\mathcal L$ the set of all
 random variables $X$ such that there exists a standard uniform random variable $U$ independent of $X$; that is, Assumption \ref{ass:ER} holds for $X$, an element of $\mathcal L$.

\begin{lemma}
\label{lem:existence}
Let $X_1\in \mathcal L$ have distribution $F_1$.
For any distribution  $F$ on $\R^{n}$ with the first marginal $F_1$,
there exists a random vector $(X_2,\dots,X_n)$
such that $(X_1, X_2,\dots,X_n)$ has distribution $F$.
\end{lemma}
\begin{proof}%[Proof of Lemma \ref{lem:existence}]
Let $(Y_1,\dots,Y_n)$ have distribution $F$ and let $K_x$ be a  regular conditional distribution of $(Y_2,\dots,Y_n)$ given $Y_1=x$ for each $x\in \R$.
We then let the random vector $(X_2,\dots,X_n)$ conditional on $X_1=x$ follow from $K_x$ for each $x\in \R$, which is possible because  $X_1\in \mathcal L$. This implies that $(X_1,\dots,X_n)$ is identically distributed to $(Y_1,\dots,Y_n)$.
 % we can always generate a sequence $U_2,\dots, U_n$ of
 % iid standard uniform random variables $U_2, \dots, U_n$
 % independent of $X_1$; see e.g., \citet[Theorem 1]{D12}.
 % Let $(Y_1, \dots, Y_n)$ be a random vector with
 % distribution $F$.
 % Let $X_2=F^{-1}_{Y_2|Y_1}(U_2|X_1)$ where
 %  $F^{-1}_{Y_2|Y_1}(\cdot|y)$ is the regular conditional
 % quantile function of $Y_2$ given $Y_1=y$, which is well
 % defined for almost every $y\in \R$. It is easy to check
 %  that $(X_1, X_2)\laweq (Y_1, Y_2)$. Similarly, we
 % can construct $X_j$ for $j=3, \dots, n$ by letting
 % $$X_j=F^{-1}_{Y_j|Y_1, \dots,
 %  Y_{j-1}}(U_j|X_1, \dots, X_{j-1}),$$
 % where $F^{-1}_{Y_j|Y_1, \dots,
 %  Y_{j-1}}(\cdot|y_1,\dots, y_{j-1})$ is the
 % regular conditional quantile function of $Y_j$ given
 %  $(Y_1, \dots, Y_{j-1})=(y_1, \dots, y_{j-1})
 % \in \R^{j-1}$. By construction, we have $(X_1,
 % \dots, X_n)\laweq (Y_1, \dots, Y_n)$, which
 % follows distribution $F$.
\end{proof}

\begin{lemma}
\label{lem:construct}
 % Fix an atomless probability space.
For any $X\in \mathcal L$ and $(X_1,\dots,X_n)\in \X^n$, there exist
$(X'_1,\dots,X'_n)\in \X^n$ and a standard uniform random variable $U$
such that
$(X,X_1',\dots,X'_n
)\laweq (X,X_1,\dots,X_n)$
and $U$ is independent of $(X,X_1',\dots,X'_n
)$.
\end{lemma}
\begin{proof}%[Proof of Lemma \ref{lem:construct}]
Let $H: \R^{n+1} \to \R$ be the joint distribution of $(X, X_1, \dots, X_n)$. Define $H': \R^{n+2} \to \R$ as $H'(x,x_1, \dots, x_n, u)=H(x, x_1, \dots, x_n)u$. It is clear that $H'$ is the joint distribution on $
\R^{n+2}$.
By Lemma \ref{lem:existence}, we can find a vector $(X, X'_1, \dots, X'_n, U) \sim H'$. Hence, we have $(X,X_1',\dots,X'_n
)\sim H$, $U \sim \mathrm{U}[0,1]$ and $U$ is independent of $(X, X'_1, \dots, X'_n)$.
\end{proof}

 % For the next lemma, recall the definition of $\rho_w$,
 % which, for $X\ge 0$, is given by
 % \begin{align*}
 % \rho_w(X)= \int  X  \d (w \circ \p) =
 % \int_0^\infty w (\p( X >x))\d x ,
 % \end{align*}
 % where $w:[0,1]\to[0,1]$ is an increasing function
 %  with $w(0)=0=1-h(1)$, called a probability
 % weighting function.

\begin{lemma}\label{lem:max}
Assume that the utility functions $u_1,\dots,u_n$ and
the weighting functions $w_1, \dots, w_n $
are continuous. For $X\in L_+^\infty\cap \mathcal L$ and $\X=L^\infty_+$,
 the set
 \begin{align}
 \label{eq:closed}
\mathrm{UPS}(X)=\Big\{\big(\rho_{w_1}(u_1(X_1)),\dots, \rho_{w_n}(u_n(X_n))\big):(X_1,\dots,X_n)\in \mathbb A_n(X) \Big\}
\end{align}
is compact.
In particular, for $(\lambda_1,\dots,\lambda_n)\in \Delta_n$,
the maximum of
$\sum_{i=1}^n\lambda_i \rho_{w_i} (u_i(X_i))$ over ${(X_1,\dots, X_n)\in \mathbb A_n(X)}$
is attainable. \end{lemma}

\begin{proof}%[Proof of Lemma \ref{lem:max}]
Let $m$ be the essential supremum of $X$.
Since the distributions of $(X,X_1,\dots,X_n)$ for $(X_1,\dots, X_n)\in \mathbb A_n(X) $ are all supported on the bounded region $[0,m]^{n+1}$, any sequence of such distributions has a weak limit.
Take any sequence of points $\mathbf v_1,\mathbf v_2,\dots$ in $\mathrm{UPS}(X)$ that converges to $\mathbf v_0\in \R^n$, and let $\mathbf X^{(1)},\mathbf X^{(2)},\dots$ be the random vectors in $\mathbb A_n(X)$ with the utility vectors $\mathbf v_1,\mathbf v_2,\dots$, respectively.
Let $F$ be a weak limit of the sequence of distributions of the random vectors $(X,\mathbf X^{(1)}),(X,\mathbf X^{(2)}),\dots$, which we argued above to exist.
Note that the first marginal of $F$ is the distribution of $X$.
By Lemma \ref{lem:existence}, there exists a random vector $\mathbf X'=(X_1',\dots,X_n')$
 such that  $(X,\mathbf X')$ has distribution $F$.
 Note that $ \mathbf 1 \cdot
 \mathbf X^{(n)}-X\to \mathbf 1 \cdot \mathbf X'-X$ in distribution, and hence $\mathbf X'\in \mathbb A_n(X)$.
Moreover, for each $i$, since  $u_i$ on $[0,\infty)$ and $w_i$ on $[0,1]$
are continuous,
the function $Y\mapsto\rho_{w_i}(u_i(Y))$ is continuous with respect to weak convergence by Theorem 4 of \cite{WWW20b}. Therefore,
$\mathbf v_0= \lim_{j\to\infty} \mathbf v_j = (\rho_{w_1}(u_1(X'_1)),\dots,\rho_{w_n}(u_n(X'_n)))$, which shows $\mathbf v_0\in \mathrm{UPS}(X).$
The boundedness statement follows by noting that  $\rho_{w_i}(u_i(X_i)) \le \rho_{w_i}(u_i(X))$ due to monotonicity, which is finite for each $i$.
The last statement on maximization follows from the compactness of $\mathrm{UPS}(X)$.
\end{proof}

Lemma \ref{lem:max} implies in particular that the maxima
$$\max_{(X_1,\dots, X_n)\in \mathbb A_n(X)}\sum_{i=1}^n\lambda_i
 \E[u_i(X_i)]
 \mbox{~~and~~}
 \max_{(X_1,\dots, X_n)\in \mathbb A_n(X)}\sum_{i=1}^n\lambda_i
 \rho_w(u(X_i))
 $$
are attainable under Assumptions \ref{ass:ER},  \ref{ass:EU} (for the first max), and \ref{ass:RDU} (for the second max). Moreover, individually rational
Pareto-optimal allocations always exist, as summarized in the next lemma.

\begin{lemma}
\label{lem:IR}
    Suppose that Assumption \ref{ass:ER} hold, and
    either Assumption \ref{ass:EU} or Assumption \ref{ass:RDU} holds.
    For any initial endowments, individually rational
 Pareto-optimal allocations exist.
\end{lemma}

\begin{proof}%[Proof of Lemma \ref{lem:IR}]
Let $(\xi_1,\dots,\xi_n)$ be the initial endowment vector.
Recall that each $\mathcal U_i$ is an EU preference functional under Assumption \ref{ass:EU}
and an RDU preference function under Assumption \ref{ass:RDU}.  By Lemma \ref{lem:max}, the set $\mathrm{UPS}(X)$ is compact.
Therefore, each set
$$
\mathrm{IR}(X):=\{(v_1,\dots,v_n)\in \mathrm{UPS}(X): v_i\ge \mathcal U_i(\xi_i) \mbox{ for all } i\in [n]\}
$$
is also compact. Maximizing $\sum_{i=1}^n v_i$ over
$\mathrm{IR}(X)$ yields a point in $\mathrm{UPS}(X)$ attained by an individually rational
 Pareto-optimal allocation.
\end{proof}

\section{Proofs} \label{app:proof}
\subsection{Proofs of results in Section \ref{sec:countermon}}

\noindent \begin{proof}[Proof of Proposition \ref{pr:3}]
 % An allocation is jackpot if and only if it is a
 % generalized Bernoulli distribution with $n$ results.
 % A probabilistic mixture of two such, with the same
 % $n$ results, is such a distribution again. By
 % atomlessness, there exists a random variable generating
 % the latter distribution.
   This result follows by noting that the conditions in
  \eqref{eq:jackpot2}
  are preserved under probabilistic mixing.
\end{proof}

\begin{proof}[Proof of Theorem \ref{thm:CT_improvement}]
In Section
 \ref{sec:ct.improve}, the result is proved when there exists a standard uniform random variable $U$ independent of $(X_1,\dots,X_n)$, that is, under Assumption \ref{ass:ERS}.
We extend this to the case where such $U$ may not exist.
Let
$(X,X_1,\dots,X_n)\laweq (X,X_1',\dots,X_n')$ be as in Lemma \ref{lem:construct}.
There exists a counter-monotonic improvement $(Y_1,\dots,Y_n)$  of $(X_1',\dots,X_n')$.
Note that $\sum_{i=1}^n X_i'=X=\sum_{i=1}^n X_i$.
Moreover, $Y_i\ge_{\rm cx} X_i'$ is equivalent to $Y_i\ge_{\rm cx} X_i$, because $X_i\laweq X_i'$.
Therefore, $(Y_1,\dots,Y_n)$ satisfies all desired conditions for $(X_1,\dots,X_n)$.
\end{proof}

\begin{proof}[Proof of Proposition
 \ref{coro:ct_improve}]
Suppose that $X_1,\dots,X_n$ are bounded from below. There exist constants $m_1,\dots,m_n$ such that $X_1+m_1,\dots,X_n+m_n$ are nonnegative. Write $m=\sum_{i=1}^n m_i$.  By Theorem \ref{thm:CT_improvement}, there exists a jackpot allocation $(Z_1,\dots,Z_n)\in \mathbb A_n(X+m)$  such that
$Z_i\ge_{\rm cx }X_i+m_i$ for each $i$.
It follows that
 $Z_i -m_i\ge_{\rm cx }X_i $ for all $i$.
Hence, the
 counter-monotonic random vector
 $(Z_1-m_1,\dots,Z_n-m_n)$ satisfies the desired conditions in the proposition.
The case that $X_1,\dots,X_n$  are bounded from above is analogous.
\end{proof}
\subsection{Proofs of results in Section \ref{sec:EUPO}}

\noindent Before proving Theorem \ref{thm:Pareto}, we first present a lemma for (ii)$\Leftrightarrow$(iv) in Theorem \ref{thm:Pareto}.
\begin{lemma}\label{lem:sum_op}
Suppose that Assumption \ref{ass:EURS} holds and  $\blambda=(\lambda_1, \dots, \lambda_n) \in \Delta_n$. Then,
$$\max_{(X_1, \dots, X_n) \in \mathbb{A}_n(X)} \sum_{i=1}^n \lambda_i\E[u_i(X_i)]=\E\left[V_\blambda(X)\right].$$
\end{lemma}
\begin{proof}%[Proof of Lemma \ref{lem:sum_op}]
 % By Lemma \ref{lem:max}, $\max_{(X_1, \dots, X_n)
 % \in \mathbb{A}_n(X)} \sum_{i=1}^n \lambda_i\E[u_i(X_i)]$
 % is attainable.
 % Take any composition $(A_1,\dots, A_n)\in \Pi_n$
 % satisfying
 % $A_i \subseteq \{\lambda_i u_i(X)=V_\blambda(X) \}$
 % for all $i$.
 % An explicit choice is
 % Let $$ A_i=\left\{\omega\in \Omega: i= \min
 % \left( \argmax_{j \in
 % [n]}\lambda_ju_j\left(X(\omega)\right)\right)\right\}~
 %  \mbox{ for $i\in [n]$}.$$

Since $\bigcup_{i=1}^n \{\lambda_i u_i(X)=V_\blambda(X)\} =\Omega,$ we can take $(A_1,\dots, A_n)\in \Pi_n$  such that $\lambda_i u_i(X)=V_\blambda(X) $ on $A_i$ for each $i\in[n]$.
We have
\begin{align*}
    \max_{(X_1,\dots, X_n)\in \mathbb A_n(X)}\sum_{i=1}^n\lambda_i\E[u_i(X_i)]
    &\geq \sum_{i=1}^n\E\left[\lambda_iu_i\left(X \id _{A_i}\right)\right] \\&= \sum_{i=1}^n\E\left[\lambda_iu_i(X) \id_{A_i}\right] = \sum_{i=1}^n \E\left[V_\blambda (X)  \id_{A_i}\right] =\E\left[V_\blambda(X)\right].
\end{align*}
Moreover, for any allocation $(X_1, \dots, X_n)\in \mathbb{A}_n(X)$, we have
\begin{align*}
\sum_{i=1}^n \E[\lambda_iu_i(X_i)]\le \sum_{i=1}^n \E[V_\blambda(X_i)]\le \E\left[V_\blambda\left(\sum_{i=1}^n X_i\right)\right]=\E[V_\blambda(X)],
\end{align*}
where the last inequality follows from the superadditivity of $V_\blambda$ because $V_\blambda$ is convex with $V_\blambda(0)=0$.
Combining the above two inequalities, we get the desired result.
\end{proof}

\begin{proof}[Proof of Theorem \ref{thm:Pareto}]

% We prove the theorem following the route %(iv)$\Leftrightarrow$(ii)$\Leftrightarrow$(i)$\Rightarrow$(iv)$\Rightarrow$(iv).

 The equivalence (i)$\Leftrightarrow$(ii) follows from Proposition \ref{prop:convex} (proved below) and the equivalence
(ii)$\Leftrightarrow$(iv) follows from Lemma \ref{lem:sum_op}. The implications
(v)$\Rightarrow$(iii)$\Rightarrow$(iv)  are straightforward.   It remains to prove (i)$\Rightarrow$(v).
 % By Proposition \ref{prop:convex}, the UPS of $X$  is a
 % convex set. By the
 % Hahn--Banach Theorem, for every
 % Pareto-optimal allocation $\mathbf X$, there
 % exists $\boldsymbol{\lambda}=(\lambda_1, \dots,
 % \lambda_n) \in \Delta_n$ such that $\mathbf X$ is a
 % $\blambda$-optimal allocation.

 % (iv)$\Rightarrow$(i):
 % If $\blambda$ has only positive components, it is clear that
 %  $\blambda$-optimality implies Pareto optimality.
 % Now suppose that some components of $\blambda$ are $0$.
 % Without loss of generality, suppose that
 % $\lambda_1,\dots,\lambda_j$ are positive, and
 % $\lambda_{j+1},\dots,\lambda_n$ are zero.
 % Let $\mathbf X=(X_1, \dots, X_n)$
 % satisfy $\sum_{i=1}^n \lambda_i\E[u_i(X_i)]=\E[V_\blambda(X)]$
 % and suppose that $\mathbf Y= (Y_1,\dots,Y_n)$
 % strictly dominates  $\mathbf X$.
 % Note that $\blambda$-optimality of $\mathbf X $ implies
 % that $\E[u_i(X_i)]=\E[u_i(Y_i)]$ for $i\in [j]$.
 % Hence, $\E[u_i(Y_i)]>\E[u_i(X_i)]$ must hold for
 % some $i\ge j+1$, and this further implies
 % $\p(\sum_{i=1}^j Y_i <X)>0$. Write $Y=\sum_{i=1}^j Y_i$.
 % Note that $V_\blambda$ is the maximum of $j$ strictly
 % increasing functions, and hence it is strictly increasing.
 % Using $Y\le X$ and $\p(Y<X)>0$, we have
 % $\E[V_{\blambda} (Y)]<\E[V_{\blambda} (X)]$.
 % By Lemma \ref{lem:sum_op} and $\lambda_i=0$ for $i\ge j+1$,
 % we have
 % $$
 % \sum_{i=1}^j \E[\lambda_i u_i(Y_i)]
 % \le \E\left[V_{\blambda} (Y)\right] < \E\left[V_{\blambda}
 % (X)\right]  = \sum_{i=1}^j \E[\lambda_i u_i(X_i)],
 % $$
 % contradicting that $\mathbf Y$ dominates $\mathbf X$. Hence, $\mathbf X$ is Pareto optimal.

(i)$\Rightarrow$(v): By Theorem \ref{thm:CT_improvement}, there is a jackpot allocation $\mathbf Y=(Y_1,\dots, Y_n)\in \mathbb A_n(X)$ such that $Y_i \geq_{\rm cx} X_i $ holds for all $i$.
As $u_i$ is strictly convex, we have $\E[u_i(Y_i)]= \E[u_i(X_i)]$ by Pareto optimality of $(X_1, \dots, X_n)$. By \citet[Theorem 3.A.43]{SS07}, we obtain  $Y_i\laweq X_i$ for each $i$.
Note that since $\mathbf Y$ is a jackpot allocation, we have $\sum_{i=1}^n\p(Y_i>0) = \p(X>0)$.
This and $Y_i\laweq X_i$ for each $i$ imply
$\sum_{i=1}^n\p(X_i>0) = \p(X>0)$,
and in particular $X_iX_j=0$ for $i\ne j$. It follows from
  \eqref{eq:jackpot2} that $\mathbf X$ is a jackpot allocation.
 For the remaining  statement, note that 
$\lambda_i u_i(X)\id_{A_i}\le V_{\blambda}(X)\id_{A_i}$ and 
$ \E[\sum_{i=1}^n  \lambda_i u_i(X)\id_{A_i}]= \E[\sum_{i=1}^n V_{\blambda}(X)\id_{A_i}]$
as implied by (iv). Therefore 
$\lambda_i u_i(X)\id_{A_i}= V_{\blambda}(X)\id_{A_i}$ holds.  
\end{proof}

\begin{proof}[Proof of Proposition \ref{prop:PO_jackpot}]
Suppose that a jackpot allocation $x\mathbf J$ for $\mathbf J\in\mathbb J_n$ is strictly dominated by an allocation $\mathbf Y$.
By Theorem \ref{thm:CT_improvement}, $\mathbf Y$ is dominated by another
jackpot allocation $x\mathbf J'$ with $\mathbf J'\in\mathbb J_n$.
The strict domination of $x\mathbf J'$ over $x\mathbf J$ implies
$\E[x\mathbf J']  \ge \E[x\mathbf J]$  componentwise with strict inequality for at least one component. This is not possible because both $\mathbf J$ and $\mathbf J'$ have components summing to $1$.
Hence, $x\mathbf J$ cannot be strictly dominated by $\mathbf Y$, and it is Pareto optimal.
\end{proof}

\begin{proof}[Proof of Proposition \ref{prop:UPF}]

We first show that $\mathrm{UPJ}(X) \subseteq \Delta_n(\E[u(X)])$. Let $(X_1, \dots, X_n)$ be a jackpot allocation of $X$ and observe that by construction we have $\E[u(X_i)]\geq 0$ for all $i$, and
 \begin{align*}
     \sum_{i=1}^n \E[u(X_i)]= \sum_{i=1}^n \E[u(X)\id_{A_i}]=\E[u(X)].
 \end{align*}
Hence, $(\E[u(X_1)], \dots, \E[u(X_n)])\in \Delta_n(\E[u(X)])$. Conversely, let  $(\theta_1, \dots, \theta_n)\in \Delta_n$. By Assumption \ref{ass:ER}, we can take $(A_1, \dots, A_n)\in\Pi_n$  independent of $X$
such that $\p(A_i)=\theta_i$  for all $i$. This gives $\E[u(X   \id_{A_i} )]=\p(A_i) \E[u(X) ] =\theta_i \E[u(X)]$ for all $i$.
Therefore, $\mathrm{UPJ}(X) = \Delta_n(\E[u(X)])$.
Note that every point in $\Delta_n(\E[u(X)]) $ is in the UPF as they are not dominated by any other points in $\mathrm{UPJ}(X)$.
Together with  $\mathrm{UPF}(X)\subseteq \mathrm{UPJ}(X)$ guaranteed by Theorem \ref{thm:Pareto}, we can conclude that $\mathrm{UPF}(X)=\mathrm{UPJ}(X) = \Delta_n(\E[u(X)])$.
\end{proof}

 % \begin{proof}[Proof of Corollary \ref{sbgrprsk}]
 %  We keep the allocation fixed for all agents $i$ outside
 % $A$, consider the resulting risk sharing problem for agents
 % in $A$, and apply the implication
 % (i)$\Ra$(iv) of Theorem \ref{thm:Pareto} to that.
 % \end{proof}

\begin{proof}[Proof of Proposition \ref{prop:convex}]
Statement (i): To show the convexity of $\mathrm{UPS}(X)$, take $\alpha\in (0,1)$ and $\mathbf x,\mathbf y\in \mathrm{UPS}(X)$,
and let $A$ be an event with $\p(A)=\alpha$. We can take allocations $\mathbf X,\mathbf Y\in \mathbb A_n(X)$ with utility vectors $\mathbf x$ and $\mathbf y$, respectively, independent of $A$ by using Lemma~\ref{lem:construct}.
Then $\alpha \mathbf x +
 (1-\alpha) \mathbf y$
 is the utility vector of $\id_A \mathbf X +
 (1-\id_A) \mathbf Y$ since the expected utility is affine on probabilistic mixtures. This shows that $\mathrm{UPS}(X)$ is convex.
 % Take any  $(X_1,\dots,X_n)\in \mathbb A_n(X)$
 % and $(Y_1,\dots,Y_n)\in \mathbb A_n(X)$ and let
 % $   \mathbf x=\left(\E\left[u_1(X_1) \right],
 % \dots,\E\left[u_n(X_n) \right]\right) $ and
 % $\mathbf y=\left(\E\left[u_1(Y_1)\right],
 % \dots,\E\left[u_n(Y_n)\right ]\right ).$
 % By Lemma \ref{lem:construct}, we can assume, without
 % loss of generality, that there exists a uniform random variable
 % $U$ independent of $(X,X_1,\dots,X_n)$
 % and $(X,Y_1,\dots,Y_n)$; otherwise we can
 % use $(X,X'_1,\dots,X'_n)$
 % and $(X,Y'_1,\dots,Y'_n)$ as in Lemma \ref{lem:construct},
 % and they yield the same vectors of utility as $\mathbf x$
 % and $\mathbf y$.
 %    Consider their convex combination  $\alpha \mathbf x +
 %  (1-\alpha) \mathbf y$  for some $\alpha \in  [0,1]$.
 % Let $Z_i =X_i\id_{\{U \le \alpha\}} + Y_i\id_{\{U>\alpha\}}$
 % for all $i$. Clearly, $\sum_{i=1}^n Z_i=X$; that
 % is, $(Z_1,\dots,Z_n) \in \mathbb{A}_n(X)$. Moreover, for
 % every $i$, we have
 % \begin{align*}
 %     \E[u_i(Z_i)] &= \E\left[u_i\left(X_i\id_{\{U \le
 % \alpha\}} + Y_i\id_{\{U>\alpha\}} \right) \right] \\&=
 % \E\left [u_i(X_i) \id_{\{U \le \alpha\}} +
 %     u_i(Y_i)\id_{\{U>\alpha\}}\right]
 %     = \alpha \E[u_i(X_i)]   +
 %   (1-\alpha)\E[u_i(Y_i)].
 % \end{align*}
 % Thus, $\alpha \mathbf x + (1-\alpha)\mathbf y
 % \in \mathrm{UPS}(X)$,  and $\mathrm{UPS}(X)$ is a convex
 % set.
 The convexity of $\mathrm{UPJ}(X)$ follows from the same argument with the additional fact that, by Proposition \ref{pr:3}, $\id_A \mathbf X +
 (1-\id_A) \mathbf Y$  is a jackpot allocation when
  $\mathbf X$ and $\mathbf Y$ are.

Statement (ii), the ``if" statement:
 Let $I=\{i\in [n]:\lambda_i=0\}$
 and $J=[n]\setminus I$.
 Take an allocation $(Y_1,\dots,Y_n)$ that dominates $(X_1,\dots,X_n)$.
 If $\E[u_i(Y_i)]>\E[u_i(X_i)]$
 for some $i\in J$, then $\E[\lambda_k u_k(Y_k)]> \E[\lambda_k u_k(X_k)]$, a contradiction.  If $\E[u_i(Y_i)]>\E[u_i(X_i)]$
 for some $i\in I$, then take $j\in J$ and let the allocation
 $(Z_1,\dots,Z_n)$ be given by
 $Z_i=0$, $Z_j=Y_j+Y_i$,
and $Z_k=Y_k$ for $k\ne i,j$. We have $\E[\lambda_k u_k(Z_k)]> \E[\lambda_k u_k(X_k)]$, a contradiction. Therefore,
$\E[u_i(Y_i)]=\E[u_i(X_i)]$.
This shows that $\mathbf X $ is Pareto optimal.

Statement (ii), the ``only if" statement: By Statement (i), the UPS of $X$ is convex. By the
 Hahn--Banach Theorem, for every
 Pareto-optimal allocation $\mathbf X$, there exists $\boldsymbol{\lambda}=(\lambda_1, \dots, \lambda_n) \in \Delta_n$ such that $\mathbf X$ is a
 $\blambda$-optimal allocation.
\end{proof}

\begin{proof}[Proof of Theorem \ref{thm:EU}]
     The case of risk-averse agents follows from  \citet[Theorem 3.1]{CDG12}. 
     Next prove the case of risk-seeking agents. 
     If there exists a uniform random variable on $[0,1]$ independent of $\mathbf X_S$, then the statement follows from (i)$\Rightarrow$(v) in Theorem \ref{thm:Pareto}.
     Otherwise, by   Lemma \ref{lem:construct}, there exists $(\mathbf X_T', \mathbf X_S')$ such that there exists a uniform random variable independent of $(\mathbf X_T', \mathbf X_S')$ and   $(X, \mathbf X_T', \mathbf X_S')$ is identically distributed to $(X,\mathbf X_T, \mathbf X_S)$, implying that $(\mathbf X_S',\mathbf X_T')$ is also Pareto optimal. 
     Now, applying (i)$\Rightarrow$(v) in Theorem \ref{thm:Pareto}, we get that $\mathbf X_T'$ is a jackpot allocation, 
     and by equality in distribution we know that $\mathbf X_T$ is also a jackpot allocation.
\end{proof}

\subsection{Proofs of results in Section \ref{sec:CEWT}}

\noindent \begin{proof}[Proof of Theorem \ref{thm:1WT}]
  The proof of Statement (i) is standard and omitted here. A
self-contained proof is in Section \ref{sec:pf_th4}. We show Statement (ii) below.
  Let $\mathbf X=(X_1,\dots,X_n)$ be a
 Pareto-optimal allocation. By Theorem \ref{thm:Pareto},  $\mathbf X=X\mathbf J$ for some $\mathbf J =(\id_{A_1}, \dots, \id_{A_n})\in \mathbb{J}_n$
and  $\boldsymbol{\lambda}=(\lambda_1,\dots,\lambda_n)\in \Delta_n$
with  $(\lambda_1u_1(X_1), \dots, \lambda_n u_n(X_n))=V_\blambda(X)\mathbf J$. From the proof of Theorem \ref{thm:Pareto}, we can take $\lambda_i=0$ if $X_i=0$ and $\lambda_i>0$ if $X_i\ne 0$.

Take the initial endowment $(\xi_1, \dots, \xi_n)=(X_1, \dots, X_n)$. We will show that $(X_1, \dots, X_n, Q)$ is a
 competitive equilibrium for $(\xi_1, \dots, \xi_n)$.
Let $z= \E[V_{\boldsymbol{\lambda}}(X)/X]$ and $Q$ be given by \eqref{eq:general_price}.
 If $\xi_i=0$, it is clear that $X_i=0$ solves individual optimality. Next, we discuss the case that $\xi_i$ is not 0, which implies $\lambda_i>0$.
 For such $i$ let
   $
  x_i=\E^Q[\xi_i] = \E^Q[X_i] =\E^Q[X\id_{A_i}],
  $
and notice that
 $$
  \lambda_i\E[u_i(X_i)] =\E \left[V_{\boldsymbol{\lambda}}(X) \id_{A_i}\right]=\E\left [X\frac{V_{\boldsymbol{\lambda}}(X)}{X}\id_{A_i}\right] =z  \E^Q\left[{X}\id_{A_i}\right]   = zx_i.
  $$
  For any $Y_i$ satisfying $0\le Y_i\le X$ and the budget constraint $\E^Q \left[Y_i \right]     \le   {x_i}$ we have
 \begin{align*}
      \lambda_i\E[u_i(Y_i)] &=\E\left[Y_i \frac{\lambda_iu_i(Y_i)}{Y_i}\right] \le \E\left[Y_i \frac{\lambda_iu_i(X)}{X }\right]  \le \E\left[Y_i \frac{V_{\boldsymbol{\lambda}}(X)}{X }\right]  =z \E^Q[Y_i]\leq zx_i,
  % =  \lambda_i\E[u(X_i)],
 \end{align*}
where the first inequality uses the fact that $x\mapsto u(x)/x$ is increasing and the second that $\lambda_i u_i\leq V_{\boldsymbol{\lambda}}$.
 Therefore, $(X_1,\dots,X_n,Q)$ satisfies individual optimality.
 Market clearance also holds because $ \sum_{i=1}^n \id_{A_i}=1.$
 Therefore, $(X_1,\dots,X_n)$ is an equilibrium allocation, being the initial endowments itself.
\end{proof}

\begin{proof}[Proof of Theorem \ref{thm:CE_only}]
Take $i\in S$, and consider the optimization problem
$$
\max_{0\le Z \le X} \E[u_i(Z)] \mbox{ subject to } \E^Q[Z] \le \E^Q[\xi_i].
$$
Since $u_i$ is strictly convex, the constraint is linear, and the probability space is atomless (by Assumption \ref{ass:ER}), the optimizer $Z=X_i$ of the above problem must be either $0$ or $X$ almost surely. This implies
 $X_i (X-X_i)=0$ for  $i\in S$. By \eqref{eq:jackpot2}, $(\mathbf X_S, Y)$ is a jackpot allocation.
\end{proof}

\begin{proof}[Proof of Proposition \ref{prop:equiv}]
The equivalence between (i),  (ii), (iii) and (v) follows from Theorem \ref{thm:Pareto} and Proposition \ref{prop:UPF}.
The equivalence between them and (iv) follows from Theorem \ref{thm:1WT}.
\end{proof}

\begin{proof}[Proof of Theorem \ref{thm:CE}] As usual, write $\boldsymbol \theta=(\theta_1,\dots,\theta_n)$,  $\boldsymbol\xi=(\xi_1,\dots,\xi_n)$
and $\mathbf J=(\id_{A_1},\dots,\id_{A_n})$.  We first show Statement (i).
    Denote by $x_i= \E^Q[\xi_i]=\theta_i\E^Q[X]$ and $z=\E[u(X)/X] >0$.
    Because $\d Q/\d\p$ is a function of  $X$, $\mathbf J$ has the same expectation $\boldsymbol \theta$ under $\p$ and $Q$, and it is independent of $X$ under both $\p$ and $Q$.
Hence,
  $
  \E^Q[X_i] = \theta_i \E^Q[X] =x_i,
  $
  and thus the budget constraint is satisfied for each $i$.
  Moreover,
  $$
  \E[u(X_i)] = \E[u(X)\id_{A_i}] = \E\left [X\frac{u(X)}{X}\id_{A_i}\right] =z  \E^Q[X\id_{A_i}]   = zx_i.
  $$
 For any $Y_i$ satisfying $0\le Y_i\le X$ and the budget constraint $\E^Q \left[Y_i \right]     \le   {x_i}$,
 using the fact that $x\mapsto u(x)/x$ is increasing,
 we have
  $$
 \E[u(Y_i)] =\E\left[Y_i \frac{u(Y_i)}{Y_i}\right] \le \E\left[Y_i \frac{u(X )}{X }\right]  =  z {\E^Q \left[Y_i \right] }    \le  z  {x_i} =   \E[u(X_i)] .
 $$
 Therefore, $(X_1,\dots,X_n,Q)$ satisfies individual optimality.
 Market clearance clearly holds.

Next, we show Statement (ii). Suppose that the tuple $(\mathbf X,Q)$ is a competitive equilibrium. By Proposition \ref{prop:equiv} we know $\mathbf X$ is a jackpot allocation.
We can write
   $\mathbf X=X \mathbf J$    for some $\mathbf J=(\id_{A_1},\dots,\id_{A_n})\in \mathbb J_n$
   satisfying $\E^Q[X \mathbf J ]=\E^Q[\boldsymbol \xi]$ by the binding budget constraint.
Without loss of generality, assume $\p(X>0)=1$, as $Q$ is arbitrary on $\{X=0\}$. By individual optimality we can see that $\p$ and $Q$ are equivalent measures.
 Define a probability measure $R$ by $\d R/\d Q= X/c$, where $c=\E^Q[X]$, and let  $Z= c (u(X)/X) (\d \p/\d Q )$.
   % $$
   % \frac{\d R}{\d Q} = \frac{X}{c},~\mbox{where $c=\E^Q[X]$},
   % $$
Note that for any $A\in \mathcal F$, we have
\begin{equation}
    \label{eq:convert}
  \E[u(X\id _A)] =      \E^R\left[\frac{\d \p}{\d Q} \frac{\d Q}{\d R}u(X) \id _A \right] =     \E^R\left[ Z\id _A \right].
\end{equation}
   Individual optimality of $X\mathbf J$ implies that for any $i$ and any $A\in \mathcal F$ satisfying $\E^Q[X \id_A] \le \E^Q[X\id_{A_i}]$, we have $   \E[u(X\id _A)]  \le  \E[u(X\id _{A_i})]$, and  by using \eqref{eq:convert}
 it gives    $
   \E^R[Z \id_A ] \le     \E^R[Z\id_{A_i} ].
   $
   Note that $ \E^Q[X \id_A] \le \E^Q[X\id_{A_i}]$ is equivalent to $R(A)\le R(\id_{A_i})$.
   Take $B_1,\dots,B_n$  with $R(B_i)= R(A_i)$ for all $i$ such that  $\id_{B_1},\dots,\id_{B_n}$ are comonotonic with $Z$; this is possible because $R$ is atomless.
We have
$ \sum_{i=1}^n  \E^R[Z \id_{B_i}]
\le \sum_{i=1}^n  \E^R[Z  \id_{A _i}]
=\E^R[Z].
$
Note that $\sum_{i=1}^n \id_{B_i}$ is not a constant because at least two of $ {A_1},\dots, {A_n}$ have positive probability under $R$ by the binding budget constraint and the assumption that $\boldsymbol \xi$ is nontrivial.
If $Z$ is not a constant, then the Fr\'echet-Hoeffding inequality gives
$\E^R[Z \sum_{i=1}^n \id_{B_i}] >
\E^R[Z] \E^R [\sum_{i=1}^n \id_{B_i}]=\E^R[Z] $, a contradiction.
%   Take $A  \in \mathcal F$ such that $R(A)  = R({A_i})$  and  $Z $ and $\id_A $ are comonotonic; this is possible because $R$ is atomless.
%  Suppose for contradiction that $Z$ is not a constant.
% The Fr\'echet-Hoeffding inequality gives
% $\cov (Z,\id_{A_i}) \ge \cov (Z,\id_{A}) \ge 0,$
% and  $ \cov (Z,\id_{A}) >0$ if $R({A})\in (0,1)$.
%   Since at least two of $\xi_1,\dots,\xi_n$ are not $0$,
%  by the binding budget constraint, at least two of $ {A_1},\dots, {A_n}$ have positive probability under $R$.
% Therefore,  $\cov (Z,\id_{A_i})>0$ for at least one $i$.
% However, $\sum_{i=1}^n \cov (Z,\id_{A_i})=\cov (Z,1)=0$, a contradiction.
 Therefore,  $Z$ is a constant, and $\d Q/\d P $ is equal to a constant times $u(X)/X$, showing that $Q$  is uniquely given by \eqref{eq:price}.

Statement (iii) follows immediately from (ii), by noting that for the unique price $Q$, there can only be one maximum utility value for every agent.
If $\boldsymbol \xi$ is trivial, say $\xi_1=X$, then the  equilibrium allocation is $(X,0,\dots,0)$, which has the unique utility vector $\E[u(X)]\boldsymbol \theta $.
\end{proof}

\subsection{Proofs of results in Section \ref{sec:RDU}}

\noindent \begin{proof}[Proof of Theorem \ref{thm:RDU}]

In (i) and (ii) below, as $X$ takes values in $[0,x_0]$, we can without loss of generality let $u$ be the identity function. Thus, each agent has a preference functional  $\rho_w.$

(i)  Write $\mathbf J=(\id_{A_1},\dots,\id_{A_n})$.
First, for each $i$, since $w=\overline{w}$ on $[0,1/n]$ and $\p(A_i)=1/n$, we have
 $$
\rho_w(X\id_{A_i})=
\int w(\p(X\id_{A_i}>x))\d x =
\int \overline w(\p(X\id_{A_i}>x))\d x
=\rho_{\overline w}(X\id_{A_i}).
 $$
Moreover,
\begin{equation}
    \label{eq:RDU-pf3}
\rho_w(X\id_{A_i})=  \rho_{\overline w}(X\id_{A_i})= \int \overline{w}\left(\frac 1n \p(X>x)\right)\d x .
\end{equation}
Next, we will show that
\begin{equation}
    \label{eq:RDU_pf2}
\sum_{i=1}^n \rho_w(X\id_{A_i})
=\sup_{(X_1,\dots,X_n)\in \mathbb A_n(X)} \rho_{\overline{w}}(X_i),
\end{equation}
and since $w\le \overline{w}$, this gives
 sum-optimality of $(X\id_{A_1},\dots,X\id_{A_n})$.
To show \eqref{eq:RDU_pf2}, it suffices to consider jackpot allocations $(X_1,\dots,X_n)$ because the preference functional $\rho_{\overline{w}}$ is risk seeking (\citealp{SZ08}), and we can apply Theorem \ref{thm:CT_improvement}.
Using the representation $(X\id_{B_1},\dots,X\id_{B_n})$ of jackpot allocations, we get
  \begin{align*}
     \sup_{(X_1,\dots,X_n)\in \mathbb A_n(X)} \rho_{\overline{w}}(X_i)  &=  \sup_{(B_{1}, \ldots, B_{n})\in \Pi_{n}} \rho_{\overline{w}}(X \id_{B_i})
     \\&=\sup \left\{\sum_{i=1}^{n} \int_{0}^{\infty} \overline{w}(\mathbb{P}(X \id_{B_{i}}>x)) \mathrm{d} x: (B_{1}, \ldots, B_{n})\in \Pi_{n}\right\}\\
        &\le \int_{0}^{\infty}       \sup \left\{\sum_{i=1}^{n}  \overline{w} (p_{i}) : \sum_{i=1}^{n}p_{i}=\mathbb{P}(X>x)\right\} \mathrm{d} x\\
        &= \int_{0}^{\infty} n \overline{w}
        \left(\frac{1}{n}\mathbb{P}(X >x)\right) \mathrm{d} x,
    \end{align*}
    where the last equality is due to concavity of $\overline{w}$.
    Using \eqref{eq:RDU-pf3}, we get \eqref{eq:RDU_pf2}, and hence the allocation $(X\id_{A_1},\dots,X\id_{A_n})$ is
 sum-optimal and Pareto optimal.

(ii)  Because $\overline{w}$ is strictly concave on $(0,1/n)$, $t\mapsto t \overline{w}(x/t)$ is strictly increasing in $t\ge 1$ for each $x\in (0,1/n)$. This implies
    $n\overline{w} (x/n)> w(x)$ for all $x\in (0,1)$. Hence,
    $$
  \frac 1n \int_{0}^{\infty}  \overline{w}
        \left( \mathbb{P}(X >x)\right) \mathrm{d} x  <
    \int_{0}^{\infty}  \overline{w}
        \left(\frac{\mathbb{P}(X >x)}{n}\right) \mathrm{d} x.
    $$
As a consequence,
    $$\rho_w(X/n)\le \rho_{\overline{w}}(X/n)= \frac{1}{n} \int_{0}^{\infty}  \overline{w}
        \left(\mathbb{P}(X >x)\right) \mathrm{d} x <   \int_{0}^{\infty}   \overline{w} \left(\frac{ \mathbb{P}(X >x)}{n}\right) \mathrm{d} x=\rho_w(X\id_{A_i}).$$
        This shows strict dominance.

(iii) Note that $w(t/n)/w(t)<1$ for $t\in (0,1)$ because $w$ is concave  on $[0,1/n]$.
Together with the continuity of $w$ and condition
 \eqref{eq:tech_con},
 we get
 $\sup_{t\in (0,1)} {w(t/n)}/{w(t)}<1$.
 Hence, we can take $\theta<1$ such that
$
\sup_{t\in (0,1)}  {w(t/n)}/{w(t)}<\theta,
$
and take $y_0>0$ such that $u(x/n)>  \theta u(x)$ for $x\ge y_0$, also guaranteed by
 \eqref{eq:tech_con}.
        Note that $X\ge y_0$ implies $u(X/n) \ge  \theta  u(X)$.
      For any event $A$ with $\p(A)=1/n$ independent of $X$,
     \begin{align*}\rho_w(u(X\id_{A}))& = \int_{0}^{\infty}   {w}
        \left(\mathbb{P}(u(X\id_A ) >x)\right) \mathrm{d} x
        \\  & =\int_{0}^{\infty}   {w}
        \left(\frac {\mathbb{P}(u(X  ) >x)}n\right) \mathrm{d} x
         \\&  <\theta \int_{0}^{\infty}  {w}
        (\mathbb{P}(  u(X)  >x )) \mathrm{d} x =  \rho_w(\theta u(X ))   \le \rho_w(u(X/n)).
        \end{align*}
Therefore, $\overline{\mathbf X}$ yields higher utility for each agent
than $X\mathbf J$ does, and thus strict domination holds.

(iv) Denote by $y=X/n>0$
and $Z^\epsilon$ the first component of
 $\overline{\mathbf X}-\epsilon \mathbf 1 + n\epsilon \mathbf J$.
 % , with $\mathbf J=(\id_{A_1},\dots,\id_{A_n})$.
 We can immediately compute
 \begin{align*}
\rho_w(u(Z^\epsilon))
 %  &=    u(y-\epsilon)+
 %\rho_w(( u(y+(n-1)\epsilon)-u(y-\epsilon))\id_{A_1})
 %\\ &
=u(y-\epsilon)+ ( u(y+(n-1) \epsilon)-u(y-\epsilon)) w(1/n).
\end{align*}
Taking derivative yields
\begin{align*}
\frac {\d }{\d \epsilon} \rho_w(u(Z^\epsilon)) &=
 - u'(y-\epsilon) +
 ( (n-1) u'(y+(n-1)\epsilon)+ u'(y-\epsilon)) w(1/n),
\end{align*}
 which converges to
 $ nu'(y)( w(1/n)-1/n) $ as $\epsilon \downarrow 0$.
Using $w(1/n)>1/n$, which is implied by the strict concavity of $w$ on $(0,1/n)$, we get that $ \d \rho_w(u(Z^\epsilon)) /\d \epsilon >0$ for $\epsilon>0$ small, and hence
$\rho_w(u(Z^\epsilon)) > u(y)$ for $\epsilon>0$ small enough. This shows that
$\overline{\mathbf X}-\epsilon \mathbf 1 + n\epsilon \mathbf J$ strictly dominates $\overline{\mathbf X}$.
\end{proof}

 \begin{proof}[Proof of Proposition \ref{thm:CE.RDU}]

Without loss of generality, assume that $u$ is the identity on $[0,x_0]$.
The budget conditions and market clearance hold by construction of $x\mathbf J$. We only need to show individual optimality.
Write $(\id_{A_1},\dots,\id_{A_n})=\mathbf J$.
Note that for any random variable $Y$ taking values in $[0,1]$
we have $Y\le_{\rm cx} \id_{A}$ where
$\p(A)=\E[Y]$.
Therefore, we have $\rho_{\overline{w}}(Y)\le \rho_{\overline{w}} (\id_{A} )
=\overline{w}(\E[Y])
$
because $\rho_{\overline{w}}$ represents a
risk-seeking preference.
It follows that for any random variable $X_i$ in $[0,x]$ satisfying the constraint $\E[X_i]= \E[\xi_i]$,
we have $$
\rho_{w}(X_i) \le
\rho_{\overline{w}}(X_i) \le \rho_{\overline{w}}(x\id_{A_i}) = x\overline{w}(\p(A_i))=x w(\p(A_i))=
\rho_w( x\id_{A_i}). $$
This extends to the case $\E[X_i]< \E[\xi_i]$ by monotonicity.
Therefore, individual optimality holds, showing that $(x\mathbf J,\p)$ is a competitive equilibrium.
 \end{proof}

\end{appendix}

\newpage

\renewcommand{\thepage}{O-\arabic{page}}
\setcounter{page}{1} % reset counter to 1
\begin{appendix}

\setcounter{table}{0}
\setcounter{figure}{0}
\setcounter{equation}{0}
\renewcommand{\thetable}{O.\arabic{table}}
\renewcommand{\thefigure}{O.\arabic{figure}}
\renewcommand{\theequation}{O.\arabic{equation}}

\setcounter{theorem}{0}
\setcounter{proposition}{0}
\renewcommand{\thetheorem}{O.\arabic{theorem}}
\renewcommand{\theproposition}{O.\arabic{proposition}}
\setcounter{lemma}{0}
\renewcommand{\thelemma}{O.\arabic{lemma}}

\setcounter{corollary}{0}
\renewcommand{\thecorollary}{O.\arabic{corollary}}

\setcounter{remark}{0}
\renewcommand{\theremark}{O.\arabic{remark}}
\setcounter{definition}{0}
\renewcommand{\thedefinition}{O.\arabic{definition}}
\setcounter{example}{0}
\renewcommand{\theexample}{O.\arabic{example}}
\setcounter{section}{0}
\renewcommand{\thesection}{O.\arabic{section}}

\begin{center}
    \Large
     Online Appendices for ``Optimal risk sharing, equilibria, and welfare with empirically realistic risk attitudes"
\end{center}

The seven online appendices provide additional results.

\section{Background on counter-monotonicity}
\label{app:CT}

\noindent We provide some technical background on
 counter-monotonicity. First, \cite{D72SM} obtained some necessary conditions for
 counter-monotonicity in dimensions more than two.

\begin{proposition}
[\citealp{D72SM}]\label{lem:pairwise}
If at least three of $X_1,\dots,X_n$ are
 non-degenerate, then $(X_1,\dots,X_n)$ are
counter-monotonic if and only if one of the following two cases holds:
\begin{align}
\label{eq:PCT01}
&\p(X_i>\essinf X_i,~ X_j>\essinf X_j)=0 \mbox{ for all  $i,j\in [n]$ with $i\ne j$}; \\
\label{eq:PCT02}
&\p(X_i<\esssup X_i,~ X_j<\esssup X_j)=0 \mbox{ for all $i,j\in [n]$ with $i\ne j$}.
\end{align}
A necessary condition for \eqref{eq:PCT01} is
$
\sum_{i=1}^n \p(X_i>\essinf X_i)\le 1,
  $
 and
a necessary condition for \eqref{eq:PCT02} is
$
\sum_{i=1}^n \p(X_i<\esssup X_i)\le 1.
  $
\end{proposition}

The alternative formulation \eqref{eq:jackpot2}  of jackpot allocations in Section
 \ref{sec:countermon} directly follows from Proposition \ref{lem:pairwise}.

Recall that $\Pi_n$ is the set of all
 $n$-compositions of $\Omega$.
The next proposition, which is a reformulation of \citet[Theorem 1]{LLW23aSM}, simplifies the stochastic representation of
 counter-monotonicity.

\begin{proposition} \label{prop:PCT}
 For $X \in \X=L^1$, suppose that at least three of $(X_1, \dots, X_n)\in \mathbb{A}_n(X)$ are
 non-degenerate. Then
$(X_1,\dots,X_n)$ is
 counter-monotonic if and only if there exist constants $m_1,\dots,m_n$ and $(A_1,\dots,A_n)\in \Pi_n$ such that
\begin{align}\label{eq:END_type1}
 \mbox{either}~~~~~&X_i = (X-m) \id_{A_i} +m_i~~~\mbox{for all $i$ with}~m=\sum_{i=1}^n m_i\le \essinf X;
 \\\label{eq:END_type2}
 \mbox{or}~~~~~ & X_i =  (X-m)\id_{A_i} +m_i~~~\mbox{for all $i$ with}~ m=\sum_{i=1}^n m_i\ge \esssup X.
 \end{align}
\end{proposition}

\begin{proof}%[Proof of Proposition \ref{prop:PCT}]
 The ``if" part follows from the fact that $\sum_{i=1}^nX_i=X$ and Proposition \ref{lem:pairwise}. We will show the ``only if" part.
 Assume that $(X_1, \dots, X_n)\in \mathbb{A}_n(X)$ is
 counter-monotonic. By \citet[Theorem 1]{LLW23aSM}, there exists $(A_1,\dots,A_n)\in\Pi_n$ such that
$$
X_i =  (X-m)  \id_{A_i} +m_i~~~\mbox{for all $i$},
$$
 where either  $m_i=\essinf X_i$ for all $i$ or $m_i=\esssup X_i$ for all $i$, and $m=\sum_{i=1}^n m_i$.
 If $m_i=\essinf X_i$  for all $i$, we have $m=\sum_{i=1}^n \essinf(X_i) \le \essinf (\sum_{i=1}^n X_i)= \essinf X$. If $m_i=\esssup X_i$ for all $i$, we have $m=\sum_{i=1}^n \esssup(X_i) \ge \esssup (\sum_{i=1}^n X_i)= \esssup X$.
  \end{proof}

Adding constants $m_i$ to $X_i$ does not affect
 counter-monotonicity. The terms with $m$ serve market clearance, with an inequality imposed to avoid crossing a boundary. The allocation $(X,0,\dots,0)$ is
counter-monotonic by taking $A=\Omega$ and $m=m_1=\essinf X$, and comonotonicity is trivial. Notice now that the allocations defined in Equations \eqref{eq:END_type1} and \eqref{eq:END_type2} reflect the conditions in Proposition \ref{lem:pairwise}. In \eqref{eq:END_type1}, for almost every $\omega \in \Omega$,  at most one agent receives more than their essential infimum. Conversely, in \eqref{eq:END_type2}, at most one agent receives less than their essential supremum.

\section{Discussion on Assumptions \ref{ass:ER} and \ref{ass:ERS}
}
\label{app:ER}

\noindent This appendix explains why  Assumption \ref{ass:ER} is more convenient than
 Assumption \ref{ass:ERS} in Theorem \ref{thm:CT_improvement}, although the latter is quite intuitive and it allows for a simple proof.
Recall that
   \ref{ass:ER} assumes
external randomization $U$ for $X$
and    \ref{ass:ERS}
assumes external randomization $U$  for $(X_1,\dots,X_n)$, which is stronger.

To characterize Pareto optimality or competitive equilibria,
we consider all random vectors $(X_1,\dots,X_n)\in \mathbb A_n(X)$ in the given probability space $(\Omega,\mathcal F,\p)$.
For instance, if we want to show as in Theorem \ref{thm:Pareto} that any
 Pareto-optimal allocation is a jackpot allocation,
then we need to be able to
 apply Theorem \ref{thm:CT_improvement}
 to any $(X_1,\dots,X_n)\in \mathbb A_n(X)$.
 Therefore, if we only have Theorem \ref{thm:CT_improvement} under Assumption \ref{ass:ERS},
we need that for \emph{all} random vectors an independent standard uniform random variable exists. Despite simple intuition, this assumption is very strong and rules out any standard Borel probability space; see   \citet[Example 7]{LWW20SM} for an illustration.
The intuition is that, in any  standard Borel probability space, there exists some random variable that exhausts randomness; that is, no external randomization is allowed  for that random variable.
Hence, to apply Theorem \ref{thm:CT_improvement} under Assumption \ref{ass:ERS} and obtain the desired results in subsequent theorems, we must exclude standard Borel probability spaces, whereas Assumption \ref{ass:ER} conveniently avoids this issue.
Thus, the difference between Assumptions \ref{ass:ER}   and  \ref{ass:ERS}  affects the applicability of the improvement theorem in risk sharing.
Nevertheless, for an application in which assuming the existence of a uniform random variable independent of all allocations is safe, Assumption \ref{ass:ERS} would suffice.

 % \begin{example}
 %     We can apply the above procedure to
 %  the constant allocation $(1/n,\dots,1/n)?
 % \mathbb A_n(X)$ with $X=1$ in Example \ref{ex:simple}.
 % Since $X_i=1/n$ for $i\in[n]$, we get $Z_i=i/n$ for
 % $i\in [n]$.
 % Therefore, the improved allocation $(Y_1,\dots,Y_n)$
 % is given by $Y_i= \id_{\{(i-1)/n\le U<i/n\}}$ for
 % $i\in [n]$, which is identically distributed to the
 % allocation $\mathbf J$ in Example \ref{ex:simple}.
 % On the other hand, if we apply the above procedure
 %  to a jackpot allocation $(\id_{A_1},\dots,\id_{A_n})$,
 % then $Z_i= \sum_{j=1}^i\id_{A_j}$, yielding
 % $(Y_1,\dots,Y_n)=(\id_{A_1},\dots,\id_{A_n})$.
 % \end{example}

On a related note,
Assumption \ref{ass:ER} or \ref{ass:ERS} is not needed for the classic comonotonic improvement theorem, and the reason is also intuitive: for
 risk-averse agents, external randomization does not enhance their utility, and therefore it is not needed. Mathematically, all comonotonic allocations of $X$ are measurable with respect to the
$\sigma$-algebra generated by $X$ (\citealp{D94SM}); this is certainly not true for
 counter-monotonic allocations.

\section{Additional technical details}
\label{app:ex}

\subsection{Computing   Pareto-optimal allocations}
 \label{sec:details_ex2}

\noindent We explain how to compute
 Pareto-optimal allocations in Section
 \ref{sec:EU_gen} for general EU agents.
For  $\boldsymbol{\lambda}=(\lambda_1,\dots,\lambda_n)\in \R_+^n \setminus\{\mathbf 0\}$, the goal is to find $(X_1,\dots,X_n)\in \mathbb A_n(X)$ that maximizes
 $ \sum_{i=1}^n \lambda_i\E[u_i(X_i)].
 $
Since the constraint $(X_1,\dots,X_n)\in \mathbb A_n(X)$  is
 pointwise on $\Omega$, and EU has a simple integral form,
the above maximum can be computed by
 point-by-point optimization for each value $x$ of $X$, that is,
 $$
 W_{\blambda}(x)= \sup\left\{ \sum_{i=1}^n   \lambda_i u_i(x_i): \sum_{i=1}^n x_i = x \mbox{ and $(x_1,\dots,x_n)\in \R_+^n $}\right\},~~x\in \R_+.
  $$
 An optimizer $(x_1(x),\dots,x_n(x))$  exists due to continuity (it may not be unique), and it
 yields a $\blambda$-optimal allocation
 $X_i=x_i(X)$ for $i\in[n]$, assuming measurability. We then have
  $$
  \E[W_{\blambda}(X)]=
 \max_{(X_1,\dots,X_n) \in \mathbb A_n(X)} \sum_{i=1}^n \lambda_i\E[u_i(X_i)].
  $$
 For any subset $S\subseteq [n]$ and $x\in \R_+$, we write
  $$ W_{\blambda}^S(x)= \sup\left\{ \sum_{i\in S}   \lambda_i u_i(x_i): \sum_{i\in S} x_i = x \mbox{ and $x_i\in \R_+ $ for $i\in S$}\right\}, \mbox{ with $W_{\blambda} ^{\varnothing}=0$.}$$
Let $T=[n]\setminus S$.  Then, we have
  \begin{equation}
  \label{eq:2groups}
  W_{\blambda} (x)= \sup\left\{  W_{\blambda}^{S}(y)+
 W_{\blambda}^{T}(x-y) : y\in [0,x]\right\},
  \end{equation}
  which is a one-dimensional optimization problem if $W_{\blambda}^{S}$
 and $W_{\blambda}^{T}$ are computable.

 In the mixed case specified in Assumption \ref{ass:EUM}, we have by Theorem \ref{thm:Pareto} that
 $W_{\blambda}^{S} (x)= \max_{i\in S} \lambda_i u_i(x)$,
 and the computation of  $W_{\blambda}^{T}$ is a standard convex program (maximization of a concave function on $\R_+^n$ under a linear constraint).
  In this case, both
$W_{\blambda}^{S}$
and $W_{\blambda}^{T}$  are easy to compute, so the overall problem boils down to a
 one-dimensional optimization for the sum of a convex and a concave function in \eqref{eq:2groups}.
Moreover, the optimal allocation can be obtained in two steps: first,
compute $(X_T,X_S)$ from \eqref{eq:2groups}; second, construct a
Pareto-optimal jackpot allocation of $X_S$ among agents in $S$ and a
 Pareto-optimal comonotonic allocation of $X_T $ among agents in $T$.

Next, let us specialize in the setting of Example \ref{ex:RARS} and verify the claims therein.
We first recall the setting.
For $i\in S$,   $u_{i}$
is the convex function $u_i(x)= 3x+x^2 $ and for $i\in T$,  $u_i$ is
a strictly increasing and strictly concave function satisfying
$u_i(x)= 5 x - tx^2  $ on $[0, 2/t]$, where $t$ is the cardinality of $T$.
The aggregate payoff $X$ is distributed on $[0,2]$.
Suppose that $\blambda =(\lambda_1,\dots,\lambda_n)\in \R_+^n$ satisfy $\lambda_i=\lambda_T >0$ for $i\in T$ and $\lambda_i=\lambda_S>0$ for $i\in S$.

First, let us compute $W_{\blambda}^S$
and $W_{\blambda}^T$. Using Theorem \ref{thm:Pareto}, we have
$$
W_{\blambda}^S(x) =  \max_{i\in S} \lambda_S u_i(x) = \lambda_S (3x+x^2),~~~x\in \R_+$$
By standard convexity argument, we also have
$$
W_{\blambda}^T (x) =
\sum_{i\in T} \lambda_T u_i\left(\frac{x}{t}\right) =
 \lambda_T (5x -x^2),~~~x\in [0,2].
$$
This means that the
 $\blambda$-optimal allocation among agents in $T$ is proportional.
Note that we do not need to specify $W_{\blambda}^T (x)$ for $x>2$. We analyze the two cases of $\blambda$ separately.

\begin{enumerate}
    \item[(a)] Let $\lambda_S=5/4$ and $\lambda_T =1$.
    In this case, $ W_{\blambda}^{S}(y)+
 W_{\blambda}^{T}(x-y)$ is convex in $y$; that is, risk seeking prevails.
To compute a $\blambda$-optimal allocation, we need to maximize, as in \eqref{eq:2groups},
$$ W_{\blambda}^{S}(y)+ W_{\blambda}^{T}(x-y)=
 \frac 5 4(3y+y^2) + 5(x-y) -(x-y)^2
$$
over $y\in [0,x]$. By convexity, we have either $y=x$ or $y=0$ at the optimum, leading to
$X_T =X\id_{\{X\le c\}}$
and
$X_S=X\id_{\{X>c\}}$, where the threshold $c=5/9$ can be easily
 computed by comparing the end-points.
Finally, using Theorem \ref{thm:EU} and Proposition \ref{prop:UPF}, we get that a
 $\blambda$-optimal allocation $(X_1,\dots,X_n)$ is
given by
 $$
X_i  = \frac{X}{t}\id_{\{X\le c\}}, ~i\in T
\mbox{~~~and~~~}X_i  = X J_i \id_{\{X> c\}} , ~i\in S,
 $$
where $(J_i)_{i\in S}$ is any jackpot vector.
 %Agents in $A$ collectively gamble with agents in $S$.

    \item[(b)] Let  $\lambda_S=1$ and $\lambda_T =2$.  In this case,  $ W_{\blambda}^{S}(y)+
 W_{\blambda}^{T}(x-y)$ is concave in $y$; that is, risk aversion prevails.
    To compute a
 $\blambda$-optimal allocation, we need to maximize, as in \eqref{eq:2groups},
$$ W_{\blambda}^{S}(y)+
 W_{\blambda}^{T}(x-y)=  3y+y^2 + 10(x-y) -2(x-y)^2 .
$$
The above function is concave in $y$, and its maximum may not necessarily be attained by $y=x$ or $y=0$. For instance, with $x=2$, the maximum is uniquely attained at $y=1/2$. Therefore, if we take $X=2$, then $X_T =3/2$ and $X_S=1/2$ necessarily hold, and
the $\blambda$-optimal allocation cannot be a jackpot allocation.
\end{enumerate}

\subsection{A self-contained proof of Theorem \ref{thm:1WT}, Statement (i)}
 \label{sec:pf_th4}

\noindent \begin{proof}[Proof of Theorem \ref{thm:1WT}, Statement (i)]
 Suppose for contradiction that $(X_1,\dots,X_n,Q)$ is an equilibrium but $(X_1,\dots, X_n)$ is strictly dominated by another allocation $(Y_1,\dots, Y_n)\in \mathbb A_n(X)$.
    There exists  $j\in[n]$ such that $\E[u_j(Y_j)]>  \E[u_j(X_j)]$, and by
    the fact that $(X_1,\dots,X_n,Q)$ is an equilibrium, we have that $\E^Q[Y_j]>\E^Q[\xi_j]\ge \E^Q[X_j]$.
    Because $\sum_{i=1}^n \E^Q[Y_i] = \E^Q[X] = \sum_{i=1}^n \E^Q[X_i],$
   there exists $i\in[n]$, $i\neq j$ such that  $\E^Q[Y_i]<\E^Q[X_i] $. By Pareto dominance we also have  $\E[u_i(Y_i)]\ge\E[u_i(X_i)]$.
    Let $\alpha = \E^Q[X_i-Y_i]/\E^Q[X-Y_i]$.
    Because $\E^Q[Y_i]<\E^Q[X_i] \le \E^Q[X]$, we have $\alpha\in (0,1]$. Let
    $
Z_i = Y_i +(X-Y_i)\alpha.
    $
It is clear that $Y_i\le Z_i\le X$ and $\E^Q[Z_i]=\E^Q[X_i]\le \E^Q[\xi_i]$. Recall that $\E^Q[Y_i]< \E^Q[X]$, which implies $Q(Z_i >Y_i)>0$, and hence,
$\p(Z_i >Y_i)>0$. Because $u_i$ is strictly increasing, we obtain  $\E[u_i(Z_i)]>\E[u_i(Y_i)]\geq \E[u_i(X_i)]$, contradicting individual optimality for agent $i$.
\end{proof}

\subsection{Details in Example \ref{ex:RDU}}
\label{sec:ex-sati}

\noindent Let $v=u(y_0)$. The allocation $\overline{\mathbf X}$ yields the maximum utility $v$ to all agents, and hence it is Pareto optimal. To show strict domination, let $(X_1,\dots,X_n)$ be a jackpot allocation and $p_i=\p(X_i>0)$ for all $i$. With this allocation, agent $i$ has utility
$$\int_0^{v} \overline{w}( \p(u(X_i)>x))\d x \le v w(p_i).$$
Since at least one $p_i$ is less than or equal to $1/n$, the condition $w(1/n)<1$ guarantees that at least one agent has a utility less than $v$, and thus the jackpot allocation is strictly dominated by $\overline{\mathbf X}$.

\section{Existence of competitive equilibria}
\label{app:CE}

\noindent We discuss the existence of competitive equilibria in the setting of Section \ref{sec:CEWT} under Assumption \ref{ass:EURS} for a given initial endowment $(\xi_1, \dots, \xi_n) \in \mathbb A_n(X)$.
We obtain two results in this appendix in the case of two
 risk-seeking EU agents and in the case of proportional initial endowments.

\subsection{Two risk-seeking EU agents}

\noindent Our next result shows that, for any two
 risk-seeking EU agents, a competitive equilibrium exists for any initial endowment. The result also illustrates that the equilibrium price is not unique.

\begin{proposition}\label{prop:two_agents}
Assume $n=2$ and Assumptions \ref{ass:ER} and \ref{ass:EURS}. For any initial endowment vector $(\xi_1, \xi_2) \in \mathbb{A}_2(X)$, there exists a competitive equilibrium $(X_1, X_2, Q)$, where
\begin{align}
      \frac{\d Q}{\d \p} = \frac{u_1(X)}{X}\frac{1}{\E[u_1(X)/X]} ~ .
      \label{eq:existence}
\end{align}
  \end{proposition}

\begin{proof}%[Proof of Proposition \ref{prop:two_agents}]
Without loss of generality, we can assume $\E[\xi_1]>0$ and $\E[\xi_2]>0$; otherwise $(X,0)$ or $(0,X)$ is an equilibrium allocation with any equilibrium price.
Moreover, we can assume $\p(X=0)=0$, because the allocation on the event $\{X=0\}$ is trivial.

For a random variable $W$,
a \emph{tail event} is an event $A$ such that for some $w\in \R$, $W\ge w$ on $A$ and $W\le w$ on $A^c$. In an atomless probability space, a tail event with any given probability $\lambda \in(0,1)$ exists, as shown by \cite{WZ21}.
Let $W=u_1(X)/u_2(X)$.
For $\lambda \in (0,1)$,
let $(A^\lambda)_{\lambda\in (0,1)}$ be an increasing family of tail events of $W$ such that $\p(A^\lambda)=\lambda$.
We can check that the mapping $\lambda \mapsto \E^Q[X\id_{A^\lambda}]$ is continuous (because $\lambda\mapsto Q(A^\lambda)$ is continuous) and its range is the open interval $(0,\E^Q[X])$.
Therefore, there exists $\lambda^*\in (0,1)$ such that $\E^Q[X\id_{A^\lambda}]=\E^Q[\xi_1]\in(0,\E^Q[X])$.
Write $A_1=A^\lambda$ and $A_2=(A^\lambda)^c$.
By definition of the tail event,  for some $w^*\ge 0$, we have  $W\ge w^*$ on $A_1$ and $W\le w^*$ on $A_2$.
Since  $W>0$ almost surely and $\E^Q[X\id_{A_1}]<\E^Q[X]$, thus $\p(A_1)<1$, we know $w^*>0$.

We will show that $(X_1,X_2,Q)=(X\id_{A _1}, X\id_{A_2},Q)$ is a competitive equilibrium.
The budget condition is satisfied by $\E^Q[X\id_{A_1}]=\E^Q[\xi_1]$ and
$\E^Q[X\id_{A_2}]=\E^Q[X]-\E^Q[\xi_1]=\E^Q[\xi_2]$.
Market clearance is immediate.
It remains to show individual optimality. 
  Denote by $z=\E[u_1(X)/X]$. For any $Y$ with $0\le Y\le X$ such that $\E^Q[Y]\le \E^Q[\xi_1]=\E^Q[X\id_{A_1}]$, using that $x\mapsto u_1(x)/x$ is increasing, we have
\begin{align*}
 \E\left[Y\frac{u_1(Y)}{Y}\right]
\le\E\left[Y\frac{u_1(X)}{X}\right]=z\E^Q[Y]\le z\E^Q[\xi_1]=z\E^Q[X\id_{A_1}]=\E[u_1(X_1)].
\end{align*}
Hence, $\E[u_1(Y)]\le \E[u_1(X_1)]$.
For any $Y$  with $0\le Y\le X$  such that $\E^Q[Y]\le \E^Q[\xi_2]=\E^Q[X\id_{A_2}]$,
we have
\begin{align*}
\E[u_2(Y)] \le\E\left[Y\frac{u_2(X)}{X}\right]&\le \E\left[Y\frac{u_1(X)}{w^* X}\id_{A_1}\right]+\E\left[Y\frac{u_2(X)}{X}\id_{A_2}\right]
\\&=\frac{z}{w^*}\E^Q[Y\id_{A_1}]+\E\left[Y\frac{u_2(X)}{X}\id_{A_2}\right].
\end{align*}
Moreover,  $\E^Q[Y]\le \E^Q[X\id_{A_2}]$ implies
 $\E^Q[Y\id_{A_1}]\le\E^Q[(X-Y)\id_{A_2}]$.
Hence,
\begin{align*}
\E[u_2(Y)]&\le
\frac{z}{w^*}\E^Q[(X-Y)\id_{A_2}] +\E\left[Y\frac{u_2(X)}{X}\id_{A_2}\right]\\
&\le \E\left[\frac{X-Y}{w ^* X} u_1(X)\id_{A_2}\right]+\E\left[Y\frac{u_2(X)}{X}\id_{A_2}\right]\\
&\le \E\left[\frac{X-Y}{X}u_2(X)\id_{A_2}\right] +\E\left[Y\frac{u_2(X)}{X}\id_{A_2}\right]=\E[u_2(X\id_{A_2})].
\end{align*}
Hence, $\E[u_2(Y)]\le \E[u_2(X_2)]$.
Therefore, individual optimality holds, and $(X_1,X_2,Q)$ is a competitive equilibrium.
\end{proof}
The equilibrium price in \eqref{eq:existence}
has the form of \eqref{eq:general_price} with $(\lambda_1,\lambda_2) =(1,0)$.
Because the positions of agents $1$ and $2$ are symmetric, we immediately get another equilibrium price by replacing $u_1$  in \eqref{eq:general_price} with $u_2$. Therefore, the equilibrium price is generally not unique, unless $u_1=u_2$ (seen in Theorem \ref{thm:CE}).
Different equilibrium prices correspond to different equilibrium allocations, but they can be derived from the same vector of initial endowments. Therefore, for a given set of initial endowments, the competitive equilibrium is generally not unique, which is in contrast to the case of strictly
 risk-averse agents, where a unique competitive equilibrium can often be derived from given initial endowments.
 % Appendix \ref{app:CE} contains more discussions.

The proof techniques for Proposition \ref{prop:two_agents} do not generalize to the case of $n\ge 3$ agents because our construction of the equilibrium (jackpot) allocation $(X_1,X_2)=(X\id_{A_1},X\id_{A_2})$ heavily relies on the ratio $u_1(x)/u_2(x)$. Roughly speaking, we choose $A_1$ as the event where $u_1(X)/u_2(X)$ is large,
and we choose $A_2$ as the event where $u_1(X)/u_2(X)$ is small. This approach is similar to Statement (v) of Theorem \ref{thm:Pareto}, but its generalization to $n\ge 3$ agents is unclear.

\subsection{A general fixed-point approach}

\noindent We outline a general approach under some assumptions, and prove the existence of competitive equilibria in the special case of proportional endowments. First, we make an assumption of
 no-ties in the weighted utility functions.

\begin{myassump}{NT}\label{ass:NT}
For   $i\neq j $,   $ \{x\in \R_+: \lambda_iu_i(x)=\lambda_ju_j(x)\}$ is finite for any $\lambda_i, \lambda_j >0$, and $X$ is continuously distributed.
\end{myassump}

A simple example of utility functions satisfying Assumption \ref{ass:NT} is that
agent $i$ is more risk seeking than agent $i+1$ for $i\in [n-1]$; that is, $u_i=T_{i}\circ u_{i+1}$ for some increasing and strictly convex function $T_{i}$ for  $i \in
 [n-1]$. In this case, for any $\lambda_i,\lambda_j> 0$, $i\ne j$, we have that $\lambda_i u_i$ and $\lambda_j u_j$ cross at most once under Assumption \ref{ass:EURS}.
For instance, this holds for $u_i(x)=x^{\alpha_i}$, $i\in [n]$, with distinct values of $\alpha_i$. For an illustration, see the left panel of Figure \ref{fig:example_UPF_not_simplex}.

For any given $\blambda=(\lambda_1, \dots, \lambda_n) \in \Delta_n$,
define the sets $A^\blambda_i = \{\lambda_iu_i(X)=V_\blambda(X)\}$ for $i\in [n]$ and let $\mathbf J^{\blambda} =(\id_{A^{\blambda}_1}, \dots, \id_{A^{\blambda}_n})$.  For $i\ne j$,  Assumption \ref{ass:NT} implies that  $ \p( \lambda_iu_i(X)=\lambda_ju_j(X))=0$ because $X$ is continuously distributed. Hence,  $\p(A_i^{\blambda}
\cap A_j^{\blambda})=0$ and $\sum_{i=1}^n \id_{A^\blambda_i}=1$ (almost surely).
Thus, $X\mathbf J^{\blambda}$ is a jackpot allocation.

Next, we introduce three useful objects.
Define a probability measure $Q^\blambda$ by
 $$
      \frac{\d Q^\blambda}{ \d \p} = \frac{V_{\boldsymbol{\lambda}}(X)}{X}\frac{1}{\E[V_{\boldsymbol{\lambda}}(X)/X]}   \mbox{ ~~ with the convention $0/0=0$},
$$
 a function $f:\Delta_n \to \Delta_n$ by
$$f(\blambda) =\left(\frac{\E^{Q^\blambda}[X\id_{A^\blambda_1}]}{\E^{Q^\blambda}[X]}, \dots,  \frac{\E^{Q^\blambda}[X\id_{A^\blambda_n}]}{\E^{Q^\blambda}[X]}\right) = \frac{\E^{Q^\blambda}[X\mathbf J^{\blambda}]}{\E^{Q^\blambda}[X]}, ~~~\blambda \in \Delta_n,$$
and a function $g:\Delta_n \to \Delta_n$ by
$$g(\blambda)=\left(\frac{\E^{Q^\blambda}[\xi_1]}{\E^{Q^\blambda}[X]}, \dots, \frac{\E^{Q^\blambda}[\xi_n]}{\E^{Q^\blambda}[X]}\right), ~~~\blambda \in \Delta_n.$$
Note that $\E^{Q^\blambda}[X]>0$ holds under Assumption \ref{ass:EURS}.
Our goal is to find $\blambda \in \Delta_n$ such that $f(\blambda)=g(\blambda)$. The next proposition justifies that finding such $\blambda$ is sufficient for finding a competitive equilibrium.
\begin{proposition}
    \label{prop:existence}
    Under Assumptions \ref{ass:EURS} and \ref{ass:NT}, if  $\blambda  \in \Delta_n$ and
  $f(\blambda)=g(\blambda)$,
     then
     $ (X\mathbf J^{\blambda}, Q^{\blambda})$ is a competitive equilibrium for the vector of initial endowments $(\xi_1, \dots, \xi_n)\in \mathbb A_n(X)$.
\end{proposition}
\begin{proof}%[Proof of Proposition \ref{prop:existence}]
Write $(X_1, \dots, X_n, Q)=(X\mathbf J^{\blambda}, Q^{\blambda})$ and
$(\lambda_1, \dots, \lambda_n)=\blambda$.
    The equality $f(\blambda)=g(\blambda) $ implies $\E^Q[X_i]=\E^Q[\xi_i]$ for $i\in [n]$.
Let $z=\E[V_\blambda(X)/X]$. Fix $i\in [n]$. For any $Y$ with $0\le Y\le X$ such that $\E^Q[Y]\le \E^{Q}[\xi_i]$, we have
\begin{align*}
\E[\lambda_iu_i(Y)]&=\E\left[Y\frac{\lambda_iu_i(Y)}{Y}\right]\\ &\le \E\left[Y\frac{\lambda_iu_i(X)}{X}\right] \le\E\left[Y\frac{V_\blambda(X)}{X}\right]=z\E^Q[Y]\le z\E^Q[\xi_i]=z\E^Q[X_i].
\end{align*}
Moreover, since $A^\blambda_i = \{\lambda_iu_i(X)=V_\blambda(X)\}$, we have $V_\blambda(X)\id_{A^\blambda_i}=\lambda_iu_i(X)\id_{A^\blambda_i}=\lambda_iu_i(X_i)$, and this implies
$$z\E^Q[X_i]=\E\left[X\id_{A^\blambda_i}\frac{V_\blambda(X)}{X}\right]=\E[\lambda_iu_i(X_i)].$$
Hence, $\E[u_i(Y_i)]\le \E[u_i(X_i)]$ and thus $X_i$ satisfies individual optimality for agent $i$. The market clearance condition $\sum_{i=1}^n X_i=X$ holds
true because $\sum_{i=1}^n \id_{A^\blambda_i}=1$. Therefore, $(X_1, \dots, X_n, Q)$ is a competitive equilibrium.
\end{proof}

The remaining task is to find  $\blambda$ with $f(\blambda)=g(\blambda)$. We do not know a general solution to this problem, but in the simplified scenario of proportional endowments, the problem can be solved.
\begin{myassump}{PE}\label{ass:PE}
The initial endowment vector $(\xi_1,\dots,\xi_n)$ is equal to $(\theta_1 X, \dots, \theta_n X)$ for some $(\theta_1, \dots, \theta_n) \in \Delta_n$.
\end{myassump}
\noindent Under Assumption \ref{ass:PE}, we have $g(\blambda)=(\theta_1, \dots, \theta_n)$ for any $\blambda \in \Delta_n$.
In this situation, we can show that  $f(\blambda)=g(\blambda)$ holds, through a technique established by \cite{JR76SM}.

\begin{proposition}
\label{prop:PE}
If Assumptions \ref{ass:EURS}, \ref{ass:NT} and \ref{ass:PE} hold, then there exists a competitive equilibrium of the form
     $ (X\mathbf J^{\blambda}, Q^{\blambda})$ for some $\blambda\in \Delta_n$.
\end{proposition}
\begin{proof}%[Proof of Proposition \ref{prop:PE}]
A face of $\Delta_n$ is the set  $\Delta_n^D=\{(x_1, \dots, x_n)\in \Delta_n:  x_j=0~ \mbox{for}~ j \in D\}$ for some  $D \subseteq [n]$.
Lemma \ref{lem:PE} below guarantees that $f$ is a continuous function that carries each face of $\Delta_n$ into itself.
This condition allows us to apply \citet[Lemma 2.1]{JR76SM}, which implies that $f$ is surjective. Hence, there exists $\blambda\in \Delta_n$ such that  $f(\blambda)=(\theta_1, \dots, \theta_n)=g(\blambda)$. By Proposition \ref{prop:existence}, $(X\mathbf J^{\blambda}, Q^\blambda)$ is a competitive equilibrium.
\end{proof}

\begin{lemma}
\label{lem:PE}
If Assumptions \ref{ass:EURS} and  \ref{ass:NT} hold, then $f$ is a continuous function that carries each face of $\Delta_n$ into itself.
\end{lemma}
\begin{proof}%[Proof of Lemma \ref{lem:PE}]

 % We first show that $f$ is a continuous function that
 % carries each face of $\Delta_n$ into itself.
For $i\in [n]$, define $f_i: \Delta_n \to \R_+$  by  $f_i(\blambda)=\E[V_\blambda(X)\id_{A^\blambda_i}]  $ for $\blambda \in \Delta_n$.
Note that
$$ f_i(\blambda)=\E[V_\blambda(X)\id_{A^\blambda_i}]=\E[\lambda_i u_i(X)\id_{A^\blambda_i}]=\E[V_\blambda (X)/X] {\E^{Q^{\blambda}}[X\id_{A_i^{\blambda}}]}.$$
Let $\boldsymbol{\lambda}^m\to\boldsymbol{\lambda}$ in $\Delta_n$. As $X$ is continuously
distributed and $u_1,\ldots,u_n$ are continuous, Assumption \ref{ass:NT} implies  
$$
   \p(\lambda_j u_j(X)=\lambda_k u_k(X))=0
   \quad\text{for all } j\ne k
$$
whenever $\lambda_j,\lambda_k>0$. Hence the maximizer of
$k\mapsto \lambda_k u_k(X)$ is almost surely unique. 
Indeed, outside the null set of ties, let \(r\) be the unique maximizer of
\(k\mapsto \lambda_k u_k(X)\). Then
\[
   \delta
   :=
   \min_{k\ne r}\{\lambda_r u_r(X)-\lambda_k u_k(X)\}>0.
\]
Since \(\boldsymbol{\lambda}^m\to \boldsymbol{\lambda}\), the same strict inequalities hold for
\(\boldsymbol{\lambda}^m\) in place of \(\lambda\) for all sufficiently large \(m\). Hence
the maximizer is eventually unchanged, which gives 
$$
   \mathbf 1_{A_i^{{\blambda}^m}}\to \mathbf 1_{A_i^{\blambda}}
   \quad\text{almost surely}.
$$
Since $X$ is bounded and $u_i$ is continuous, dominated convergence yields
$$
   \E\!\left[\lambda_i^m u_i(X)\mathbf 1_{A_i^{\blambda^m}}\right]
   \to
   \E\!\left[\lambda_i u_i(X)\mathbf 1_{A_i^\blambda}\right].
$$
Thus $f_i$ is continuous.
%     Let $\blambda=(\lambda_1, \dots, \lambda_n)$ and  $\bzeta=(\zeta_1, \dots, \zeta_n) \in \Delta_n$ be such that
% $\Vert \blambda -\bzeta\Vert :=\sum_{i=1}^n
% \vert\lambda_i-\zeta_i\vert< \epsilon$. As $X$ is continuously distributed and $u_1,\dots,u_n$ are continuous, Assumption \ref{ass:NT} implies
% \begin{align*}
%     p:&=\p(A_i^\blambda
%  \cup A_i^\bzeta)-\p(A_i^\blambda \cap A_i^\bzeta)\\&=\p\left(X \in \bigcup_{x^*: \lambda_iu_i(x^*)=\lambda_j
% u_j(x^*)}\{x^*-c_1\epsilon<x<x^*+c_1\epsilon\}\right)< c_2\epsilon \end{align*}  for some $c_1, c_2 >0$, because switching from $A_i^\blambda$
% to $A_i^\bzeta$ or back can only happen at points in some neighborhoods of $\{x:\lambda_iu_i(x)=\lambda_j u_j(x)\}$.
% This implies that
% $\blambda \mapsto \E[u_i(X)\id_{A_i^{\blambda}}]$ is continuous, further guaranteeing that $f_i$ is continuous.
Therefore,   $\hat f:=\sum_{i=1}^n f_i$ is also a continuous function. Moreover, under Assumption \ref{ass:EURS}, we have $\hat f>0$.
Since $f=(f_1/\hat f, \dots, f_n /\hat f)$, we know that $f$ is continuous. Moreover, if $\lambda_i=0$, then $\p(A^\blambda_i)=0$ and thus $\E^Q[X\id_{A^\blambda_i}]/\E^Q[X]=0$. Hence, for any face $\Delta_n^D$ and $\blambda\in \Delta_n^D$, we have $f(\blambda) \in \Delta_n^D$.
\end{proof}

\section{An equilibrium for mixed agents and binary aggregate endowment}
\label{app:mixed}

\noindent Here, we   consider the general Assumption \ref{ass:EUM}, and for simplicity assume that $X$ only takes two values $a$ and $b$, with $0<a<b$.
 %% "inspired" praises an external source, not one-self.
Based on Theorem
 \ref{thm:CE_only} and Example \ref{ex:RARS}.a,
a candidate allocation $\mathbf X=(X_1,\dots,X_n)$ is given by
\begin{align}
    \label{eq:CE_cand}
X_i &= a_i \id_{\{X=a\}}, ~i\in T
\mbox{~~~and~~~}X_i  = b J_i \id_{\{X=b\}} , ~i\in S,
\end{align}
where $a_i\in \R_+$
satisfies $\sum_{i\in T}a_i=a$ and
$(J_i)_{i\in S}$ is a jackpot vector.
A possible equilibrium price $Q$ is given by
\begin{equation}
    \label{eq:Q_cand}
\frac{\d Q}{\d \p}= \alpha \id_{\{X=a\}} + \beta\id_{\{X=b\}} ~~ \mbox{ for some $\alpha,\beta >0$.}
\end{equation}
The next result shows that
 \eqref{eq:CE_cand}--\eqref{eq:Q_cand} yield the only possible form of competitive equilibria in this setting.

\begin{theorem}[An equilibrium for the mixed case]
\label{thm:CEEUM} Suppose that Assumptions \ref{ass:ER} and \ref{ass:EUM} hold and that $X$ only takes two values $a,b$, with $0<a<b$. If
$(\mathbf X,Q)$ has the form
\eqref{eq:CE_cand}--\eqref{eq:Q_cand} and satisfies\footnote{If
  %%%%%%%%%%%
 $u_i$ is not differentiable,
$u_i'(a_i)$
designates the left derivative
and $u_i'(0)$ the right derivative.}
  %%%%%%%%%%%%
\begin{equation}
\label{eq:cond_CE_EUM}
    \min_{i\in T}\frac{u_i'(a_i)}{u_i'(0)} \ge \frac{\alpha}{\beta}\ge \max_{j\in S} \frac{b u_j (a) }{a u_j (b)},
\end{equation}
then it is a competitive equilibrium.
If
$(\mathbf X,Q)$ is a competitive equilibrium
and $\mathbf X_S$
and $\mathbf X_T $
are nontrivial, then
 \eqref{eq:CE_cand}, \eqref{eq:Q_cand}, and
\eqref{eq:cond_CE_EUM} hold.
\end{theorem}

\begin{proof}

We first show that $(\mathbf X,Q)$ is a competitive equilibrium with the initial endowment vector chosen as $\mathbf X$.
Write $p=\p(X=a)$ and $q=\p(X=b)=1-p$.
Fix an agent $i\in [n]$,
and let $x=\E^Q[X_i].$  The optimization problem for agent $i$ is
\begin{equation}
    \label{eq:CEEUM_RA}  \mbox{maximize}~ \E[u_i(Y)] ~~~~ \mbox{over $0\le Y\le X$} ~~~~\mbox{subject to} ~\E^Q [ Y] \le \E^Q[X_i] =x.
\end{equation}

Suppose $i\in T$. By the strict concavity of $u_i$, Jensen's inequality guarantees that any optimal $Y$ for
 \eqref{eq:CEEUM_RA} takes at most two values, denoted by $y$ and $z$, when $X=a$ and $X=b$, respectively. Noting that the budget
constraint in \eqref{eq:CEEUM_RA} is binding, we have
 $z= (x-y\alpha p) /(\beta q)$ and $y\in [0,a_i]$.
  Taking the derivative of $\E[u_i(Y)]$, which exists for almost every $y$, and using the concavity of $u_i$, we get
\begin{align*}
    \frac{\d }{\d y}\E[u_i(Y)]
     =    \frac{\d }{\d y}( pu_i(y)+qu_i(z))
     = pu_i'(y)-q \frac{\alpha p }{\beta q} u_i'(z)
     \ge p u_i'(a_i) -\frac{\alpha p }{\beta  } u_i'(0)\ge 0,
\end{align*}
where the last inequality follows from the first inequality in
\eqref{eq:cond_CE_EUM}.
 Hence, the maximum of  \eqref{eq:CEEUM_RA} is attained by  $y=a_i$, and
   $Y=X_i$ is optimal for agent $i\in T$.
 From the above argument, the first inequality in
\eqref{eq:cond_CE_EUM} is necessary for $Y=X_i$ to be optimal.

Now, suppose $i\in S$.
By the strict convexity of $u_i$ and the linearity of the constraint, any optimal $Y$ for
 \eqref{eq:CEEUM_RA} takes the values $0$, $a$ and $b$ with probabilities $p_0$, $p_1$ and $p_2$, respectively, where $p_0+p_1+p_2=1$.  Noting that the budget constraint in
\eqref{eq:CEEUM_RA} is binding, we have
$p_1=(x-b\beta p_2)/(a\alpha)$.
With this, the objective in
 \eqref{eq:CEEUM_RA}, that is,
 $ p_1 u_i(a)  + p_2 u_i(b)
 $,
 is linearly increasing in $p_2$ by the second inequality in
\eqref{eq:cond_CE_EUM}.
  Hence, the maximum of
 \eqref{eq:CEEUM_RA} is attained by the largest possible value of $p_2$, and
   $Y=X_i$ is optimal for agent $i\in S$.
 From the above argument, the second inequality in
\eqref{eq:cond_CE_EUM} is also necessary for $Y=X_i$ to be optimal.

We next show that for a competitive equilibrium $(\mathbf X,Q)$ the forms
 \eqref{eq:CE_cand}--\eqref{eq:Q_cand} are necessary. With this established,
 condition \eqref{eq:cond_CE_EUM}  is verified above.
First, $\d Q/ \d \p>0$ because $X$ is positive.
On each of $\{X=a\}$ and $\{X=b\}$, $\d Q/\d\p$ must be a constant due to the constant supply and at least two agents in each group
participate. Hence, the form  \eqref{eq:Q_cand} holds.
Let $X_S= \sum_{i\in S} X_i$ and $X_T =X-X_S$. By
 Theorem \ref{thm:CE_only}, we know that $(X_T,\mathbf X_S)$ is a jackpot allocation.
Note that by the strict concavity of the utility functions for agents in $T$,  they will have constant payoffs on $\{X=a\}$ and on $\{X=b\}$.
As a consequence, and also noting that  $(X_T,X_S)$ is a nontrivial jackpot allocation, we have $X_T =a\id_{\{X=a\}}$ or $X_T =b\id_{\{X=b\}}$.
Moreover, individual optimization implies that the
 risk-averse agents invest in the cheaper one between $\{X=a\}$ and $\{X=b\}$.
If $\beta \le \alpha$, then all
 risk-seeking agents will prefer payoffs on $\{X=b\}$, and thus $X_S=b\id_{\{X=b\}}$, but the
 risk-averse agents also invest in $\{X=b\}$, contradicting $X_T =a\id_{\{X=a\}}$.
Therefore, $\beta >\alpha$, and this implies $X_T =a\id_{\{X=a\}}$.
Now using Theorem \ref{thm:EU} we obtain
 \eqref{eq:CE_cand}.
\end{proof}

 The  equilibrium price in \eqref{eq:Q_cand}--\eqref{eq:cond_CE_EUM}  is typically not unique.
To interpret the equilibrium in Theorem \ref{thm:CEEUM}, the
 risk-seeking agents prefer gambling on $\{X=b\}$  over gambling on $\{X=a\}$ because their utility function is convex and $b>a$, whereas the
 risk-averse agents do not take payoffs on the event $\{X=b\}$ because they are more expensive. These two considerations together require the price ratio $\beta/\alpha$ to be large enough, reflected by the first inequality in
 \eqref{eq:cond_CE_EUM}, to drive away the
 risk-averse agents from $\{X=b\}$, but not too much, reflected by the second inequality in
 \eqref{eq:cond_CE_EUM},  to attract  the
 risk-seeking agents. Among their own groups,
 risk-seeking agents gamble and
 risk-averse agents proportionally share, fitting with intuition.

\begin{example}
    \label{ex:RARS3} In the setting of
 Example \ref{ex:RARS}.a, assume that $X$  takes the values $a=1/2$ and $b=3/2$, each with probability $1/2$.
 By Theorem \ref{thm:CEEUM},  $(X_1,\dots,X_n)$  in
\eqref{eq:RARS_a}
 is an equilibrium allocation with equilibrium price $Q$ given by
 $\d Q/\d \p= (1-\epsilon )\id_{\{X=1/2\}} +  (1+\epsilon)\id_{\{X=3/2\}} $ for any $\epsilon \in [1/9,1/8]$. Here,  $\d Q/\d \p$ is close to $1$ and slightly larger on $\{X=b\}$.
\end{example}

In the setting of Theorem \ref{thm:CEEUM}, because $a/n \le a_i$ and $ {u_i'(a/n)} \ge {u_i'(a_i)}
$ for some $i\in T$, the condition
\begin{equation}
    \label{eq:cond-CE-2}\max_{i\in T}  \frac{u_i'(a/n)}{u_i'(0)} \ge  \max_{j\in S}  \frac{b u_j (a) }{a u_j (b)}
\end{equation}
is necessary for  \eqref{eq:cond_CE_EUM}, and thus also for  an equilibrium allocation $\mathbf X$ with nontrivial $\mathbf X_S,\mathbf X_T$ to exist.
In the limiting case $b \downarrow a$,  \eqref{eq:cond-CE-2} fails because the right-hand side converges to $1$ and the left-hand side is less than $1$. Hence, with a constant aggregate endowment $X$, there is no competitive equilibrium $(\mathbf X,Q)$ with nontrivial $\mathbf X_S,\mathbf X_T$.

\section{Extensions of Theorem \ref{thm:RDU}}
\label{app:extensions}

\noindent We briefly discuss a few dimensions in which the statements in Theorem \ref{thm:RDU} can readily be generalized. We did not pursue these generalizations because they do not seem to offer stronger empirical relevance than Assumption \ref{ass:RDU}.

\begin{enumerate}[(a)]
    \item Assumption \ref{ass:RDU}
    allows for $w$ to be concave. In this case, $\overline{w}=w$, and the agents are risk seeking for payoffs valued in  $[0,x_0]$.

\item By inspecting the proof of Theorem \ref{thm:RDU}, it suffices to require $w=\overline{w}$ on $[0,1/n]$, and whether $w$ is concave or convex beyond $1/n$ is irrelevant.
    \item The result remains true if $u$ is convex on $[0,x_0]$ instead of being linear, following the same proof, by noting that an agent with a convex utility function and the probability weighting function $\overline w$ is risk seeking, which is the main step to apply Theorem \ref{thm:CT_improvement}.
    \item
A careful inspection of the proofs of main results reveals that, for most of our results on
 risk-seeking EU agents, it suffices to assume that $x\mapsto u(x)/x$  is increasing instead of the convexity of $u$ (this condition is weaker than convexity with $u(0)=0$). Moreover, for the RDU agents in Assumption \ref{ass:RDU}, we can use this condition on $[0,x_0]$ instead of linearity, and the results in Theorem \ref{thm:RDU} hold true.
\end{enumerate}
We formally prove the assertion in (d) below.
Let $u$ be an increasing function with $u(0)=0$
and $x\mapsto u(x)/x$  is increasing.
We first show an analogue of Theorem \ref{thm:CT_improvement}.
Let $(X_1,\dots,X_n)$ and $(Y_1,\dots,Y_n)$ be as in the proof of Theorem \ref{thm:CT_improvement}.
Note that
 $$
 u(Y_i) =u\left(X \id_{\{Z_{i-1}\le U  < Z_i\}}\right)  =u (X)
\id_{\{Z_{i-1}\le U  < Z_i\}}   \ge
\frac{Xu(X_i)}{X_i}\id_{\{Z_{i-1}\le U  < Z_i\}} \id_{\{X_i>0\}}.
 $$
 Hence,
\begin{align*}
 \E\left[u(Y_i) \mid X_1,\dots,X_n\right]
& \ge \frac{Xu(X_i)}{X_i}\id_{\{X_i>0\}}
 \E\left[ \id_{\{Z_{i-1}\le U  < Z_i\}}   \mid X_1,\dots,X_n \right] \\
&  = \frac{Xu(X_i)}{X_i}\id_{\{X_i>0\}}   \frac{X_i}{X}\\
&  =  u(X_i )\id_{\{X_i>0\}}=  u(X_i ).
\end{align*}
This shows $u(Y_i)\ge_{\rm icx} u(X_i)$, where $\ge_{\rm icx}$
is increasing convex order (meaning $\E[\phi(Y_i)]\ge \E[\phi(X_i)]$ for all increasing convex $\phi$).
This implies $\rho_{\overline{w}}(u(Y_i))\ge \rho_{\overline{w}}(u(X_i))$ because $\rho_{\overline{w}}$ is increasing in convex order (e.g., \citealp[Theorem 3]{WWW20bSM}),
and increasing convex order can be decomposed into convex order and
 first-order stochastic dominance (e.g., \citealp[Theorem 4.A.6]{SS07SM}).
Therefore, the jackpot allocation $(Y_1,\dots,Y_n)$ dominates $(X_1,\dots,X_n)$, and for
 sum-optimality it suffices to consider jackpot allocations.
 The rest of the proof follows the same arguments in the proof of Statements (i) and (ii) of Theorem \ref{thm:RDU} with $X$ replaced by $u(X)$. Statements (iii) and (iv) do not rely on the properties of $u$ on $[0,x_0]$.

We finally note that if $u$ is convex on $[0,a]$ and concave on $[a,\infty)$, then  $x\mapsto u(x)/x$ is increasing on some interval $[0,b]$ with $b\ge a$ (often $b>a$).

\section{Upper bounds on outcomes and scapegoat allocations under expected utility}
\label{onl.app.:scapegoats}

\noindent This online appendix gives results for outcome sets that are bounded above and not below, in which case scapegoat allocations result for
 risk-seeking agents. The proofs of these results follow by symmetry from their counterparts.

\begin{theorem}[Counter-monotonic improvement; dual to Theorem \ref{thm:CT_improvement}]\label{thm:dualCT_improvement}
Assume that $X_1,\dots,X_n\in L^1$ are nonpositive, $X=\sum_{i=1}^n X_i$, and Assumption \ref{ass:ER} holds.
Then there exists $(Y_1,\dots,Y_n)\in \mathbb A_n(X)$ such that
\begin{enumerate}[(i)]
\item $(Y_1,\dots,Y_n)$ is counter-monotonic;
\item $Y_i\ge_{\rm cx} X_i$ for all $ i \in [n]$;
\item $ Y_1,\dots,Y_n $ are nonpositive.
\end{enumerate}
Moreover, $(Y_1,\dots,Y_n)$ can be chosen as a scapegoat allocation of $X$.
\end{theorem}

The results below will use some of the following assumptions.

\begin{myassump}{EU-dual}\label{ass:dualEU}
 Each agent maximizes  EU with a strictly increasing continuous utility function $u_i:\R_- \rightarrow \R_-$ with $u_i(0)=0$.
The domain of allocations is $\X=L^\infty_-$, where $L^\infty_-$ is the set of all nonpositive bounded random variables. The total payoff $X\in \X$ satisfies $\p(X<0)>0$.
\end{myassump}

\begin{myassump}{EURS-dual}\label{ass:dualEURS}
On top of Assumption \ref{ass:dualEU}, all agents are strictly risk seeking.
\end{myassump}

\begin{myassump}{EUM-dual}\label{ass:dualEUM}
On top of Assumption \ref{ass:dualEU},  agents in a subgroup
$S\subseteq[n]$ are strictly risk seeking and the others in  $T=[n]\setminus S$ are strictly risk averse, with $S,T$ nonempty.
\end{myassump}

\begin{myassump}{H-EURS-dual}\label{ass:dualHO}
On top of Assumption \ref{ass:dualEURS}, agents are {homogeneous}; that is, $u_1 = \dots = u_n = u$ holds.
\end{myassump}

The function $V_{\boldsymbol{\lambda}}$ below was defined in (\ref{eq:V}).

\begin{theorem}[Pareto optimality for risk seekers; dual to Theorem \ref{thm:Pareto}]\label{thm:dualPareto}
Suppose that Assumptions \ref{ass:ER} and \ref{ass:dualEURS} hold. For an allocation $\mathbf X=(X_1,\dots, X_n)\in \mathbb A_n(X)$,
the following statements are equivalent:
\begin{enumerate}[(i)]
 \item $\mathbf X$ is Pareto optimal;
 \item $\mathbf X$ is
 $\blambda$-optimal for some $\boldsymbol{\lambda} \in \Delta_{n}$;
\item  $\mathbf X$ satisfies $\sum_{i=1}^n \lambda_iu_i(X_i) = V_{\boldsymbol{\lambda}}(X)$ for some $\boldsymbol{\lambda}=(\lambda_1,\dots,\lambda_n)\in \Delta_{n}$;
 \item  $\mathbf X$ satisfies $\sum_{i=1}^n \lambda_i\E[u_i(X_i)] = \E\left[V_{\boldsymbol{\lambda}}(X)\right]$ for some $\boldsymbol{\lambda}=(\lambda_1,\dots,\lambda_n)\in \Delta_{n}$;
 \item $\mathbf X$ is a scapegoat allocation satisfying the following restriction:  for some $\boldsymbol{\lambda}=(\lambda_1,\dots,\lambda_n)\in \Delta_{n}$,
 $ \lambda_i u_i(X)\id_{A_i}  =V_{\blambda}(X) \id_{A_i}$  for each $i\in [n]$, where we write $\mathbf X=X(\id_{A_1},\dots,\id_{A_n})$.

 \end{enumerate}
 \end{theorem}

\begin{proposition}[Dual to Proposition \ref{prop:PO_jackpot}]
 \label{prop:dualPO_jackpot}
If $X=x<0$ is a constant and Assumptions \ref{ass:ER} and \ref{ass:dualEURS} hold, then all scapegoat allocations of $X$ are Pareto optimal.
\end{proposition}

\begin{proposition}[Dual to Proposition \ref{prop:UPF}]
\label{prop:dualUPF}
 Under Assumptions \ref{ass:ER} and \ref{ass:dualHO},
 $\mathrm{UPF}(X)=\mathrm{UPJ}(X)=\Delta_n(\E[u(X)])$.
\end{proposition}

   \begin{proposition}[Dual to Proposition \ref{prop:convex}]
    \label{prop:dualconvex}
    Under Assumptions \ref{ass:ER} and \ref{ass:dualEU},
    \begin{enumerate}[(i)]
    \item
    both $\mathrm{UPJ}(X)$
    and $\mathrm{UPS}(X)$  are convex;
    \item  an allocation of $X$ is Pareto optimal if and only if
it is $\blambda$-optimal for some $\boldsymbol{\lambda} \in \Delta_{n}$.
    \end{enumerate}
\end{proposition}

Remember that for $\mathbf X=(X_1,\dots,X_n)$ and $S\subseteq[n]$, we write $\mathbf X_S=(X_i)_{i\in S} $.

\begin{theorem}[Subgroups; dual to Theorem \ref{thm:EU}]
\label{thm:dualEU}
Suppose that Assumptions \ref{ass:ER} and \ref{ass:dualEU} hold, and that $\mathbf X \in \mathbb A_n(X)$ is Pareto optimal.
 For any set $S\subseteq[n]$ of strictly
 risk-seeking agents, $\mathbf X_S$ is a scapegoat allocation.
 For any set $T\subseteq[n]$ of strictly
 risk-averse agents,
 $\mathbf X_T$ is a comonotonic allocation.
\end{theorem}

\end{appendix}
\end{document}